\documentclass[11pt, reqno]{amsart}
\usepackage{amsmath,amssymb,amsthm,mathrsfs,amstext,amsfonts}
\usepackage{cancel}
\usepackage{mathtools}
\usepackage{bbold}
\usepackage[english]{babel}
\usepackage{physics}
\usepackage{braket}
\usepackage{dsfont}
\usepackage[table,xcdraw]{xcolor}
\usepackage{multirow}
\usepackage{graphicx}
\usepackage{fontawesome}
\usepackage{stmaryrd}
\usepackage{makecell}
\usepackage{yhmath}
\usepackage[margin=1in]{geometry}
\usepackage[pagebackref, colorlinks = true, linkcolor = blue, urlcolor  = blue, citecolor = red]{hyperref}
\usepackage[capitalise]{cleveref}
\usepackage[title]{appendix}%
\usepackage{textcomp}%
\usepackage{manyfoot}%
\usepackage{booktabs}%
\usepackage{algorithm}%
\usepackage{algorithmicx}%
\usepackage{algpseudocode}%
\usepackage{listings}%

\usepackage{tikz, tikzscale, tikz-cd}
\usetikzlibrary{3d}
\usetikzlibrary{perspective}
\usepackage{pgfplots}
\usetikzlibrary{pgfplots.fillbetween}
\pgfplotsset{compat=1.18}
\usetikzlibrary{matrix}

\DeclareMathOperator{\tmin}{\dot{\otimes}}
\DeclareMathOperator{\tmax}{\hat{\otimes}}
\DeclareMathOperator{\Pos}{PSD_\textit{d}}

\DeclareMathOperator{\Posa}{PSD_{\textit{d}_A}}
\DeclareMathOperator{\Posb}{PSD_{\textit{d}_B}}

\DeclareMathOperator{\one}{\mathds{1}}
\DeclareMathOperator{\id}{id}
\DeclareMathOperator{\vp}{\varphi}

\DeclareMathOperator{\Rnum}{\mathds{R}}
\DeclareMathOperator{\Cnum}{\mathds{C}}

\newcommand{\acone}{V_A^+}
\newcommand{\cone}{V^+}

\newcommand{\aconed}{(V_A^*)^+}
\newcommand{\bcone}{V_B^+}
\newcommand{\bconed}{(V_B^*)^+}
\newcommand{\CS}{CS_{k,g}^1}
\newcommand{\CSmn}{CS_{k,g}}
\newcommand{\CSC}{CS_{k,g}^+}
\newcommand{\CSCd}{(CS_{k,g}^*)^+}

\newcommand{\CStwo}{CS_{2,2}^+}
\newcommand{\CStwodu}{(CS_{2,2}^*)^+}

\newcommand{\red}[1]{\textcolor{red}{#1}}
\newcommand{\blue}[1]{\textcolor{blue}{#1}}
\newcommand{\green}[1]{\textcolor{green}{#1}}

\newcommand{\nat}{\mathds{N}}

\renewcommand{\epsilon}{\varepsilon}

\numberwithin{equation}{section}

\newtheorem{thm}{Theorem}[section]
\newtheorem{lem}[thm]{Lemma}
\newtheorem{cor}[thm]{Corollary}
\newtheorem{ex}[thm]{Example}
\newtheorem{defi}[thm]{Definition}
\newtheorem{prop}[thm]{Proposition}
\newtheorem{remark}[thm]{Remark}

\newtheorem{conjecture}[thm]{Conjecture}

\makeatletter
\newtheorem*{rep@theorem}{\rep@title}
\newcommand{\newreptheorem}[2]{
\newenvironment{rep#1}[1]{
 \def\rep@title{#2 \ref{##1}}
 \begin{rep@theorem}}
 {\end{rep@theorem}}}
\makeatother

\newmuskip\pFqmuskip

\newcommand*\pFq[6][8]{
  \begingroup
  \pFqmuskip=#1mu\relax
  \mathchardef\normalcomma=\mathcode`,
  \mathcode`\,=\string"8000
  \begingroup\lccode`\~=`\,
  \lowercase{\endgroup\let~}\pFqcomma
  {}_{#2}F_{#3}{\left[\genfrac..{0pt}{}{#4}{#5};#6\right]}
  \endgroup
}
\newcommand{\pFqcomma}{{\normalcomma}\mskip\pFqmuskip}
\allowdisplaybreaks

\newcommand{\nocontentsline}[3]{}
\newcommand{\tocless}[2]{\bgroup\let\addcontentsline=\nocontentsline#1{#2}\egroup}

\begin{document}

\title{Factorization of multimeters: a unified view on nonclassical quantum phenomena}

\author{Tim Achenbach}
\email[Tim Achenbach]{tim.n.achenbach@jyu.fi}
\address[Tim Achenbach and Leevi Lepp\"{a}j\"{a}rvi]{Faculty of Information Technology, University of Jyv\"{a}skyl\"{a}, 40100 Jyv\"{a}skyl\"{a}, Finland}
\address[Tim Achenbach]{Naturwissenschaftlich-Technische Fakult\"{a}t, Universit\"{a}t Siegen, Walter-Flex-Stra\ss e 3, 57068 Siegen, Germany}

\author{Andreas Bluhm}
\email[Andreas Bluhm]{andreas.bluhm@univ-grenoble-alpes.fr}
\address[Andreas Bluhm]{Univ.\ Grenoble Alpes, CNRS, Grenoble INP, LIG, 38000 Grenoble, France}

\author{Leevi Lepp\"{a}j\"{a}rvi}
\email[Leevi Lepp\"{a}j\"{a}rvi]{leevi.i.leppajarvi@jyu.fi}

\author{Ion Nechita}
\email[Ion Nechita]{nechita@irsamc.ups-tlse.fr}
\address[Ion Nechita]{Laboratoire de Physique Th\'eorique, Universit\'e de Toulouse, CNRS, UPS, France}

\author{Martin Pl\'{a}vala}
\email[Martin Pl\'{a}vala]{martin.plavala@itp.uni-hannover.de}
\address[Martin Pl\'{a}vala]{Institut f\"{u}r Theoretische Physik, Leibniz Universit\"{a}t Hannover, Hannover, Germany}

\begin{abstract}
Quantum theory exhibits various nonclassical features, such as measurement incompatibility, contextuality, steering, and Bell nonlocality, which distinguish it from classical physics. These phenomena are often studied separately, but they possess deep interconnections. This work introduces a unified mathematical framework based on commuting diagrams that unifies them. By representing collections of measurements (multimeters) as maps to the set of column-stochastic matrices, we show that measurement compatibility and simulability correspond to specific factorizations of these maps through intermediate systems. We apply this framework to put forward connections between different nonclassical notions and provide factorization-based characterizations for steering assemblages and Bell correlations, including a new perspective on the CHSH inequality witnessing measurement incompatibility. We also investigate the symmetric $n$-extensions of multimeters and no-signaling behaviors and connect these extensions to a notion of $n$-wise compatibility and to the existence of $n$-wise LHV models, respectively. Furthermore, we investigate robustness to noise of nonlocal features by examining factorization conditions for maps involving noisy state spaces, providing geometric criteria for when noisy multimeters can be simulated by simpler measurement settings.
\end{abstract}

\maketitle
\vspace{-0.5cm}
\tableofcontents

\section{Introduction}
There are several nonclassical features of quantum theory that are of foundational importance since they illustrate the departure from classical physics and classical probability theory. These nonclassical features can be classified as local, meaning that they involve only preparations, transformations, and measurements involving a single quantum system, and as nonlocal, meaning that they require at least two or more spatially separated quantum systems. Examples of local nonclassical features of quantum theory are the uncertainty relations \cite{heisenberg1927anschaulichen,kennard1927quantenmechanik,robertson1929uncertainty,schrodinger1935gegenwartige}, no-cloning and no-broadcasting theorems \cite{wootters1982single,barnum1996noncommuting,barnum2007generalized}, incompatibility of measurements \cite{Heinosaari2016,guhne2023colloquium}, and contextuality as defined by Kochen and Specker \cite{kochen1990problem,budroni2022kochen} or as defined by Spekkens \cite{spekkens2005contextuality}. Nonlocal nonclassical features of quantum theory are entanglement \cite{guhne2009entanglement}, steering \cite{Uola2020}, Bell nonlocality \cite{Brunner2014}, and nonlocality without inputs in quantum networks \cite{tavakoli2022bell}.

The various nonclassical features also serve as the backbone of its applications. One can trace contextuality as the source of the quantum advantage present in quantum computers \cite{howard2014contextuality} and no-broadcasting is necessary for quantum cryptography. Steering and Bell nonlocality are necessary resources for semi-device-independent and device-independent quantum cryptography \cite{wolf2021quantum}, the latter of which was recently demonstrated experimentally \cite{zhang2022device,nadlinger2022experimental}. In nonlocal phenomena as well as in communication tasks making use of quantum protocols, e.g., quantum random access codes, incompatibility is a necessary condition for observation of the phenomena or usefulness of the protocol.

Some of these concepts gave rise to further generalizations with their own operational interpretations. Spekkens contextuality was generalized to test whether a system is quantum \cite{muller2023testing}. Measurement incompatibility, which asks whether several measurements can be replaced by a single measurement, was generalized to simulability of measurements which concerns the question of whether one can obtain measurement statistics of a given set of measurements by performing a different set of measurement \cite{oszmaniec2017simulating,guerini2017operational,Oszmaniec2019,filippov2018simulability,IoannouSimulability2022,Jones2023,veeren2024semi,Jokinen2024,bluhm2025simulation}. Moreover, many of these concepts are interconnected: Both measurement incompatibility and entanglement are necessary both for steering \cite{Uola2014} and Bell nonlocality \cite{Fine2}, while contextuality is a precondition for entanglement \cite{plavala2024contextuality} and Bell nonlocality \cite{wright2023invertible}. No-broadcasting is equivalent to contextuality \cite{jokinen2024no}, while measurement incompatibility and steering are equivalent to a weaker version of contextuality \cite{tavakoli2020measurement,plavala2022incompatibility}.

In this work we show that most of these concepts are instances of the same mathematical problem: factorizability. We say that a map $f: X \to Y$ factorizes through $Z$ if there are maps $g: X \to Z$ and $h: Z \to Y$ such that
\begin{equation}
    f = h \circ g,
\end{equation}
or, equivalently, if the following diagram commutes:
\begin{equation}
\begin{tikzcd}
X \arrow[rd, "g"] \arrow[rr, "f"] & & Y \\
& Z \arrow[ru, "h"] &
\end{tikzcd}  
\end{equation}
It was previously observed that Spekkens contextuality is equivalent to simplex-embeddability \cite{schmid2021characterization}, which itself is an instance of factorizability. Since existence of a broadcasting map was recently shown to be equivalent to contextuality \cite{jokinen2024no}, we immediately get that this problem reduces to factorizability as well.

Firstly, factorizability allows us to approach the problems of certification and quantification of the nonclassical features of quantum theory within a unified and consistent mathematical framework. Secondly, this allows us to use existing mathematical methods to offer solutions to the main problem in quantum information: finding certificates, the so-called witnesses, that yield computable criteria for the presence of the nonclassical features. Thirdly, this will allow us to find streamlined proofs of known results. The main mathematical tool that we will be using is the reformulation of all underlying concepts within the language of \emph{general probabilistic theories} (GPTs) \cite{lami2018non,muller2021probabilistic,leppajarvi2021measurement,plavala2023general} and convex cones. While GPTs will play an important role in our proofs and constructions, we will now formulate all of our results within the language of quantum theory in order to showcase the power of our approach.

This paper is organized as follows. We start with a short section gathering some of our main results (\Cref{sec:main-results}). In \cref{sec:preliminaries} we introduce the necessary mathematical background on general probabilistic theories, convex cones and their tensor products, and quantum measurements. In \cref{sec:compatibility} we develop our unified framework based on factorization of maps and show how it captures measurement incompatibility and simulability. \cref{sec:steering} applies this framework to quantum steering, providing new characterizations of steering assemblages and their local hidden state models. In \cref{sec:bell} we extend our analysis to Bell nonlocality, demonstrating how factorization conditions naturally describe local hidden variable models and providing a novel perspective on CHSH inequalities. In \cref{sec:extendability}, we explore the symmetric extensions of multimeters and no-signaling behaviors, relating them to a constrained form of compatibility and a restricted local hidden variable (LHV) model, respectively. Finally, in \cref{sec:robustness} we investigate how noise affects these nonclassical features by examining factorization conditions for noisy state spaces and deriving geometric criteria for simulability of noisy measurements.

\section{Main results}\label{sec:main-results}
One of the main ideas in our framework is the interpretation of a collection of POVMs as maps from quantum states to outcome statistics which are arranged as matrices. A \emph{multimeter} $M = \{M_{\cdot|x}\}_{x \in [g]}$ is a collection of POVMs that consists of $g$ POVMs $M_{\cdot|x}$ each with $k$ effects $\{M_{a|x}\}_{a \in [k]}$ and thus $k$ outcomes. We will consider every multimeter as a map
\begin{equation}
\label{POVMs_are_maps}
M:  \varrho \mapsto M(\varrho) = (\Tr[M_{a|x} \, \varrho])_{a \in [k], x \in [g]}, 
\end{equation}
from the set of density matrices $D(\Cnum^d)$ to the set of column-stochastic matrices, denoted by $CS^1_{k,g}$. In particular, we can think of this map as a channel (i.e., a positive map) between two GPT state spaces, namely from the quantum state space $D(\Cnum^d)$ to the state space $CS^1_{k,g}$. Then we can look into the properties of the POVMs just by looking into the properties of this map. In particular, we look into what it means if this map can be factorized into a composition of two different maps. In the special case when $g=1$, i.e., when we have just one POVM, we consider it as map to the probability simplex $S_k$ which is isomorphic to $CS^1_{k,1}$.

Since GPT state spaces define cones, another way to study positive maps between state spaces is to consider them as maps between cones, in which case we can write the map as an element of a particular tensor product of the cones. Thus, since the set of density matrices $D(\Cnum^d)$ generates the cone of $d \times d$ positive semi-definite matrices on $\Cnum^d$, denoted by $\Pos$, and $CS^1_{k,g}$ generates the cone of nonnegative matrices with equal column sums, denoted by $CS^+_{k,g}$, we have that
\begin{equation}
    M \in \Pos^* \tmax CS^+_{k,g}  \simeq \Pos \tmax CS^+_{k,g},
\end{equation}
where $\tmax$ denotes the \emph{maximal tensor product} of the cones, see \cref{appendix:tensor-cones}. Thus, we can also look into the properties of the POVMs by looking into the properties of this tensor. In particular, we look into what it means for this tensor to be separable, that is, to be an element of the \emph{minimal tensor product} $\tmin$.

Our motivating starting point are the different characterizations of measurement incompatibility. In contrast to classical physics, it is in general not possible to have a simultaneous readout of all measurement outcomes in quantum theory. To be more exact, a multimeter $M = \{M_{\cdot|x}\}_{x \in [g]}$  of $g$ POVMs with $k$ outcomes is called \emph{compatible} if there exists a single POVM $C$ with $\Lambda$ outcomes and a conditional probability distribution $p=(p_{\cdot|x, \lambda})$ on $[k]$ such that
\begin{equation}
    M_{a|x} = \sum_{\lambda=1}^\Lambda p_{a|x, \lambda} \, C_\lambda
\end{equation}
for all $a \in [k]$ and $x \in [g]$. Our guiding example, which already appears in \cite{jencova2018incompatible} and from which we take inspiration, is the following known result that one can capture compatibility of multimeters either as a form of generalized separability, or via factorizability: for any multimeter $M =\{M_{\cdot|x}\}_{x \in [g]}$ of $g$ POVMs with $k$ outcomes on $\Cnum^d$ the following are equivalent:
\begin{enumerate}
    \item $M$ consists of compatible measurements,
    \item $M \in \Pos \tmax CS^+_{k,g}$ as a tensor is separable, i.e., $M \in \Pos \tmin CS^+_{k,g}$, where $\tmin$ denotes the minimal tensor product of cones,
    \item $M: D(\Cnum^d) \to CS^1_{k,g}$ as a channel can be factorized as $M = \Phi \circ N$ for some  single POVM $N: D(\Cnum^d) \to S_\Lambda$ and some channel $\Phi: S_\Lambda \to CS^1_{k,g}$ for some $\Lambda \in \nat$, i.e., the following diagram commutes:
\begin{equation}
    \begin{tikzcd}
    D(\Cnum^d) \arrow[rd, "N"] \arrow[rr, "M"] & & CS^1_{k,g} \\
    & S_{\Lambda} \arrow[ru, "\Phi"] &
    \end{tikzcd}
\end{equation}
\end{enumerate}

One can also formulate simulability of measurements in the same fashion: here given a multimeter $M: D(\Cnum^{d_A}) \to CS^1_{k,g}$ on a $d_A$-dimensional Hilbert space we can ask whether $M$ can be simulated (or compressed) to a multimeter on a $d_B$-dimensional Hilbert space. That is, we ask whether there are channels $\Phi_\lambda: D(\Cnum^{d_A}) \to D(\Cnum^{d_B})$ and multimeters $N_\lambda: D(\Cnum^{d_B}) \to CS^1_{k,g}$, where $\lambda$ is a label for classical information that can be used in the simulation scheme as well, such that the following diagram commutes:
\begin{equation}
    \begin{tikzcd}
    D(\Cnum^{d_A}) \arrow[rd, "\Phi"] \arrow[rr, "M"] & & CS^1_{k,g} \\
    & D(\Cnum^{d_B}) \dot{\otimes} S_\Lambda \arrow[ru, "N"] &
    \end{tikzcd}  
\end{equation}
Here we have replaced all of the channels $\Phi_\lambda: D(\Cnum^{d_A}) \to D(\Cnum^{d_B})$ by a single channel $\Phi: D(\Cnum^{d_A}) \to D(\Cnum^{d_B}) \otimes S_\Lambda$, where $S_\Lambda$ is a simplex, that is a classical state space, which stores the classical information $\lambda$, and similarly for $N_\lambda$. That is, $\Phi$ is an \emph{instrument}. Given a multimeter $M: D(\Cnum^{d}) \to CS^1_{k,g}$ we can also ask whether $M$ can be performed using less measurements or outcomes, that is, whether there is a multimeter $N: D(\Cnum^{d}) \to CS^1_{l,r}$ and a classical post-processing channel $\Phi: CS^1_{l,r} \to CS^1_{k,g}$ such that:
\begin{equation}
    \begin{tikzcd}
        D(\Cnum^d) \arrow[rd, "N"] \arrow[rr, "M"] & & CS^1_{k,g} \\
        & CS^1_{l,r} \arrow[ru, "\Phi"] &
    \end{tikzcd}
\end{equation}

Steering is a quantum phenomenon that may occur when one of the two parties sharing a bipartite state performs one out of several possible measurements. That is, in the steering scenario we have a bipartite state $\varrho \in D(\Cnum^{d_A d_B})$ and a multimeter $M: D(\Cnum^{d_A}) \to CS^1_{k,g}$ and we investigate the steering assemblage $\sigma = (M \otimes \id)(\varrho)$. In general $\sigma \in CS^1_{k,g} \tmax D(\Cnum^{d_B})$ and we say that $\sigma$ has local hidden state model if and only if $\sigma$ is separable, $\sigma \in CS^1_{k,g} \tmin D(\Cnum^{d_B})$. Due to the known relation between steering and measurement incompatibility \cite{uola2015one} we can immediately transport all results about incompatibility to steering with the steering assemblage $\sigma$ replacing the multimeter $M$. In quantum theory any steering assemblage can be obtained by measuring a bipartite state \cite{steering_g, steering_lrw}. While the same is known to not hold in GPTs \cite{barnum2013ensemble, stevens2014steering}, we do provide an explicit counterexample.

Similarly to steering, Bell nonlocality may occur when both parties sharing a bipartite state perform one out of several possible measurements. In this case we are working with a bipartite state $\varrho \in D(\Cnum^{d_A d_B})$ and multimeters $M: D(\Cnum^{d_A}) \to CS^1_{k,g}$, $N: D(\Cnum^{d_B}) \to CS^1_{l,r}$ and we investigate the so-called behavior $(M \otimes N)(\varrho)$. In general we have $(M \otimes N)(\varrho) \in CS^1_{k,g} \tmax CS^1_{l,r}$ and $(M \otimes N)(\varrho)$ has local hidden variable model if and only if it is separable, 
$(M \otimes N)(\varrho) \in CS^1_{k,g} \tmin CS^1_{l,r}$, which was recently exploited to find connection between measurement incompatibility and Bell nonlocality \cite{plavala2024all}. We will then use the formalism of multimeters to provide streamlined proofs of known results, such as that block-positive operators do not provide post-quantum behaviors \cite{Barnum_2010}, that in quantum theory any pair of dichotomic measurements violates the CHSH inequality if and only if it is incompatible \cite{Wolf_2009}, and that if one party measures a pair of dichotomic measurements while the other applies an arbitrary finite multimeter, then all Bell inequalities are just post-processings of the CHSH inequality \cite{Pironio_2014}. To prove the latter two results we will use another key idea also observed in \cite[Example 14]{jencova2018incompatible}: the CHSH inequalities can be identified with isomorphisms of a square via isomorphisms between tensors and linear maps.

We also investigate symmetric $n$-extensions of multimeters $M \in \Pos \tmax CS^+_{k,g}$ and no-signaling behaviors $P \in CS^1_{k,g} \tmax CS^1_{k,g}$. These extensions are now tensors in $\Pos \tmax (CS^+_{k,g})^{\tmax n}$ and in $CS^1_{k,g} \tmax (CS^1_{k,g})^{\tmax n}$, respectively, such that by applying a symmetric reduction map $\gamma^\Phi_n:  (CS^+_{k,g})^{\tmax n} \to CS^+_{k,g}$ the extension reduces to the original multimeter $M$ and the original no-signaling behavior $P$, respectively. We find that for both $M$ and $P$ the existence of an $g$-extension is equivalent to the tensor being separable, in which case $M$ is compatible and $P$ has an LHV model. On the other hand, for $n <g$, the existence of an $n$-extension for $M$ is equivalent to every $n$-subset of measurements in $M$ being compatible with the joint measurements satisfying particular no-signaling constraints. Similarly for $P$ having an $n$-extension is equivalent to $P$ having restricted LHV models, i.e., LHV models for all the cases where the second party is restricted to only performing some $n$ measurements. Again these restricted LHV models must satisfy some particular no-signaling constraints.

Finally we investigate the robustness of factorizability to noise: it is known that if a certain amount of white noise is added to the system, incompatible measurements become compatible, steering assemblages start having local hidden state models, and behaviors start having hidden variable models. We inspect factorization conditions for maps mixed with noise to get geometric criteria for when noisy multimeters can be simulated by multimeters with simpler measurement settings.

\section{Preliminaries}\label{sec:preliminaries}
In this section, we will review the formalism of GPTs and their transformations. We will start with a short review of the formalism in \Cref{sec:GPT-formalism} and continue by describing positive maps and channels between such theories in \Cref{sec:maps-between-GPTs}. In \Cref{sec:prelim-meters}, we focus on special classes of channels, namely measurements and instruments. Then, we will concentrate on the GPT of column stochastic matrices in  \Cref{sec:CS-GPT}, which we will need to formalize multimeters as maps from state spaces to column stochastic matrices in \Cref{sec:multimeters-in-GPTs}.

\subsection{The formalism of general probabilistic theories} \label{sec:GPT-formalism}
In this section, we will review a general framework for describing operational theories. For a review of cones and the terminology used to describe them, see \Cref{appendix:cones}. A GPT is described by a triple $(V, V^+, \mathds{1}_V)$, where $V$ is a real finite-dimensional vector space with a proper cone $V^+$, and where $\mathds{1}_V$ is an order unit in the interior of the dual cone $A^+ := (V^+)^\ast \subset V^\ast =: A$. Here, we have written $V^*$ for the dual vector space of $V$, i.e., the vector space of linear functionals on $V$, and $(V^+)^\ast$ for the cone dual to $V^+$. Then the \emph{state space} of the theory is described by the compact convex subset $K \subset V$ defined as
\begin{equation}\label{eq:state-space}
K := \{ v \in V^+  :  \mathds{1}_V(v) = 1\}\, ,
\end{equation}
which also forms a base for $V^+$. On the other hand, if we have a compact convex set $K$, then if we set $A = A(K)$ (the set of affine functionals $f: K \to \Rnum$), $A^+ = A(K)^+$ (the set of positive affine functionals on $K$) and take $\mathds{1}_V$ to be the constant function $\mathds{1}_{K}$ giving value $1$ on ${K}$ (which is an order unit in $A$). Then we can choose $V=A^*$ and $V^+ = (A^+)^*$ and it follows that ${K}$ is isomorphic to a compact convex base of the cone $V^+$ determined by $\one_{K}$ as in \cref{eq:state-space}. Thus, sometimes we might also write $(V({K}),V({K})^+, \one_{K})$ for a GPT which is determined by its state space ${K}$. We refer the reader to \cite{lami2018non, plavala2023general} for more background on GPTs. Before we continue, we will briefly describe what classical and quantum mechanics look like in the formalism of GPTs.

\begin{ex}\label{ex:CM} Any \emph{classical system} is described by the triple $\mathrm{CM}_d := (\Rnum^d,\Rnum^d_+,1_d)$, $d \in \mathds N$, where $\Rnum^d_+$
denotes the set of elements with nonnegative coordinates and $1_d=(1,1,\dots,1)$. The classical state space is the probability simplex with $d$ vertices $S_d$.
\end{ex}

\begin{ex}\label{ex:QM} \emph{Quantum theory} corresponds to the triple $\mathrm{QM}_d:=(\mathcal M(\Cnum)_d^{\mathrm{sa}}, \mathrm{PSD}_d ,\operatorname{Tr})$, $d \in \mathds N$, where  $\mathrm{PSD}_d$ is the cone of $d \times d$ positive semi-definite complex matrices in the real vector space of self-adjoint matrices $\mathcal{M}(\Cnum)_d^{\mathrm{sa}}$. The quantum state space is the set of density matrices $D(\Cnum^d)$, i.e., the elements in $\mathrm{PSD}_d$ with trace equal to one.
\end{ex}

Until now, we have only considered single systems. However, as in quantum mechanics, it makes sense to have systems that are multipartite, where each subsystem is a GPT. Let $(V_A, V_A^+, \one_{V_A})$ and $(V_B, V_B^+, \one_{V_B})$ be two GPTs. To construct a bipartite GPT out of them, we first need a joint vector space: For this, we consider simply the tensor product of the individual vector spaces, i.e., $V_{AB} := V_A \otimes V_B$. Second, we need an order unit, for which we can simply take the tensor product of order units, i.e., $\one_{V_{AB}} := \one_{V_A} \otimes \one_{V_B}$. Thus, it remains to choose a cone for the GPT on the bipartite system. Here, we are free to choose any tensor cone, i.e., any cone $V_{AB}^+$ such that 
\begin{equation}
    V_A^+\tmin V_B^+ \subseteq V_{AB}^+ \subseteq V_A^+\tmax V_B^+ \, .
\end{equation}
We refer the reader to \Cref{appendix:tensor-cones} for a recap of tensor products of cones. In fact, we can also define the \emph{minimal tensor product of state spaces}
\begin{equation}
    K_A \tmin K_B := \operatorname{conv} \{x_A \otimes x_B : x_A \in K_A, \, x_B \in K_B \} 
\end{equation}
and the \emph{maximal tensor product of state spaces}
\begin{equation}
    K_A \tmax K_B := \{y \in V_{AB} : \langle \one_{K_A} \otimes \one_{K_B}, y \rangle = 1 , \, \langle f_A \otimes f_B, y \rangle \geq 0 ~ \forall f_A \in A(K_A)^+,\,\forall f_B \in A(K_B)^+ \} \, . 
\end{equation}
Here, we have written for $x \in V$ an element in a vector space and $\alpha \in V^\ast$ a functional the expression $\langle \alpha, x \rangle$ to mean $\alpha(x)$. As $D(\mathds C^{d_A}) \tmin D(\mathds C^{d_B})$ is the set of separable states, elements in $K_A \tmin K_B$ are called \emph{separable}. However, the set $D(\mathds C^{d_A}) \tmax D(\mathds C^{d_B})$ encompasses more than $D(\mathds C^{d_A d_B})$. The set is in fact equal to the set of block-positive matrices of unit trace and can therefore be interpreted as the set of entanglement witnesses \cite{chruscinski2014entanglement}. This can be seen from the fact that $\acone \tmax \bcone = (\aconed \tmin \bconed)^*$. With these definitions, it follows that
\begin{align}
    V( K_A \tmin K_B)^+ &= V(K_A)^+ \tmin V(K_B)^+ \, , \\
    V( K_A \tmax K_B)^+ &= V(K_A)^+ \tmax V(K_B)^+ \, .
\end{align}

\subsection{Positive maps and channels between GPTs} \label{sec:maps-between-GPTs} 
Let $(V_A, V_A^+, \mathds{1}_{V_A})$ and $(V_B, V_B^+, \mathds{1}_{V_B})$ be two GPTs, transformations between these two GPTs are described by \emph{channels}, i.e., linear mappings $\Phi: V_A \to V_B$ that are positive, i.e., $\Phi(V_A^+) \subseteq V_B^+$, and normalization-preserving, i.e., $\mathds{1}_{V_B}(\Phi(v)) = \mathds{1}_{V_A}(v)$ for all $v \in V_A^+$. The latter property means that channels map elements of the state space $K_{A}$ of $(V_A, V_A^+, \mathds{1}_{V_A})$ to elements of the state space $K_B$ of $(V_B, V_B^+, \mathds{1}_{V_B})$. Regarding notation, we may also denote the channel $\Phi: V_A \to V_B$ simply as a map $\Phi: {K}_A \to {K}_B$ between the state spaces of the GPTs. If the map is merely positive but not a channel, we will often write $\Phi: V_A^+ \to V_B^+$ to emphasize that it is a map between cones.  For the dual map $\Phi^\ast: V_B^\ast \to V_A^\ast$, it holds that $\Phi$ is positive if and only if $\Phi^\ast$ is. Additionally, $\Phi$ is a channel if and only if $\Phi^\ast(\one_{V_B})=\one_{V_A}$, i.e., $\Phi^\ast$ is \emph{unital}.

Now, we can consider tensors that play the role of Choi matrices for GPTs. Let us consider two GPTs $(V(K_A), V(K_A)^+, \mathds 1_{K_A})$ and $(V(K_B), V(K_B)^+, \mathds 1_{K_B})$. Let $\{e_i\}_{i\in [d]}$ be a basis of $V(K_A)$ and let $\{a_j\}_{j\in [d]}$ be the corresponding dual basis of $A(K_A)$ for some $d \in \mathds N$. Then, we can define a special tensor $\chi_{V(K_A)} \in A(K_A) \otimes V(K_A)$ as
\begin{equation}  \label{eq:max-ent-tensor}
    \chi_{V(K_A)} = \sum_{i=1}^d a_i \otimes e_i \, .
\end{equation}
It holds that $\chi_{V(K_A)} \in A(K_A) \tmax V(K_A)$, an explicit proof can be found in \cite[Lemma A.1]{jencova2018incompatible}. Also by \cite[Lemma A.1]{jencova2018incompatible} (see also \Cref{appendix:positive-maps}), there is a one-to-one correspondence between $\xi_\Phi \in V(K_B)^+ \tmax V(K_A)^+$ and $\Phi: A(K_A)^+ \to V(K_B)^+$ as $\xi_\Phi = (\Phi \otimes \id)(\chi_{V(K_A)})$ and
\begin{equation}\label{eq:tensor-to-map}
    \langle \Phi(f_A), f_B \rangle = \langle \xi_\Phi, f_B \otimes f_A \rangle, \qquad \forall f_A \in A(K_A), \forall f_B \in A(K_B).
\end{equation}
In fact, the correspondence between maps and tensors is unique:
\begin{lem} \label{lem:xi-to-phi}
    Let $\xi \in K_A \tmax K_B$. Then, there is a  unique map $\Phi: A(K_A)^+ \to V(K_B)^+$ such that $\xi = (\Phi \otimes \id)(\chi_{V(K_A)})$ and $\Phi(\mathds 1_{K_A}) \in K_B$.
\end{lem}

As $\Pos$ is self-dual, one can embed $D(\mathds C^d)$ into both $V(D(\mathds C^d))^+$ and $A(D(\mathds C^d))^+$. However, GPTs are in general not self-dual and one cannot embed $K_A$ into $A(K_A)^+$. Sometimes, it is nonetheless helpful to also define a state space in $A(K_A)^+$. Let $\varrho \in K_A$ and define
\begin{equation}\label{eq:dual-state-space}
    (K_A)^\ast_{\varrho}:=\{f \in A(K_A)^+ : f(\varrho) = 1\} \, .
\end{equation}
If $ \varrho$ is in the relative interior of $K$, then $(K_A)^\ast_{\varrho}$ is compact and $V((K_A)^\ast_{\varrho})^+= A(K_A)^+$, $A((K_A)^\ast_{\varrho})^+= V(K_A)^+$, see again \cite[Lemma A.1]{jencova2018incompatible} for an explicit proof. Thus, $(K_A)^\ast_{\varrho}$ is a state space for the GPT $(A(K_A), A(K_A)^+, \varrho)$. If $\varrho$ is not in the relative interior, there is $f \in A(K_A)^+$ such that $f \neq 0$ and $f(\varrho) = 0$, and $(K_A)_\varrho^\ast$ is not bounded.

To illustrate this construction, let us again consider classical and quantum mechanics:
\begin{ex} \label{ex:dual-state-space-CM}
    Let us consider classical system as a GPT $\mathrm{CM}_d = (\Rnum^d,\Rnum^d_+,1_d)$, $d \in \mathds N$ with some classical state $q >0$ (i.e., a probability distribution with positive definite entries). Then, we can readily verify that 
    \begin{equation}
        (S_d)_q^\ast = \{ p \in \Rnum^d_+ : \sum_{i=1}^d p_i q_i = 1\} \, .
    \end{equation}
    Thus, we can identify $(S_d)_q^\ast \simeq S_d$, where the vertices of the new simplex are 
    \begin{equation}
        \delta_i = (0, \ldots, 0, \underbrace{q_i^{-1}}_{i-\text{th~entry}}, 0, \ldots, 0) \, .
    \end{equation}
\end{ex}

\begin{ex} \label{ex:dual-state-space-QM}
    Let us consider quantum theory as a GPT $\mathrm{QM}_d=(\mathcal M(\Cnum)_d^{\mathrm{sa}}, \mathrm{PSD}_d ,\operatorname{Tr})$, $d \in \mathds N$ with some quantum state $\varrho >0$. Then, we can readily verify that 
    \begin{equation}
        (D(\Cnum^d))_\varrho^\ast = \{ \varrho^{-\frac{1}{2}} \sigma \varrho^{-\frac{1}{2}} : \sigma \in D(\Cnum^d)\} \, .
    \end{equation}
\end{ex}

To conclude this subsection, we will characterize maps corresponding to separable tensors (also called \emph{entanglement-breaking} in analogy to the case of quantum mechanics) as having a certain factorization:
\begin{prop} \label{prop:min-factoring-and-tensor}
    Let $\xi_\Phi \in K_A \tmax K_B$ and $\Phi: A(K_A)^+ \to V(K_B)^+$ such that $\Phi(\one_{K_A}) \in K_B$ is the unique associated map in \Cref{lem:xi-to-phi}. Then, $\xi_\Phi \in K_A \tmin K_B$ if and only there exist finitely many outcomes $\Lambda$, a positive map $\Psi_1: A(K_A)^+ \to S_\Lambda^+$ such that $\Psi_1(\one_{K_A}) \in S_{\Lambda}$, and a channel $\Psi_2: S_\Lambda \to K_B$ such that the following diagram commutes:
    \begin{equation}
        \begin{tikzcd}
        A(K_A)^+ \arrow[rd, "\Psi_1"] \arrow[rr, "\Phi"] & & V(K_B)^+ \\
        & S_{\Lambda}^+ \arrow[ru, "\Psi_2"] &
        \end{tikzcd}  
    \end{equation}
If in addition $\Phi$ is a channel, then $\Psi_1$ can be taken to be a channel as well.
\end{prop}
\begin{proof}
    The factorization into positive $\Psi_1$ and $\Psi_2$ follows from \cite[Lemma 3.7]{bluhm2022incompatibility} and \cite[Proposition A.1]{jencova2018incompatible}. Thus, it remains to prove that indeed we can assume $\Psi_1(\one_{K_A}) \in S_{\Lambda}$ and $\Psi_2$ to be a channel. Decomposing 
    \begin{equation}
        \xi_\Phi = \sum_{\lambda = 1}^\Lambda p_\lambda x_\lambda \otimes y_\lambda 
    \end{equation}
    with $x_\lambda \in K_A$ and $y_\lambda \in K_B$ for all $\lambda \in [\Lambda]$ and a probability distribution $(p_\lambda)_{\lambda \in [\Lambda]}$, we can explicitly choose
    \begin{align}
        \Psi_1&: \alpha \mapsto \sum_{\lambda = 1}^\Lambda p_\lambda \alpha(x_\lambda) \delta_\lambda,  & \forall \alpha \in A(K_A) \, , \\
        \Psi_2&: \delta_\lambda \mapsto y_\lambda,  & \forall \lambda \in [\Lambda] \, .
    \end{align}
    We can thus see directly that $\Psi_2$ is a channel and that $\Psi_1(\one_{K_A}) \in K_B$, since $\one_{K_A}(x_\lambda)=1$ for all $\lambda \in [\Lambda]$. The fact that $\Psi_1$ can also be chosen to be a channel in case $\Phi$ is a channel follows from \cite[Proposition A.1]{jencova2018incompatible}.
\end{proof}

\subsection{Meters and instruments in GPTs} \label{sec:prelim-meters}
Measurements are a special form of channels between GPTs, as we will see in this section. An \emph{effect} corresponds to the most elementary 'yes-no' -measurement and it is described by an element $e \in A^+$ such that $0 \leq e(\varrho) \leq 1$ for all $\varrho \in {K}$. The set of effects on a state space ${K}$ is denoted by $E({K})$ and it is called the effect space of ${K}$. The number $e(\varrho)$ is interpreted as the probability that the event corresponding to the effect $e$ is detected in a measurement of a system in a state $\varrho$. We say that a nonzero effect $e$ is \emph{indecomposable} if any decomposition of $e$ into a sum of two other nonzero effects $e_1, e_2$ as $e=e_1+e_2$ implies that $e = \alpha_1 e_1 = \alpha_2 e_2$ for some $\alpha_1, \alpha_2 >0$ \cite{KIMURA2010175}. It can be shown that an effect is indecomposable if and only if it lies on an extreme ray of the dual cone $A(K)^+$.

\begin{defi}
Let $(V, V^+, \mathds{1}_V)$ be a GPT with a state space ${K} $. A measurement (or \emph{meter}) on ${K}$ with $k < \infty$ outcomes is described by a channel $f: {K} \to S_k$. 
\end{defi}

Let $f$ be a $k$-outcome measurement on $K$. Then for all $\varrho \in {K}$ we have that $f(\varrho) \in S_k$ such that if  we write $\{\delta_i\}_{i=1}^k$ for the vertices of the simplex $S_k$, then $f(\varrho)$ has a unique convex decomposition $f(\varrho) = \sum_{i=1}^k \alpha^{(f, \varrho)}_i \delta_i$ for some unique  coefficients $ \alpha^{(f, \varrho)}_i \geq 0$ for all $i \in [k]$ such that $\sum_{i=1}^k \alpha^{(f, \varrho)}_i =1$. For all $i \in [k]$ we define a mapping $f_i: {K} \to [0,1]$ by setting $f_i(\varrho) = \alpha^{(f, \varrho)}_i$. For all $i \in [k]$ we see that $f_i$ is clearly well-defined and it is straightforward to see that $f_i$ is affine. Furthermore, each $f_i$ can be uniquely extended to an element (for which we use the same notation) $f_i$ in $A(K)$. It is then clear that $f_i \in E({K})$ for all $i \in [k]$ and that $\sum_{i=1}^k f_i(\varrho) = 1$ for all $\varrho \in {K}$; hence $\sum_{i=1}^k f_i = \mathds{1}_V$. Thus, each $k$-outcome measurement $f$ on ${K}$ is characterized by some $k$ effects $f_1, \ldots, f_k$ that satisfy $\sum_{i=1}^k f_i = \mathds{1}_V$; hence
\begin{equation}
f(\varrho) = \sum_{i=1}^k f_i(\varrho) \delta_i
\end{equation} 
for all $\varrho \in {K}$. On the other hand, given $k$ effects $f_1, \ldots, f_k$ that satisfy $\sum_{i=1}^k f_i = \mathds{1}_V$ one can use the above equation to define a channel from $K$ to $S_k$. The interpretation is that when we measure a system which is in a state $\varrho$ with a measurement/meter $f$, then $f_i(\varrho)$ is the probability that an outcome $i$ is obtained. 

Measurement devices which not only produce a classical measurement outcome but also output a post-measurement state are called \emph{instruments}. Formally, we can define instruments as specific types of channels:
\begin{defi}
Let $(V_A, V_A^+, \mathds{1}_{V_A})$ and $(V_B, V_B^+, \mathds{1}_{V_B})$ be two GPTs with state spaces ${K_A} $ and $K_B$, respectively. An instrument from $K_A$ to $K_B$ with $k$ outcomes is described by a channel $\Phi: K_A \to K_B\dot{\otimes} S_k$.
\end{defi}
Let $\Phi: K_A \to K_B\dot{\otimes} S_k$ be an instrument between state spaces $K_A$ and $K_B$. Thus, we can write it as $\Phi(\varrho) = \sum_{i=1}^k \lambda^{\varrho}_i \xi^{\varrho}_i \otimes \delta_i$ for some $\xi^{\varrho}_i \in K_B$ and $\lambda^{\varrho}_i \geq 0$ for all $i \in [k]$ and $\varrho \in K_A$ such that $\sum_{i=1}^k \lambda^{\varrho}_i =1$. Since $\{\delta_i\}_{i=1}^k$ is a basis of $\Rnum^k$, this decomposition is unique. Thus, for all $i \in [k]$ we can define a map $\Phi_i: V_A \to V_B$ by setting $\Phi_i( \varrho) = \lambda^{\varrho}_i \xi^{\varrho}_i$ for all $\varrho \in K_A$, and it is straightforward to see that each $\Phi_i$ is positive, i.e., $\Phi_i(V_A^+) \subseteq V_B^+$ and that it is normalization-non-increasing, i.e., $\one_{K_B}(\Phi_i(\varrho)) \leq \one_{K_A}(\varrho)$ for all $\varrho \in V_A^+$. We refer to maps satisfying these two properties as \emph{operations}. Furthermore, we see that $\sum_{i=1}^k \Phi_i$ is normalization-preserving, i.e., it is a channel between $K_A$ and $K_B$. Thus, each $k$-outcome instrument $\Phi$ between $K_A$ and $K_B$ is characterized by a collection of operations $\Phi_1, \ldots, \Phi_k$ such that $\sum_{i=1}^k \Phi_i$ is a channel from $K_A$ to $K_B$ which maps
\begin{equation}
\Phi(\varrho) = \sum_{i=1}^k \Phi_i(\varrho) \otimes \delta_i
\end{equation}
for all $\varrho \in K_A$. On the other hand, given a collection of such maps $\Phi_1, \ldots, \Phi_k$ one can construct an instrument $\Phi: K_A\to K_B\dot{\otimes} S_k$ by just setting $\Phi(\varrho) = \sum_{i=1}^k \Phi_i(\varrho) \otimes \delta_i$ for all $\varrho \in K_A$. The operational interpretation of an instrument is the following: when the system, which is in a state $\varrho$, is operated on by an instrument  $\Phi$, then a measurement outcome $i$ is detected with probability $\one_{K_B}(\Phi_i(\varrho))$ after which the system can be found in the conditional output state $\Phi_i(\varrho)/\one_{K_B}(\Phi_i(\varrho))$. Note that we can without loss of generality assume that $\one_{K_B}(\Phi_i(\varrho)) \neq 0$, because otherwise outcome $i$ never occurs and the post-measurement state is irrelevant.

\subsection{State spaces of column stochastic matrices} \label{sec:CS-GPT}
It is now time to consider a GPT that will be essential for the definition of multimeters. Let us consider the GPT $(CS_{k,g},CS^+_{k,g},\one_{CS^1_{k,g}})$, where $CS_{k,g}$ is the real vector space of real $k \times g$ -matrices whose column sums are equal, $CS^+_{k,g}$ is the cone of nonnegative $k \times g$ -matrices whose column sums are equal and $\one_{CS^1_{k,g}}: CS_{k,g} \to \Rnum $ is a functional defined as $\one_{CS^1_{k,g}}(M) := \frac{1}{g} \left\langle  J_{k,g},M\right\rangle$ for all $M \in CS_{k,g}$, where $J_{k,g}$ is the $k\times g$ -matrix of all ones and the inner product is the Hilbert-Schmidt inner product. The state space of this GPT is the set of column stochastic $k\times g$ -matrices (nonnegative $k\times g$ -matrices whose columns sum to one), which we denote by $CS^1_{k,g}$. This GPT is a special case of the polysimplex considered in \cite{jencova2018incompatible}, where the number of vertices of all simplices involved is chosen to be equal for simplicity.

The extreme points of $CS^1_{k,g}$ are matrices whose columns have an entry 1 on all columns on some positions, i.e., matrices of the form 
\begin{equation}
    s_{i_1, \ldots, i_g} = \sum_{j=1}^g E_{i_j,j}
\end{equation}
for all $i_j \in [k]$ and $j \in [g]$, where $E_{i,j}$ is a matrix with entry 1 on the position $(i,j)$ and zeros everywhere else. We can define measurements $m^{(j)}$ which project onto the $j$-th column and have effects $m^{(j)}_i$ defined as
\begin{equation}\label{eq:sim-irr-measurements}
    m^{(j)}_i(s_{i_1, \ldots,i_g}) := \left\langle E_{i,j},  s_{i_1, \ldots,i_g} \right\rangle = \begin{cases}1 & i_j=i \\0 & \mathrm{otherwise}\end{cases}.
\end{equation}
We note that $\sum_{i=1}^k m^{(j)}_i(M) = \one_{CS^1_{k,g}}(M)$ for all $M \in CS_{k,g}$ so that as functionals from $CS_{k,g}$ to $\Rnum$ they are equal; $\sum_{i=1}^k m^{(j)}_i = \one_{CS^1_{k,g}}$ for all $j \in [g]$. Moreover, we define
\begin{equation}
    e^{(j)}_i := E_{i,j} - E_{k,j} \,.
\end{equation}
These definitions give us convenient bases for $CS_{k,g}$ and $CS^*_{k,g}$, respectively. \cite[Lemma 1]{jencova2018incompatible} states that the $m^{(j)}_i$, $i \in [k]$, $j \in [g]$ generate the extreme rays of $(CS^+_{k,g})^*$. Furthermore,
\begin{equation}\label{eq:effect-basis}
    \one_{CS^1_{k,g}}, m^{(1)}_1, \ldots, m^{(1)}_{k-1}, m^{(2)}_{1}, \ldots, m^{(g)}_{k-1}
\end{equation}
form a basis of $CS^*_{k,g}$ and 
\begin{equation}\label{eq:state-basis}
    s_{k, \ldots, k}, e^{(1)}_1, \ldots, e^{(1)}_{k-1}, e^{(2)}_{1}, \ldots, e^{(g)}_{k-1}
\end{equation}
the corresponding dual basis of $CS_{k,g}$.

\subsection{Multimeters as channels}\label{sec:multimeters-in-GPTs}
As we saw earlier in \Cref{sec:prelim-meters}, a measurement or meter is characterized by a collection of effects. Next we focus on measurement devices that can implement the measurement of multiple meters; we call this a \emph{multimeter}.

\begin{defi}
Let $(V, V^+, \mathds{1}_V)$ be a GPT with a state space ${K}$. A multimeter on ${K}$ with $g$ measurements each with $k$ outcomes is described by a channel $f: {K} \to  CS^1_{k,g}$.
\end{defi}

Measurements will be from now on seen as just multimeters with one measurement setting. For the same reason, in everything that follows we will simply talk about multimeters even if we are considering single measurements unless we want to specifically emphasize that we are talking about measurements. 

Let $f$ be a multimeter with $g$ measurements and $k$ outcomes on ${K}$. Then for all $\varrho \in {K}$ we have $f(\varrho) \in CS^1_{k,g}$. Let us define maps $M_{a |x}: {K} \to S_k$ by setting $M_{a | x} = m^{(x)}_{a} \circ f$ for all  $a \in [k]$ and $x \in [g]$, where $m^{(x)}: CS_{k,g} \to S_k$ are the measurements defined in  \cref{eq:sim-irr-measurements} which project onto the $x$-th column of the outcome probability matrix. As compositions of a channel and a measurement, it follows that $M_{a|x}$ are effects such that $\sum_{a=1}^k M_{a|x} = \one_K$ so that the collection $\{M_{a|x}\}_{a =1}^k$ defines a measurement $M_{\cdot|x}$ on $K$ for all $x \in [g]$. 

In the basis of \cref{eq:state-basis} we can write the multimeter $f: K \to CS^1_{k,g}$ now as
\begin{equation}
    f(\varrho) = \one_{K}(\varrho) s_{k, \ldots,k} + \sum_{x=1}^g \sum_{a =1}^{k-1} M_{a|x}(\varrho) e^{(x)}_a 
\end{equation}
for all $\varrho \in K$. One can readily check that indeed $m_a^{(x)}(f(\varrho)) = M_{a|x}(\varrho)$ for all $a \in [k]$, $x \in [g]$ and $\varrho \in K$ as defined before. Conversely, given a collection of $g$ measurements $\{M_{\cdot|x}\}_{x \in [g]}$ on $K$ each with $k$ outcomes one can use the above equation to define a unique channel $f: K \to CS^1_{k,g}$.

Alternatively, we can also represent $f$ more naturally as
\begin{equation}
f(\varrho) = \sum_{x=1}^g \sum_{a=1}^k M_{a|x}(\varrho) E_{a,x} =  (M_{\cdot|1}(\varrho), \cdots , M_{\cdot|g}(\varrho))    
\end{equation}
for all $\varrho \in K$. In particular, then we see that a  $k$-outcome multimeter $f$ on ${K}$ is characterized by $g$ measurements $M_{\cdot | 1}, \ldots, M_{\cdot | g}$ with $k$ outcomes respectively such that $f(\varrho) = (M_{\cdot|1}(\varrho)| \cdots | M_{\cdot|g}(\varrho))$ where the measurements define the columns of the outcome probability matrix. Thus, we may write $f=f_M$ for some collection of $k$-outcome measurements $M=\{M_{\cdot|x}\}_{x \in [g]}$. In this representation, $f$ manifestly maps states to states. For multimeters the interpretation is that when we measure a system which is in a state $\varrho$ with a multimeter $M$ by using a measurement setting $x \in [g]$, i.e., by using a measurement $M_{\cdot|x}$, then $M_{a|x}(\varrho)$ is the probability that an outcome $a \in [k]$ is obtained. 

An important property of measurements is compatibility.

\begin{defi}
A multimeter $M=\{M_{\cdot|x}\}_{x\in[g]}$ of $g$ measurements with $k$ outcomes on $K$ is compatible if there exists a single $l$-outcome measurement $N$ on $K$ and some conditional probability distribution $\nu = ( \nu_{\cdot|b,x})_{b \in [l], x \in [g]}$ on $[k]$ such that
\begin{equation}
    M_{a|x} = \sum_{b=1}^l \nu_{a|b,x} N_b
\end{equation}
for all $a \in [k]$ and $x \in [g]$.
\end{defi}
We refer the reader to \cite{Heinosaari2016,guhne2023colloquium} for an introduction to incompatible measurements in quantum mechanics and other operational theories. In the next section, we will characterize compatibility in terms of the corresponding map $f: K \to CS^1_{k,g}$.

\section{Compatibility and simulation of multimeters}\label{sec:compatibility}

\subsection{Simulating multimeters on one state space by multimeters on another state space}
We continue by generalizing to GPTs the notion of simulability of quantum measurements considered in the recent works \cite{IoannouSimulability2022,Jones2023,Jokinen2024}.
\begin{defi}
    Let $M=\{M_{\cdot|x}\}_{x \in [g]}$ be a multimeter of $g$ measurements with $k$ outcomes on a state space $K_A$. It is \emph{$K_B$-simulable} if there exists a finite number of outcomes $\Lambda$, an instrument $\Phi: K_A \to K_B \dot{\otimes} S_\Lambda$ with operations $\Phi_\lambda: V(K_A)^+ \to V(K_B)^+$, and a multimeter $N = \{N_{\cdot|x,\lambda}\}_{x \in [g], \lambda \in [\Lambda]}$ of $g \cdot \Lambda$  measurements with $k$ outcomes on $K_B$ such that 
\begin{equation}
    M_{a|x} = \sum_{\lambda =1}^\Lambda\Phi^\ast_\lambda(N_{a|x,\lambda})
\end{equation}
for all $a \in [k]$ and $x \in [g]$.
\end{defi}

Here the recipe of simulation is as follows:  a measurement with a label $x$ is obtained first by measuring the input state $\varrho \in K_A$ by the instrument $\Phi$, obtaining an outcome $\lambda$ and resulting the system in a (unnormalized) conditional output state $\Phi_\lambda(\varrho)$, which is then measured by the POVM $N_{\cdot|x,\lambda}$ from which an outcome $a $ is obtained and then reported as the final outcome of the transformed measurement with a label $x$. 

 We will show now that $K_B$-simulability is equivalent to a factorization of the associated multimeter, seen as a channel into column stochastic matrices. 

\begin{thm}\label{thm:compression-as-factorization}
 Let $M=\{M_{\cdot|x}\}_{x \in [g]}$ be a multimeter of $g$ measurements with $k$ outcomes on a state space $K_A$. Then $M$ is $K_B$-simulable if and only if there exists a finite number of outcomes $\Lambda$, an instrument $\Phi: K_A \to K_B \dot{\otimes} S_\Lambda$ and a multimeter $\tilde N: K_B \dot \otimes S_\Lambda \to CS^1_{k,g}$ such that the following diagram commutes:
\begin{equation}
    \begin{tikzcd}
    K_A \arrow[rd, "\Phi"] \arrow[rr, "M"] & & CS^1_{k,g} \\
    & K_B \dot{\otimes} S_\Lambda \arrow[ru, "\tilde N"] &
    \end{tikzcd}  
\end{equation}
\end{thm}

We start by proving the following lemma.
\begin{lem}\label{lemma:1}
    There is a one-to-one correspondence between multimeters $N: K_B \to CS^1_{k, g \cdot \Lambda}$ and channels $\tilde{N}: K_B \dot \otimes S_\Lambda \to CS^1_{k,g}$.
\end{lem}

\begin{proof}
Informally, the multimeter $N$ maps states in $K_B$ to $g\cdot \Lambda$ tuples of probability vectors of size $k$. If we see the range of this map as $\Lambda$-tuples of elements of $CS^1_{k,g}$, the data of the channel $N$ is equivalent to mapping $K_B \dot \otimes S_\Lambda$ to $CS^1_{k,g}$. Next, we make this intuition formal.

First, given the measurements $\{N_{\cdot|x,\lambda}\}_{x \in [g], \lambda \in [\Lambda]}$ of the multimeter $N: K_B \to CS^1_{k, g \cdot \Lambda}$ we can define a linear map $\tilde{N}: V({K_B}) \otimes \Rnum^\Lambda \to CS_{k,g}$ by setting 
\begin{equation}
    \tilde{N}(X) = \one_{K_B \dot\otimes S_\Lambda}(X) s_{k, \ldots, k} + \sum_{\lambda=1}^\Lambda \sum_{x=1}^g \sum_{a=1}^{k-1}  (N_{a|x,\lambda} \otimes b_\lambda)(X) e_{a}^{(x)},
\end{equation}
for all $X = \sum_{\lambda=1}^\Lambda X_\lambda \otimes \delta_\lambda \in V({K_B}) \otimes \Rnum^\Lambda$, where $\{b_\lambda\}_{\lambda \in [\Lambda]}$ is the dual basis of $\{\delta_\lambda\}_{\lambda \in [\Lambda]}$. Clearly now 
\begin{equation}\label{eq:compression-mm-channel}
    \tilde{N}(\varrho \otimes \delta_\lambda) = \one_{K_B}(\varrho)s_{k, \ldots,k} + \sum_{x=1}^g \sum_{a=1}^{k-1} N_{a|x,\lambda}(\varrho) e^{(x)}_a = \sum_{x=1}^g \sum_{a=1}^k N_{a|x,\lambda}(\varrho) E_{a,x}
\end{equation}
for all $\varrho \in V(K_B)^+$ and $\lambda \in [\Lambda]$. Thus, by linearity $\tilde{N}$ maps states to states so that it is a channel.

On the other hand, if we have a channel $\tilde{N}: K_B \dot \otimes S_\Lambda \to CS^1_{k,g}$, then it defines a multimeter $N: K_B \to CS^1_{k,g \cdot \Lambda}$ as follows. Namely, as $\tilde N$ is linear and maps states to states, it can be written as
\begin{equation}
    \tilde N(X) = \one_{K_B \dot \otimes S_\Lambda}(X)s_{k, \ldots, k} + \sum_{x=1}^g \sum_{a=1}^{k-1} A_{a|x}(X) e_{a}^{(x)},
\end{equation}
where $A_{a|x}: V(K_B) \otimes \Rnum^\Lambda \to \Rnum$ are some linear functionals. We see immediately that 
\begin{equation}
     m_a^{(x)}(\tilde N(X)) =  A_{a|x}(X) \quad \forall a\in[k-1], \, x \in [g],
\end{equation}
which is positive for all $X \in V(K_B)^+ \dot \otimes S_\Lambda^+$, because $\tilde N$ is positive and the $m_a^{(x)}$ generate the extreme rays of $(CS^+_{k,g})^*$. Thus, $A_{a|x} \in (V(K_B)^+ \dot \otimes S_\Lambda^+)^*$ for all $a \in [k - 1]$, $x \in [g]$ and also $A_{k|x} := \one_{K_B \dot \otimes S_\Lambda} - \sum_{a=1}^{k-1} A_{a|x} \in (V(K_B)^+ \dot \otimes S_\Lambda^+)^*$ since
\begin{equation}
m_{k}^{(x)}(\tilde N(X)) =  \one_{K_B \dot \otimes S_\Lambda}(X) - \sum_{a=1}^{k-1} A_{a|x}(X).    
\end{equation}
Hence, $\{A_{\cdot|x}\}_{x \in [g]}$ is a set of $g$ measurements with $k$ outcomes. Defining $N_{a|x,\lambda}(\varrho) := A_{a|x}(\varrho \otimes \delta_\lambda)$ for all $\varrho \in K_B$, $a \in [k]$, $x \in [g]$ and $\lambda \in [\Lambda]$, we find for $X = \sum_{\lambda=1}^\Lambda X_\lambda \otimes \delta_\lambda$ that
\begin{equation}\label{eq:compression-mm-channel-2}
    \tilde N(X) = \one_{K_B \dot \otimes S_\Lambda}(X)s_{k, \ldots, k} + \sum_{\lambda=1}^\Lambda \sum_{x=1}^g \sum_{a=1}^{k-1} N_{a|x,\lambda}(X_\lambda) e_{a}^{(x)} \, .
\end{equation}
Clearly $N = \{N_{\cdot|x,\lambda}\}_{x \in [g], \lambda \in [\Lambda]}$ now defines a multimeter $N: K_B \to CS^1_{k,g \cdot \Lambda}$.
\end{proof}

By using the previous lemma we prove \Cref{thm:compression-as-factorization}. 
\begin{proof}[Proof of \Cref{thm:compression-as-factorization}]
Let us now first assume that the multimeter $M$ on $K_A$ is $K_B$-simulable. Thus, there exists some finite number of outcomes $\Lambda$, an instrument $\Phi: K_A \to K_B \dot{\otimes} S_\Lambda$ and a multimeter $N: K_B \to CS^1_{k, g \cdot \Lambda}$  such that $ M_{a|x} = \sum_{\lambda =1}^\Lambda\Phi^\ast_\lambda(N_{a|x,\lambda})$ for all $a \in [k]$ and $x \in [g]$. As was explained in \Cref{lemma:1}, we can use the multimeter $N: K_B \to CS^1_{k, g \cdot \Lambda}$ to define a multimeter $\tilde N: K_B \dot \otimes S_\Lambda \to CS^1_{k,g}$. What remains to show is just that $M = \tilde N \circ \Phi$. First, as we explained in \Cref{sec:multimeters-in-GPTs}, we have that
\begin{equation}
    \mathds 1_{CS^1_{k,g}}(M(\varrho))=\one_{K_A}(\varrho), \qquad m_a^{(x)}(M(\varrho)) = M_{a|x}(\varrho) \quad \forall a \in [k-1],~x\in [g], \, \varrho \in K_A \,.
\end{equation}
By using \cref{eq:compression-mm-channel} it is easy to see that $\one_{CS^1_{k,g}}(\tilde N(\Phi(\varrho)))=\one_{K_A}(\varrho)$. Finally, for all  $a \in [k-1]$, $x\in [g]$, $\varrho \in K_A$,
\begin{align}
    m_a^{(x)}(\tilde N(\Phi(\varrho))) &= \sum_{\lambda=1}^\Lambda \tilde N_{a|x,\lambda}(\Phi_\lambda(\varrho))= \sum_{\lambda=1}^\Lambda (\Phi^*_\lambda(N_{a|x,\lambda}))(\varrho) = \left(\sum_{\lambda=1}^\Lambda \Phi^*_\lambda(\tilde N_{a|x,\lambda}) \right)(\varrho)  = M_{a|x}(\varrho),
\end{align}
which shows that indeed $M= \tilde N \circ \Phi$.

On the other hand, if $M = \tilde N \circ \Phi$ with the channels $\tilde N$ and $\Phi$ described in the diagram, then we can use \Cref{lemma:1} again to see that $M$ is $K_B$-simulable because  by using \cref{eq:compression-mm-channel-2} we have that 
\begin{equation}
       M_{a|x}(\varrho)=m_a^{(x)}(M(\varrho)) = m_a^{(x)}(\tilde N(\Phi(\varrho))) = \sum_{\lambda=1}^\Lambda N_{a|x,\lambda}(\Phi_\lambda(\varrho))= \left(\sum_{\lambda=1}^\Lambda \Phi^*_\lambda(N_{a|x,\lambda}) \right)(\varrho)
\end{equation}
for all $a \in [k]$, $x \in [g]$ and $\varrho \in K_A$.
\end{proof}

\subsection{Classical simulation of multimeters}
In this section, we consider a different notion of simulation of measurements, purely in terms of classical mixing and post-processing of measurements \cite{oszmaniec2017simulating,guerini2017operational, filippov2018simulability}. 

\begin{defi}
   Let $M= \{M_{\cdot|x}\}_{x \in g}$ be a multimeter of $g$ measurements with $k$ outcomes on a state space $K$. We say that $M$ can be \emph{classically simulated} (or is \emph{classically simulable}) with a multimeter $N = \{N_{\cdot|y}\}_{y \in [r]}$ of $r$ measurements with $l$ outcomes on $K$ if there exist conditional probability distributions $\pi = (\pi_{\cdot|x})_{x \in [g]}$ on $[r]$ and $\nu= (\nu_{\cdot|b,x,y})_{b \in [l], x \in [g], y\in [r]}$ on $[k]$ such that 
\begin{equation}\label{eq:cl-simulation}
    M_{a|x} = \sum_{y=1}^r \pi_{y|x} \sum_{b=1}^l \nu_{a|b,x,y} N_{b|y} 
\end{equation}
for all $a \in [k]$ and $x \in [g]$.
\end{defi}

The operational interpretation of classical simulation is the following: given a measurement label $x$, we measure the input state $\varrho \in K$ with the measurement $N_{\cdot | y}$ with probability $\pi_{y|x}$ from which we obtain a measurement outcome $b$. Instead of registering this outcome we report an outcome $a$ with probability $\nu_{a|b,x,y}$ as the final outcome of the simulated measurement with the label $x$. 

 It turns out that we can also express classical simulability as factorization of the multimeter through a $CS$ of different size. For this we need to first characterize the channels between two sets of column stochastic matrices of different size. 

\begin{thm}\label{thm:channel-polysimplices-characterization}
A linear map $\Phi: CS_{l,r} \to CS_{k,g}$ is a channel if and only if there exist conditional probability distributions $\pi = (\pi_{\cdot|x})_{x \in [g]}$ on $[r]$ and $\nu= (\nu_{\cdot|b,x,y})_{b \in [l], x \in [g], y\in [r]}$ on $[k]$ such that 
\begin{equation}\label{eq:cl-sim-channel}
    \Phi(X) = \one_{CS^1_{l,r}}(X) s_{k, \ldots,k} + \sum_{x=1}^g \sum_{a =1}^{k-1} \left[ \sum_{y=1}^r \pi_{y|x} \sum_{b=1}^l \nu_{a|b,x,y} m^{(y)}_b(X)\right] e^{(x)}_a
\end{equation}
for all $X \in CS_{l,r}$.
\end{thm}
\begin{proof}
As was explained in \Cref{sec:multimeters-in-GPTs}, for a state space $K$, a linear map $\Phi: V(K) \to CS_{k,g}$ is a channel, i.e., a multimeter, if and only if there exist $g$ measurements $\{A_{\cdot|x}\}_{x \in [g]}$ with $k$ outcomes such that 
\begin{equation}\label{eq:CS-channel}
\Phi(\varrho) = \one_{K}(\varrho) s_{k, \ldots,k} + \sum_{x=1}^g \sum_{a =1}^{k-1} A_{a|x}(\varrho) e^{(x)}_a  
\end{equation}
for all $\varrho \in V(K)$.

Let us now take $K = CS^1_{l,r}$. Since now $A_{\cdot|x}$ is a measurement on $CS^1_{l,r}$ for each $x \in [g]$ we can write it as
\begin{equation}
A_{\cdot|x}(\varrho) = \sum_{a=1}^{k} A_{a|x}(\varrho) \delta_{a}
\end{equation}
for all $\varrho \in CS^1_{l,r}$, where $A_{1|x}, \ldots, A_{k|x}$ are the effects of $A_{\cdot|x}$. Since $A_{a|x}$ is an effect, we know that in particular $A_{a|x} \in (CS_{l,r}^+)^*$. Recall that the $m_b^{(y)}$ generate the extreme rays of $(CS_{l,r}^+)^*$. Therefore, since $(CS_{l,r}^+)^*$ is a convex cone, it follows that we can write $A_{a|x}$ as 
\begin{equation}
 \forall x \in [g], \, \forall a \in [k] \qquad A_{a|x} = \sum_{y=1}^r \sum_{b=1}^{l} \gamma_{a,b,x,y} m^{(y)}_{b}
\end{equation}
 for some $\gamma_{a,b,x,y} \geq 0$ . Positivity of $\Phi$ imposes the following constraints:
\begin{equation}
\forall x \in [g], \, \forall a \in [k],\, \forall \vec{b} \in [l]^y \qquad A_{a|x}(s_{b_1, \ldots, b_r}) = \sum_{y=1}^r \gamma_{a,b_y,x,y} \geq 0
\end{equation}
and the fact that $\Phi$ maps states to states:
\begin{equation}\label{eq:gamma_normalization}
\forall x \in [g],\, \forall \vec{b} \in [l]^y \qquad \sum_{a=1}^{k} A_{a|x}(s_{b_1, \ldots, b_r}) = \sum_{a=1}^{k} \sum_{y=1}^r  \gamma_{a,b_y,x,y} = 1.
\end{equation}

Now fix some $y_0 \in [r]$ and let us take $b_y=l$ for all $y \neq y_0$ in the above equation. Then 
\begin{equation}
\forall x \in [g],\, \forall b \in [l] \qquad \sum_{a=1}^{k}  \gamma_{a,b,x,y_0} = 1 -  \sum_{a=1}^{k} \sum_{y \in [r]\setminus \{y_0\}} \gamma_{a,l,x,y} .
\end{equation}
As the right hand side is independent of $b$, it follows that we can define
\begin{equation}
    \pi_{y_0|x} := \sum_{a=1}^{k}  \gamma_{a,b,x,y_0}
\end{equation}
for some $b \in [l]$ and for all $x \in [g]$ and $y_0 \in [r]$. Since $y_0 \in [r]$ was arbitrary, it follows from \cref{eq:gamma_normalization}  that $\pi := (\pi_{\cdot|x})_{x \in [g]}$ is a conditional probability distribution on $[r]$. Setting
\begin{equation}
 \forall x \in [g],\, \forall a \in  [k],\, \forall y \in [r],\, \forall b \in [l] \qquad   \nu_{a|b,x,y} := \begin{cases}
     \gamma_{a,b,x,y}/\pi_{y|x}, & \pi_{y|x} \neq 0 \\
     1/k, & \pi_{y|x} = 0,
 \end{cases}
\end{equation}
 it follows that $\nu= (\nu_{\cdot|b,x,y})_{b \in [l], x \in [g], y \in [r]}$ is a conditional probability distribution on $[k]$. Hence, we may rewrite $A_{a|x}$ as
\begin{equation}
    A_{a|x} = \sum_{y=1}^r \sum_{b=1}^{l} \gamma_{a,b,x,y} m^{(y)}_{b} =   \sum_{y=1}^r \pi_{y|x} \sum_{b=1}^l \nu_{a|b,x,y} m^{(y)}_b
\end{equation}
for all $a \in [k]$ and $x \in [g]$ so that by \cref{eq:CS-channel} we have that
\begin{align} 
    \Phi(X) &= \one_{CS^1_{l,r}}(X) s_{k, \ldots,k} + \sum_{x=1}^g \sum_{a =1}^{k-1} \left[ \sum_{y=1}^r \pi_{y|x} \sum_{b=1}^l \nu_{a|b,x,y} m^{(y)}_b(X)\right] e^{(x)}_a 
\end{align}
for all $X \in CS_{l,r}$. For the converse, it is straightforward to verify that any map as in \cref{eq:cl-sim-channel} is a channel.
\end{proof}

We can use the previous result to conclude that every factorization through a state space of column stochastic matrices is given by some simulation scheme.

\begin{cor}\label{cor:simulation-factorization}
A multimeter $M =\{M_{\cdot|x}\}_{x \in [g]}$ of $g$ measurements with $k$ outcomes on a state space $K$ can be classically simulated with a multimeter $N = \{N_{\cdot |y}\}_{y \in [r]}$ of $r$ measurements with $l$ outcomes on $K$ if and only if there exists a channel $\Phi: CS^1_{l,r} \to CS^1_{k,g}$ such that the following diagram commutes:
\begin{equation}
    \begin{tikzcd}
    K \arrow[rd, "N"] \arrow[rr, "M"] & & CS^1_{k,g} \\
    & CS^1_{l,r} \arrow[ru, "\Phi"] &
    \end{tikzcd}  
\end{equation}
\end{cor}
\begin{proof}
    If $M$ can be classically simulated with $N$, then the conditional probability distributions $\pi$ and $\nu$ in \cref{eq:cl-simulation} can be used to define a channel $\Phi: CS^1_{l,r} \to CS^1_{k,g}$ by using \cref{eq:cl-sim-channel} of \Cref{thm:channel-polysimplices-characterization}. Then
    \begin{align}
        \Phi(N(\varrho)) &=  \one_{CS^1_{l,r}}(N(\varrho)) s_{k, \ldots,k} + \sum_{x=1}^g \sum_{a =1}^{k-1} \left[ \sum_{y=1}^r \pi_{y|x} \sum_{b=1}^l \nu_{a|b,x,y} m^{(y)}_b(N(\varrho))\right] e^{(x)}_a \label{eq:cl-sim-factorization-1}\\ 
        &=  \one_K(\varrho) s_{k, \ldots,k} + \sum_{x=1}^g \sum_{a =1}^{k-1} \left[ \sum_{y=1}^r \pi_{y|x} \sum_{b=1}^l \nu_{a|b,x,y} N_{y|b}(\varrho)\right] e^{(x)}_a \label{eq:cl-sim-factorization-2} \\
        &=  \one_K(\varrho) s_{k, \ldots,k} + \sum_{x=1}^g \sum_{a =1}^{k-1} M_{a|x}(\varrho) e^{(x)}_a \label{eq:cl-sim-factorization-3} \\
        &= M(\varrho) \label{eq:cl-sim-factorization-4}
    \end{align}
for all $\varrho \in V(K)$.

On the other hand, if there exists a channel $\Phi: CS^1_{l,r} \to CS^1_{k,g}$ such that $M=\Phi \circ N$, then there exists conditional probability distributions $\pi$ and $\nu$ as in \Cref{thm:channel-polysimplices-characterization} such that \cref{eq:cl-sim-factorization-1}-\cref{eq:cl-sim-factorization-4} hold. In particular, comparing the basis decompositions in \cref{eq:cl-sim-factorization-2} and \cref{eq:cl-sim-factorization-3} one immediately sees that $M_{a|x} = \sum_{y=1}^r \pi_{y|x} \sum_{b=1}^l \nu_{a|b,x,y} N_{b|y} $ for all $a \in [k]$ and $x \in [g]$ so that $M$ can be classically simulated with $N$.
\end{proof}

As a special case of \Cref{cor:simulation-factorization} we get the following known characterization of measurement incompatibility \cite{jencova2018incompatible}.
\begin{cor} \label{cor:sep-com}
Let $M =\{M_{\cdot|x}\}_{x \in [g]}$ be a multimeter of $g$ measurements with $k$ outcomes on a state space $K$. Then, the following are equivalent:
\begin{enumerate}
    \item $(\id \otimes M)(\chi_{V(K)}) \in K_\varrho^\ast \tmin CS^1_{k,g}$ for some $\varrho$ in the relative interior of $K$
    \item $M$ consists of compatible measurements
    \item There exist a finite number of outcomes $\Lambda$, a measurement $N: K \to S_\Lambda$, and a channel $\Phi: S_\Lambda \to CS^1_{k,g}$ such that the following diagram commutes:
\begin{equation}
  \begin{tikzcd}
    K \arrow[rd, "N"] \arrow[rr, "M"] & & CS^1_{k,g} \\
    & S_{\Lambda} \arrow[ru, "\Phi"] &
    \end{tikzcd}  
\end{equation}
\end{enumerate}
\end{cor}
\begin{proof}
    The equivalence between $(1)$ and $(3)$ follows from \Cref{prop:min-factoring-and-tensor}.  The equivalence between $(2)$ and $(3)$ follows directly from \Cref{cor:simulation-factorization} with $r=1$, as the factorization is equivalent by the corollary to
    \begin{equation}
        M_{a|x} = \sum_{\lambda = 1}^\Lambda \nu_{a|\lambda, x} N_{\lambda} \qquad \forall x \in [g], \, \forall a \in [k] \, .
    \end{equation}
\end{proof}

If we add an additional classical system $S_\Lambda$ into the factorization of \Cref{cor:simulation-factorization} then we get the following characterization:

\begin{cor}
    For a multimeter $M =\{M_{\cdot|x}\}_{x \in [g]}$ of $g$ measurements with $k$ outcomes on a state space $K$ there exist channels $N: K \to CS^1_{l,r} \dot \otimes S_\Lambda$ and $\tilde \Phi: CS^1_{l,r} \dot \otimes S_\Lambda \to CS^1_{k,g}$ such that the following diagram commutes
    \begin{equation}
        \begin{tikzcd}
            K \arrow[rd, "N"] \arrow[rr, "M"] & & CS^1_{k,g} \\
            & CS^1_{l,r} \dot \otimes S_\Lambda \arrow[ru, "\tilde \Phi"] &
        \end{tikzcd}
    \end{equation}
    if and only if there exist conditional probability distributions $\pi = (\pi_{\cdot|x,\lambda})_{x \in [g],\lambda\in[\Lambda]}$ on $[r]$ and $\nu= (\nu_{\cdot|b,x,\lambda,y})_{b \in [l], x \in [g], \lambda \in [\Lambda], y\in [r]}$ on $[k]$ and a multimeter $N = \{N_{\cdot, \cdot|y}\}_{y \in [r]}$ consisting of $r$ measurements with $l \cdot \Lambda$ outcomes on $K$ satisfying
    \begin{equation}
        \sum_{b=1}^l N_{b,\lambda|y} = \sum_{b=1}^l N_{b,\lambda|y'}    
    \end{equation}
    for all $\lambda \in [\Lambda]$ and $y,y' \in [r]$ such that 
    \begin{equation}\label{eq:cl-sim-with-cl}
        M_{a|x} = \sum_{\lambda=1}^\Lambda \sum_{y=1}^r \pi_{y|x,\lambda} \sum_{b=1}^l \nu_{a|b,x,y,\lambda} N_{b,\lambda|y} 
    \end{equation}
    for all $a \in [k]$ and $x \in [g]$.
\end{cor}

\begin{proof}
    Let us first assume that the diagram commutes. Per \Cref{lemma:1} we can use $\tilde \Phi: CS^1_{l,r} \dot \otimes S_\Lambda \to CS^1_{k,g}$ to define a channel $ 
 \Phi: CS^1_{l,r} \to CS^1_{k,g \cdot \Lambda}$ so that by \Cref{thm:channel-polysimplices-characterization} we can write it as 
    \begin{equation}\label{eq:class-sim-with-class0}
        \Phi(X) = \one_{CS^1_{l,r}}(X) s_{k, \ldots,k} + \sum_{\lambda=1}^\Lambda \sum_{x=1}^g \sum_{a =1}^{k-1} \left[\sum_{y=1}^r \pi_{y|x, \lambda} \sum_{b=1}^l \nu_{a|b,x,\lambda,y} m^{(y)}_b(X)\right] e^{(x, \lambda)}_a
    \end{equation}
    for all $X \in CS_{l,r}$. By \Cref{thm:compression-as-factorization} together with the previous expression we have that 
    \begin{equation}\label{eq:class-sim-with-class1}
        M_{a|x} = \sum_{\lambda=1}^\Lambda N^*_\lambda(\Phi_{a|x,\lambda}) = \sum_{\lambda=1}^\Lambda \Phi_{a|x,\lambda} \circ N_\lambda = \sum_{\lambda=1}^\Lambda \left[  \sum_{y=1}^r \pi_{y|x, \lambda} \sum_{b=1}^l \nu_{a|b,x,\lambda,y} m^{(y)}_b \right] \circ N_\lambda,
    \end{equation}
    where $N_\lambda: V(K)^+ \to CS^+_{l,r}$ are the operations of the instrument $N$, which we can write as
    \begin{equation}\label{eq:class-sim-with-class2}
        N_\lambda(X) = \one_{CS^1_{l,r}}(N_\lambda(X)) s_{l,\ldots,l} + \sum_{y=1}^r \sum_{b=1}^{l-1} N_{b, \lambda|y}(X) e^{(y)}_b
    \end{equation}
    for all $\lambda \in [\Lambda]$ and $X \in V(K)$. Here, $N_{b,\lambda|y} \in (V(K)^+)^*$ are positive functionals on $K$ for all $b \in [l-1]$, $\lambda \in [\Lambda]$ and $y \in [r]$. Let us define $N_{l,\lambda|y} := \one_{CS^1_{l,r}} \circ N_\lambda - \sum_{b=1}^{l-1} N_{b,\lambda|y} \in (V(K)^+)^*$ for all $\lambda \in [\Lambda]$ and $y \in [r]$. Now we have that 
    \begin{align}
        \forall y \in[r], \forall \lambda \in [\Lambda]&: \quad \sum_{b=1}^l N_{b, \lambda|y} = \one_{CS^1_{l,r}} \circ N_\lambda, \\
         \forall y \in[r]&: \quad \sum_{\lambda=1}^\Lambda \sum_{b=1}^l N_{b, \lambda|y} = \one_{K},
    \end{align}
    where the second equation follows from the fact the the operations $N_\lambda$ form an instrument. Thus, the effects $N_{b, \lambda|y}$ form a multimeter $N: K \to CS^1_{l \cdot \Lambda,r}$ such that $\sum_{b\in[l]} N_{b,\lambda|y} = \sum_{b\in[l]} N_{b,\lambda|y'}$ for all $\lambda \in [\Lambda]$ and $y,y' \in [r]$. It follows from \cref{eq:class-sim-with-class1} and \cref{eq:class-sim-with-class2} that
    \begin{equation}
        M_{a|x} = \sum_{\lambda=1}^\Lambda \sum_{y=1}^r \pi_{y|x,\lambda} \sum_{b=1}^l \nu_{a|b,x,\lambda,y} N_{b,\lambda|y}
    \end{equation}
    for all $a \in [k]$ and $x \in [g]$.

    On the other hand if \cref{eq:cl-sim-with-cl} holds for some conditional probability distributions $\pi$ and $\nu$ and some multimeter $N: K \to CS^1_{l \cdot \Lambda,r}$ such that $\sum_{b\in[l]} N_{b,\lambda|y} = \sum_{b\in[l]} N_{b,\lambda|y'}$ for all $\lambda \in [\Lambda]$ and $y,y' \in [r]$, then we can use \cref{eq:class-sim-with-class0} and \cref{eq:class-sim-with-class2} to define a channel $\Phi: CS^1_{l,r} \to CS^1_{k,l \cdot \Lambda}$ and an instrument $N: K \to CS^1_{l,r} \dot \otimes S_\Lambda$, respectively. By \Cref{lemma:1} we can use $\Phi$ to define a channel $\tilde \Phi: CS^1_{l,r} \dot \otimes S_\Lambda \to CS^1_{k,g}$ so that then by \cref{eq:cl-sim-with-cl} we can see that \cref{eq:class-sim-with-class1} must hold so that by \Cref{thm:compression-as-factorization} the diagram commutes.
\end{proof}

We will next explore classical simulability a bit more with the aim of providing more operational intuition behind \Cref{thm:channel-polysimplices-characterization}. In particular, we will give an operational proof for that result solely based on classical simulation showing that it is no coincidence that a channel between two spaces of column stochastic matrices is characterized by classical simulation. First we need to explore some concepts related to classical simulability.

Naturally every measurement can be used to trivially classically simulate itself. Following \cite{filippov2018simulability} we next look into measurements for which this is the only way to simulate them (up to postprocessing equivalence).

\begin{defi}\label{def:sim-irr}
A measurement $A$ is \emph{simulation irreducible} if for any multimeter $M$ that can be used to simulate $A$ there exists some measurement in $M$ that is postprocessing equivalent with $A$.
\end{defi}

\begin{remark}
Here by postprocessing equivalent measurements we mean the case when two measurements $A$ and $B$ can be postprocessed from each other via some conditional probability distributions $\nu$ and $\mu$ so that $A_i = \sum_j \nu_{i|j} B_j$ and $B_j = \sum_i \mu_{j|i} A_i$. In fact, then it follows that the set of measurements can be partitioned into equivalence classes of measurements and in many cases it is convenient to consider properties of measurements only between different postprocessing equivalence classes. 
\end{remark}

As already the name suggests, it turns out that every measurement can be reduced to a classical simulation of some set of simulation irreducible measurements \cite{filippov2018simulability}. 

\begin{prop}[\cite{filippov2018simulability}]\label{prop:sim-irr-reduction}
For every multimeter $M$ there exists a multimeter consisting only of simulation irreducible measurements which can classically simulate $M$.
\end{prop}

Thus, in terms of understanding classical simulability studying the simulation irreducible measurements is extremely useful. It can be shown that in the postprocessing equivalence class of each simulation irreducible measurement there exists a unique extremal measurement which we take as representative of the equivalence class (and which is the same as the \emph{minimally sufficient representative} in \cite{Kuramochi2015}). In \cite{filippov2018simulability} the following characterization was given:

\begin{prop}[\cite{filippov2018simulability}]\label{prop:sim-irr-extr}
A measurement is simulation irreducible and extremal if and only if it consists of linearly independent indecomposable effects.
\end{prop}

Now we can use the above result to characterize all the extremal simulation irreducible measurements on the state space of column stochastic matrices.

\begin{prop}\label{prop:sim-irr-polysimplices}
The $k$-outcome measurements $m^{(1)}, \ldots, m^{(g)}$  defined in \cref{eq:sim-irr-measurements} are the only extremal simulation irreducible measurements on $CS^1_{k,g}$.
\end{prop}
\begin{proof}
First, let us note that each $m^{(x)}$ consists of linearly independent indecomposable effects so by \Cref{prop:sim-irr-extr} it is simulation irreducible and extremal for all $x \in [g]$. 

On the other hand, let us look into the structure of (other possible) simulation irreducible measurements on $CS^1_{k,g}$. Since simulation irreducible measurements consist of indecomposable effects and since the effects of the measurements $m^{(1)}, \ldots, m^{(g)}$ are exactly all the extremal indecomposable effects, without loss of generality we can represent an arbitrary extremal indecomposable measurement $\tilde{f}$ as having the outcome set $[g] \times [k]$ such that $\tilde{f}_{(x,a)} = \alpha^{(x)}_{a} m^{(x)}_{a}$ for some $\alpha^{(x)}_{a} \in [0,1]$ for all $a \in [k]$ for all $x \in [g]$. Since $\tilde{f}$ and each $m^{(x)}$ are measurements the normalization $\sum_{x \in [g]} \sum_{a \in [k]} \tilde{f}_{(x,a)} = \mathds{1}_{CS^1_{k,g}} = \sum_{a \in [k]} m^{(x)}_{a}$ holds for all $x \in [g]$ so that we have
\begin{align}
0 = \mathds{1}_{CS^1_{k,g}} - \sum_{x =1}^g \sum_{a =1}^{ k} \alpha^{(x)}_{a} m^{(x)}_{a} = \left(1-\sum_{x=1}^g \alpha^{(x)}_{k} \right) \mathds{1}_{CS^1_{k,g}} + \sum_{x =1}^g \sum_{a =1}^{ k-1} \left( \alpha^{(x)}_{k} -\alpha^{(x)}_{a} \right) m^{(x)}_{a}.
\end{align}
Now, since the set $\{\mathds{1}_{CS^1_{k,g}}, m^{(1)}_{1}, \ldots, m^{(1)}_{k_g-1},\ldots, m^{(g)}_{1}, \ldots, m^{(g)}_{k_g-1}\}$ forms a basis of $CS_{k,g}^*$ so that in particular the set is linearly independent, we have that $\alpha^{(x)}_{a} = \alpha^{(x)}_{a'} =: \alpha^{(x)}$ for all $a,a' \in [k]$ for all $x \in [g]$ and that $\sum_{x \in [g]} \alpha^{(x)} =1$. 

Since $\tilde{f}$ is extremal it consist of linearly independent effects. Thus, if we denote $G_{\neq 0} = \{ x \in [g] \, | \, \alpha^{(x)} \neq 0 \}$, which is clearly non-empty since $\tilde{f}$ is a measurement, then $\{\tilde{f}_{(x,a)}\}_{a \in [k], x \in G_{\neq 0}}$ is a set of linearly independent elements. Let us now fix $x \in G_{\neq 0}$. Now we see that
\begin{align}
\frac{1}{\alpha^{(x)}} \sum_{a = 1}^{k} \tilde{f}_{(x,a)} = \frac{1}{\alpha^{(x)}} \sum_{a = 1}^{k} \alpha^{(x)} m^{(x)}_{a} = \mathds{1}_{CS^1_{k,g}} = \sum_{y =1}^g \alpha^{(y)} \mathds{1}_{CS^1_{k,g}} = \sum_{y =1}^g \alpha^{(y)} \sum_{a =1}^{k} m^{(y)}_{a} = \sum_{y \in G_{\neq 0}} \sum_{a=1}^{k} \tilde{f}_{(y,a)}.
\end{align}
By rearranging the terms we have that
\begin{equation}
0 = \alpha^{(x)} \sum_{y \in G_{\neq 0} \setminus \{x\}} \sum_{a=1}^{k}  \tilde{f}_{(y,a)}+ \left(\alpha^{(x)}-1\right) \sum_{a=1}^{k}  \tilde{f}_{(x,a)}.
\end{equation} 
Since $\{\tilde{f}_{(x,a_x)}\}_{a_x \in [k_x], x \in G_{\neq 0}}$ is a set of linearly independent elements and since we have that $\alpha^{(x)}\neq 0$ we must have that $G_{\neq 0} \setminus \{x\} = \emptyset$, i.e., $G_{\neq 0 } = \{x \}$, so that by the fact that $\sum_{y \in [g]} \alpha^{(y)} =1$ we in fact have that $\alpha^{(x)}=1$ and thus the only nonzero effects of $\tilde{f}$ are $\tilde{f}_{(x,a)} = m^{(x)}_{a}$ for all $a \in [k]$. This concludes the proof.
\end{proof}

Together with \Cref{prop:sim-irr-reduction} we can use the above result for a more operational proof of \Cref{thm:channel-polysimplices-characterization}.
\begin{proof}[Alternative proof for \Cref{thm:channel-polysimplices-characterization}]
As we have already argued before, a linear map $\Phi: CS_{l,r} \to CS_{k,g}$ is a channel (or a multimeter) if and only if there exist $g$ measurements $\{A_{\cdot|x}\}_{x \in [g]}$ with $k$ outcomes such that 
\begin{equation}\label{eq:CS-channel-2}
\Phi(\varrho) = \one_{CS^1_{l,r}}(\varrho) s_{k, \ldots,k} + \sum_{x=1}^g \sum_{a =1}^{k-1} A_{a|x}(\varrho) e^{(x)}_a  
\end{equation}
for all $\varrho \in CS_{l,r}$. Since now $A_{\cdot|x}$ is a measurement on $CS^1_{l,r}$ for each $x \in [g]$ we can write it as
\begin{equation}
A_{\cdot|x}(\varrho) = \sum_{a=1}^{k} A_{a|x}(\varrho) \delta_{a}
\end{equation}
for all $\varrho \in CS^1_{l,r}$, where $A_{1|x}, \ldots, A_{k|x}$ are the effects of $A_{\cdot|x}$.

Since $A_{\cdot|x}$ is a measurement on $CS^1_{l,r}$ for each $x \in [g]$, by \Cref{prop:sim-irr-reduction} there exists some multimeter $M^{(x)}$ consisting of simulation irreducible measurements on $CS^1_{l,r}$ which can be used to simulate $A_{\cdot|x}$ for all $x \in [g]$. Since every simulation irreducible measurement is postprocessing equivalent with an extremal simulation irreducible measurement we can take $M^{(x)}$ to contain only extremal simulation irreducible measurements. Furthermore, since $M^{(x)}$ can be used to simulate $A_{\cdot|x}$ also the multimeter $M$ which consists of all of the extremal simulation irreducible measurements on $CS^1_{l,r}$ can be used to simulate $A_{\cdot|x}$ for all $x \in [g]$. By \Cref{prop:sim-irr-polysimplices} the multimeter $M$ thus consists exactly of the measurements $m^{(1)}, \ldots, m^{(r)}$ defined in \cref{eq:sim-irr-measurements}. By the definition of classical simulability there thus exists some conditional probability distributions $\pi = (\pi_{\cdot|x})_{x \in [g]}$ on $[r]$ and $\nu= (\nu_{\cdot|b,x,y})_{b \in [l], x \in [g], y\in [r]}$ on $[k]$ such that 
\begin{equation}
A_{a|x} = \sum_{y=1}^r \pi_{y|x} \sum_{b=1}^{l} \nu_{a|b,x,y} m^{(y)}_{b}
\end{equation}
for all $a \in [k]$ and $x \in [g]$. Then by \cref{eq:CS-channel-2} we have that
\begin{align} 
    \Phi(X) &= \one_{CS^1_{l,r}}(X) s_{k, \ldots,k} + \sum_{x=1}^g \sum_{a =1}^{k-1} \left[ \sum_{y=1}^r \pi_{y|x} \sum_{b=1}^l \nu_{a|b,x,y} m^{(y)}_b(X)\right] e^{(x)}_a
\end{align}
for all $X \in CS_{l,r}$. Again, the converse is straightforward to verify.
\end{proof}

\section{Steering}\label{sec:steering}

\subsection{Steering in quantum mechanics and general probabilistic theories}
One form of quantum nonlocality that is closely related to measurement incompatibility is quantum steering. Therefore, we will now show how a steering assemblage can be understood in the formalism used so far. We will first stick with steering in quantum mechanics and then generalize it to GPTs. Proofs will come at the end of the section, as it is convenient to prove everything in the GPT formalism and to recover quantum mechanics as a special case.

Let us consider a multimeter consisting of $g$ POVMs acting on a $d$-dimensional system with $k$ outcomes each and a bipartite state $\varrho \in D(\Cnum^{d_A d_B})$.
They enable us to construct the following object:
\begin{equation}
\label{assembl_def}
    [(M \otimes \id)(\varrho)]_{a|x} = \Tr_A[(M_{a|x} \otimes \one_{K_B})\varrho] = \sigma_{a|x}. 
\end{equation}
This object is called an \emph{assemblage}. Assemblages in quantum mechanics are defined as sets of positive operators $\sigma := \{\sigma_{a|x}\}_{a \in [k], x \in [g]} \subset \mathrm{PSD}_{d_B}$ with the property that there exists a $\bar\sigma \in D(\mathds C^{d_B})$ such that
\begin{equation}\label{eq:assembl-no-sign}
    \bar \sigma = \sum_{a = 1}^k \sigma_{a|x} \qquad \forall x \in [g] \,.
\end{equation}
In \cref{assembl_def}, we can easily verify that $\bar \sigma = \mathrm{Tr}_A(\varrho)$ and that indeed $\sigma_{a|x}\geq 0$ for all $a\in [k]$, $x \in [g]$ as required. 

The interpretation of an assemblage is that it corresponds to a probabilistic state preparator such that given a classical input $x$ the state $\hat{\sigma}_{a|x} :=\sigma_{a|x}/p_{a|x}$ is prepared with probability $p_{a|x} = \Tr[\sigma_{a|x}]$. The condition of \cref{eq:assembl-no-sign} then means that without knowing the classical output $a$, i.e., when averaging over all the possible states, we can have no knowledge of the classical input $x$. Now \cref{assembl_def} is a particular realization of the assemblage in the two-party setting so that the states are prepared by performing a measurement on one part of a joint state, and in this case the condition of \cref{eq:assembl-no-sign} corresponds to a no-signaling condition so that without having access to the measurement outcome the party with the prepared states cannot know which measurement setting was chosen on the other part of the shared state.

As for multimeters, we can now consider assemblages as tensors:
\begin{lem}
\label{assemblage}
    Given a multimeter $M: D(\mathds C^{d_A}) \rightarrow \mathrm{CS}_{k,g}^1$ and a bipartite state $\varrho \in D(\Cnum^{d_A d_B})$, the assemblage obtained as $\sigma = (M \otimes \id)(\varrho)$ in \cref{assembl_def} are elements of the set $\mathrm{CS}_{k,g}^1 \tmax D(\mathds C^{d_B})$.
\end{lem}

The statement follows from the fact that $\varrho \in D(\Cnum^{d_A d_B}) \subset D(\Cnum^{d_A}) \tmax D(\Cnum^{d_B})$. In fact, it is known that the converse also holds. From the GHJW-theorem \cite{steering_g, steering_lrw}, it follows that for each assemblage, there is a set of measurements and a bipartite state such that the assemblage is obtained as in \cref{assembl_def}.

\begin{lem}
\label{assem_real}
    Given an element $\sigma \in \mathrm{CS}_{k,g}^1 \tmax D(\mathds C^{d_B})$, for $d_A \geq d_B$ there exists a multimeter $M: D(\mathds C^{d_A}) \rightarrow \mathrm{CS}_{k,g}^1$ and a bipartite state $\varrho \in D(\Cnum^{d_A d_B})$ such that $\sigma = (M \otimes \id)(\varrho)$.
\end{lem}

Since the assemblages are in one-to-one correspondence with elements in the maximal tensor product of quantum states and column-stochastic matrices, the question arises what happens if a certain assemblage is separable, i.e., part of the minimal tensor product of these state spaces. This question lies at the core of this work and the answer will follow from our results at the end of this section. 
We recover the result that separable elements are in one-to-one correspondence with assemblages admitting a local hidden state (LHS) model \cite[Theorem 4]{jencova2018incompatible}.

\begin{prop} \label{prop:LHS-separable}
    For an assemblage $\sigma$, it holds that $\sigma \in \mathrm{CS}_{k,g}^1 \tmin D(\mathds C^d)$ if and only if there exists an LHS model that reproduces it. By LHS model we mean that there exists a $\Lambda \in \mathds N$, operators $\{B_\lambda\}_{\lambda\in[\Lambda]} \subset \Pos$ and a conditional probability distribution $p =(p_{\cdot|\lambda, \, x})_{x \in [g], \lambda \in [\Lambda]}$ on $[k]$ such that
    \begin{equation}
        \sigma_{a|x} = \sum_{\lambda =1}^\Lambda p_{a|\lambda,x} \, B_\lambda
    \end{equation}
for all $a \in [k]$ and $x \in [g]$.
\end{prop}
The proposition follows in particular from \Cref{cor:separbility-LHS}. For a recent review of quantum steering, we refer the reader to \cite{Uola2020}.

With this preparation, we can define steering in the GPT setting. Let us choose a GPT $(V(K), V(K)^+, \mathds 1_K)$ for what follows. We start with the definition of an assemblage. 
\begin{defi} \label{def:GPT-assemblage}
    Let $(V(K), V(K)^+, \mathds 1_K)$ be a GPT. An \emph{assemblage} is a set $\{q_{a|x}, \varrho_{a|x}\}_{a \in [k], x \in [g]}$ of states $\varrho_{a|x} \in K$ and conditional probability distributions $(q_{a|x})_{a \in [k]}$ for all $x \in [g]$ together with an average state $\bar \varrho \in K$ such that
\begin{equation}
    \bar \varrho = \sum_{a \in [k]} q_{a|x} \varrho_{a|x} \qquad \forall x \in [g] \, .  
\end{equation}
\end{defi}

As was explained earlier in the quantum setting, the interpretation of an assemblage is that it corresponds to a probabilistic state preparator where given a classical input $x$ a state $\varrho_{a|x}  \in K$ is prepared with probability $q_{a|x}$. For an assemblage $\varrho = \{q_{a|x}, \varrho_{a|x}\}_{a \in [k], x \in [g]}$ we may also write $\varrho = \{\tilde{\varrho}_{a|x}\}_{a \in [k], x \in [g]}$ where $\tilde{\varrho}_{a|x} = q_{a|x} \varrho_{a|x} \in V(K)^+$. Since the state space $K$ forms a base for the cone $V(K)^+$ and $\one(\tau) = 1$ for all $\tau \in K$, we can recover the conditional probabilities $q$ simply as $q_{a|x} = \one(\tilde{\varrho}_{a|x})$. Thus, we may go back and forth between expressing the assemblage as probabilities together with the states and expressing the assemblages just with the unnormalized states. We use both notations interchangeably. 

The definition for assemblages having an LHS model can be straightforwardly generalized:
\begin{defi}\label{def:lhs_model}
    Let $\{q_{a|x}, \varrho_{a|x}\}_{a \in [k], x \in [g]}$ be an assemblage. Then, the assemblage admits an \emph{LHS model} if there exist $\Lambda \in \mathds N$, a set of states $\{\sigma_{\lambda}\}_{\lambda \in [\Lambda]}$, a probability distribution $(p_\lambda)_{\lambda \in [\Lambda]}$ and conditional probability distributions $(\nu_{a|\lambda,x})_{a\in [k]}$ for all $\lambda \in [\Lambda]$, $x \in [g]$, such that
     \begin{equation}
    q_{a|x} \varrho_{a|x} = \sum_{\lambda=1}^\Lambda \nu_{a|\lambda,x} p_{\lambda} \sigma_{\lambda} \qquad \forall a \in [k],~\forall x \in [g] \, .
\end{equation}
\end{defi}

Theorem 4 of \cite{jencova2018incompatible} states that for any $\xi \in CS^1_{k,g} \tmax K$, there is an assemblage $\{q_{a|x}, \varrho_{a|x}\}_{a \in [k], x \in [g]}$ such that 
\begin{equation} \label{eq:assemblages-as-tensors}
    \xi = s_{k, \ldots, k} \otimes \bar \varrho + \sum_{x = 1}^g \sum_{a =1}^{k-1}  q_{a|x} e_a^{(x)} \otimes \varrho_{a|x} \, .
\end{equation}
Conversely, any element of the form in \cref{eq:assemblages-as-tensors} is in $CS^1_{k,g} \tmax K$, such that there is a one-to-one correspondence between assemblages and tensors in $CS^1_{k,g}  \tmax K$. This implies in particular \Cref{assemblage}.

\subsection{Assemblages as measurements on bipartite states}
The considerations at the end of the previous section, however, do not give us a counterpart to \Cref{assem_real}, because \Cref{def:GPT-assemblage} does not make any reference to applying a multimeter to any bipartite state. We will start with a purely mathematical construction to work around this. As discussed in \Cref{sec:maps-between-GPTs}, any $\xi \in CS^1_{k,g}  \tmax K$ can be written as $\xi = (\id \otimes \Phi)(\chi_{CS^\ast_{k,g} }) = (\Phi^\ast \otimes \id)(\chi_{V(K)})$, where $\Phi: (CS^\ast_{k,g})^+ \to V(K)^+$ such that (see \cref{eq:assemblages-as-tensors})
\begin{align}
    \Phi(\one_{CS^1_{k,g} }) &= \bar \varrho  \label{eq:Phi-steering-1} \\
    \Phi(m_a^{(x)}) &= q_{a|x}\varrho_{a|x} \qquad \forall x \in [g], \, \forall a \in [k] \label{eq:Phi-steering-2}\, .
\end{align}
The dual map $\Phi^\ast: A(K)^+ \to CS^+_{k,g} $ is such that for any $\alpha \in A(K)^+$, 
\begin{equation}
    \Phi^\ast: \alpha \mapsto \alpha(\bar \varrho) s_{k, \ldots k} + \sum_{x = 1}^g \sum_{a = 1}^{k-1} q_{a|x} \alpha(\varrho_{a|x}) e_a^{(x)} \, .
\end{equation}

Now, we would like to see $\Phi^\ast$ as a multimeter, i.e., we need a state space in $A(K)^+$ which $\Phi^\ast$ maps to $CS^1_{k,g} $ and hence makes $\Phi^\ast$ a channel. Let us assume that $\bar \varrho$ is in the relative interior of $K$. Then, we consider $K^\ast_{\bar \varrho}$ defined in \cref{eq:dual-state-space}. It follows from the definition of $K^\ast_{\bar \varrho}$ that for all $\alpha \in K^\ast_{\bar \varrho}$ 
\begin{equation}
    \one_{CS_{k,g}^1}(\Phi^\ast(\alpha)) = \alpha(\bar \varrho) = 1 \, ,
\end{equation}
thus indeed $\Phi^\ast: K^\ast_{\bar \varrho} \to CS_{k,g}^1$ is a channel and hence a multimeter.

\begin{remark} \label{rem:bar-rho-relative-interior}
    The assumption that $\bar \varrho$ is in the relative interior of $K$ is no restriction, since otherwise $K$ is in a sense too large. More precisely, if $\bar \varrho$ is not in the relative interior of $K$, then it is in a face of the convex set $K$ \cite[Theorem 18.2]{Rockafellar1970}. Let $F$ be the smallest face that contains $\bar \varrho$. This implies that $\bar \varrho$ is in the relative interior of $F$, since otherwise there would be a smaller face containing $\bar \varrho$. $F$ is compact because $K$ is. Therefore, we can define a GPT $(V(F), V(F)^+, \one_F)$ from $F$. As $\bar \varrho \in F$ and $\bar \varrho=\sum_{a=1}^k q_{a|x} \varrho_{a|x}$ for all $x \in [g]$, it follows that $\varrho_{a|x} \in F$ for all $a\in [k]$, $x \in [g]$. Indeed, for $0 <q_{a|x} <1$ this is a consequence of $F$ being a face. For $q_{a|x} =1$, it follows that $\varrho_{a|x} = \bar \varrho$ and for $q_{a|x} =0$, we can choose $\varrho_{a|x}$ any way we like. Thus, $\Phi((CS_{k,g}^\ast)^+) \subseteq V(F)^+$ and we can consider $\Phi$ as a map $\Phi:(CS_{k,g}^\ast)^+ \to V(F)^+$ instead. 
\end{remark}

We can now more generally raise the question if there always exists a realization as in \cref{assembl_def} for a given steering assemblage as it is the case in quantum theory.
To be more precise, given a state space $K_A$ we consider an assemblage $\sigma= \{\sigma_{a|x}\}_{a \in [k],x \in [g]} \subset V(K_A)^+$ with the property that $\sum_{a \in [k]} \sigma_{a|x} = \bar{\sigma} \in K_A$ for all $x \in [g]$.
Then, there are two questions:
\begin{enumerate}
    \item Is there a state space $K_B$, a tensor product $\otimes$, a state $z \in K_A \otimes K_B$ and a multimeter $M: K_B \to CS^1_{k,g}$ such that $\sigma_{a|x} = (\id \otimes M_{a|x})(z)$ ?
    \item Is there a tensor product $\otimes$, a state $w \in K_A \otimes K_A$ and a multimeter $N: K_A \to CS^1_{k,g}$ such that $\sigma_{a|x} = (\id \otimes N_{a|x})(w)$ ?
\end{enumerate}
The answer to these questions will show that the existence of a realization cannot be taken for granted and depends on the second state space.

If we may choose $K_B$ freely, we can give a suitable choice and thus a positive answer to question (1) in the following proposition. This already follows from the considerations at the beginning of this section, but we will give a simpler proof below.
\begin{prop}
    An assemblage $\sigma= \{\sigma_{a|x}\}_{a \in [k],x \in [g]} \subset V(K_A)^+$, with an average state  $\sum_{a} \sigma_{a|x}=\bar{\sigma}$, always finds a realization, i.e.,  $\sigma_{a|x} = (\id \otimes M_{a|x})(z)$ with a state $z \in K_A \otimes K_B$ and a multimeter $M: K_B \to CS^1_{k,g}$ if $K_B$ is chosen as $K_B \coloneqq (K_A)_{\bar{\sigma}}^*$ and $\bar{\sigma}$ is in the relative interior of $K_A$.
\end{prop}
\begin{proof}
    Let $\sigma= \{\sigma_{a|x}\}_{a \in [k],x \in [g]}$ be an assemblage with $\sum_a \sigma_{a|x} = \bar{\sigma}$.
    We consider a trivial map $\varphi_{\mathrm{id}}: A(K_A)^{+} \rightarrow A(K_A)^{+}$, that is, as a tensor $\xi_{\varphi_{\mathrm{id}}} = \chi_{A(K_A)} \in V(K_A)^{+} \tmax A(K_A)^+$ (see \cref{eq:max-ent-tensor}). 
    If we could choose $M_{a|x} = \sigma_{a|x}$, we would be done since then by \cref{eq:tensor-to-map}
    \begin{equation}
        (\id \otimes M_{a|x})(\xi_{\varphi_{\mathrm{id}}}) = \sigma_{a|x}\, .
    \end{equation}
    Thus, we set $K_B = (K_A)_{\bar{\sigma}}^*$ and the assemblages become functionals on the states in $K_B$. We recall from \Cref{sec:maps-between-GPTs} that $\bar{\sigma}$ is the order unit for the GPT constructed from the state space $(K_A)_{\bar{\sigma}}^*$. Moreover, for all $a\in [k]$, $x\in [g]$ it holds that $\sigma_{a|x} \in A((K_A)_{\bar{\sigma}}^*)^+$, as $A((K_A)_{\bar{\sigma}}^*)^+ = V(K_A)^+$. Thus, setting $M_{a|x}:=\sigma_{a|x}$ indeed gives a valid multimeter acting on $(K_A)_{\bar{\sigma}}^*$. It remains for us to verify that $\xi_{\varphi_{\mathrm{id}}} \in K_A \tmax K_B$: Using again that the order unit on $(K_A)_{\bar{\sigma}}^*$ is $\bar{\sigma}$, we can confirm that
    \begin{equation}
        \one_{K_A} \otimes \one_{K_B}(\xi_{\varphi_{\mathrm{id}}}) = \one_{K_A}(\bar{\sigma}) = 1 \,,
    \end{equation}
    which is all we needed to check.
\end{proof}

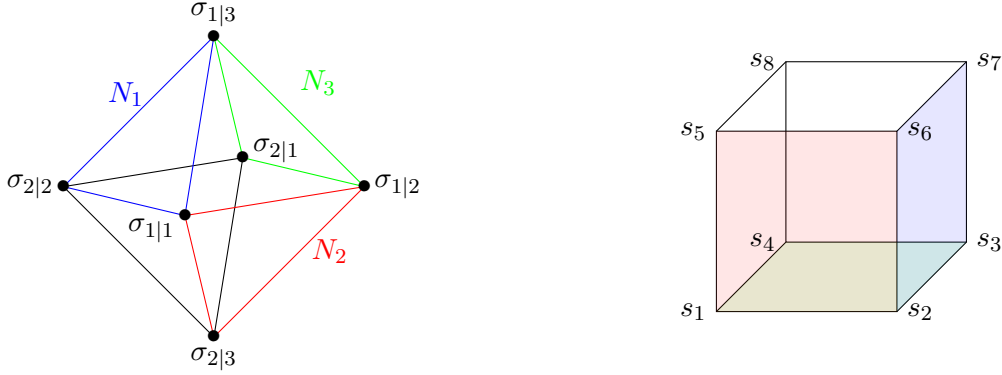
\begin{figure}
    \begin{center}
    \begin{minipage}{0.45\textwidth}
    \begin{tikzpicture}[scale=2]
      \coordinate (A) at (0,0,1);
      \coordinate (B) at (1,0,0.5);
      \coordinate (C) at (0,1,0.5);
      \coordinate (D) at (-1,0,0.5);
      \coordinate (E) at (0,-1,0.5);
      \coordinate (F) at (0,0,0);

      \draw (E) -- (D) -- (F) -- (E);
      \draw[red] (A) -- (B);
      \draw[red] (B) -- (E);
      \draw[red] (A) -- (E);
      \draw[blue] (C) -- (D);
      \draw[blue] (C) -- (A);
      \draw[blue] (D) -- (A);
      \draw[green] (B) -- (F);
      \draw[green] (B) -- (C);
      \draw[green] (F) -- (C);
      
      \node[left, yshift=-2mm] at (A) {$\sigma_{1|1}$};
      \node[black] at (A) {$\bullet$};
      \node[right] at (B) {$\sigma_{1|2}$};
      \node[black] at (B) {$\bullet$};
      \node[above] at (C) {$\sigma_{1|3}$};
      \node[black] at (C) {$\bullet$};
      \node[left] at (D) {$\sigma_{2|2}$};
      \node[black] at (D) {$\bullet$};
      \node[below] at (E) {$\sigma_{2|3}$};
      \node[black] at (E) {$\bullet$};
      \node[above, right, yshift=1mm] at (F) {$\sigma_{2|1}$};
      \node[black] at (F) {$\bullet$};

      \node at (-.7,.5,.2) {$\blue{N_1}$};
      \node at (.5,.5,0) {$\green{N_3}$};
      \node at (.5,-.7,-.2) {$\red{N_2}$};
    \end{tikzpicture}
\end{minipage}
\hfill
\begin{minipage}{0.45\textwidth}
    \begin{tikzpicture}[scale=1.2]
        \draw (0,0,0) -- (2,0,0) -- (2,2,0) -- (0,2,0) -- cycle;
        \draw (0,0,2) -- (2,0,2) -- (2,2,2) -- (0,2,2) -- cycle;
        \draw (0,0,0) -- (0,0,2);
        \draw (2,0,0) -- (2,0,2);
        \draw (2,2,0) -- (2,2,2);
        \draw (0,2,0) -- (0,2,2);

        \fill[green, opacity=.1] (0,0,2) -- (2,0,2) -- (2,0,0) -- (0,0,0);
        \fill[red, opacity=.1] (0,0,2) -- (2,0,2) -- (2,2,2) -- (0,2,2);
        \fill[blue, opacity=.1] (2,0,2) -- (2,0,0) -- (2,2,0) -- (2,2,2);

        \node[left] at (0,0,0) {$s_4$};
        \node[right] at (2,0,0) {$s_3$};
        \node[left] at (0,0,2) {$s_1$};
        \node[right] at (2,0,2) {$s_2$};
        \node[left] at (0,2,0) {$s_8$};
        \node[right] at (2,2,0) {$s_7$};
        \node[left] at (0,2,2) {$s_5$};
        \node[right] at (2,2,2) {$s_6$};
    \end{tikzpicture}
\end{minipage}
    \caption{The two state spaces represented by an octahedron (left) and a cube (right) are dual to each other. For the measurements on the octahedron we choose the measurements $N_i$ that pick out a face (colored). The opposing face is then picked by $\one_{K_{octa}} - N_i$ and the fourth measurements follows from normalization. For the cube we choose measurements that pick front (red, $M_3$), bottom (green, $M_2$) and right side (blue, $M_1$) of the cube.}
    \label{cube_n_octa}
    \end{center}

\end{figure}

However, if $K_B$ is fixed, a realization does not always exist. While this was known previously \cite{jencova2018incompatible}, explicit counterexamples are lacking in the literature. Necessary conditions on state spaces such that every assemblage admits a realization can be found in \cite{barnum2013ensemble, stevens2014steering}. We will find an explicit counterexample by considering a cube and an octahedron as shown in \Cref{cube_n_octa}. One can choose a state space represented by an octahedron. The corresponding dual state space is then represented by the cube. Note that this also works the other way round. For the measurements on the respective bodies we consider the ones that pick certain faces of the body (see also \Cref{cube_n_octa}), which generate all other measurements on that state space. 

Before we can proceed, we need to prove some facts about the measurement incompatibility in the cube and the octahedron: For a multimeter $M: K \to \mathrm{CS}^1_{k,g}$, we define its compatibility robustness as
\begin{equation}
        R_m(M) = \max \left\{\lambda \in [0,1] \, : \exists (p_{\cdot|x})_{x \in [g]} \ \mathrm{such~that}\, \left\{\lambda M_{\cdot|x} + (1-\lambda) p_{\cdot|x} \one_K \right\}_{x \in [g]} \text{ is compatible}\right\}\, .
    \end{equation}
\begin{lem} \label{lem:cube-robustness-eq-third}
    Let $M=\{M_1,M_2,M_3\}$ be a multimeter containing the measurements that pick out the front ($M_3$), bottom ($M_2$) and right ($M_1$) side of the cube as in \Cref{cube_n_octa}. Then, $R_m(M) = \frac{1}{3}$.
\end{lem}
\begin{proof}
    Analogously to an argument later in the proof of \Cref{thm:CHSH-is-linear-isom} which shows that $CS^1_{2,2}$ is isomorphic to a square, one can show that $CS^1_{2,3}$ is isomorphic to a cube (see also \cite[Proposition 3]{jencova2018incompatible} for a more general statement). Thus, we can consider the multimeter $M$ as a channel $M: CS^1_{2,3} \to CS^1_{2,3}$. Then, as shown in \cite[Example 12]{jencova2018incompatible}, the dichotomic measurements of the multimeter $M$ are maximally incompatible, that is, $R_m(M) = \frac{1}{3}$.
\end{proof}
Let us now shift the focus from measurements on the cube to measurements on the octahedron.
The aim is to find distinct separability constraints, i.e., different thresholds for robustness.
\begin{lem} \label{lem:simulate-binary-octa}
    Let $K_{\mathrm{octa}}$ be the state space represented by an octahedron and let  $N=\{N_1,N_2,N_3,N_4\}$ be the multimeter containing the face-picking measurements on $K_{\mathrm{octa}}$ from \Cref{cube_n_octa}.
    Then every dichotomic measurement can be classically simulated by $N$.
\end{lem}
\begin{proof}
    In \cite[Corollary 3]{filippov2018simulability} it was shown that a dichotomic measurement can be classically simulated by a multimeter containing dichotomic measurements if and only if the effects of the measurement are contained in the convex hull of the effects of the measurements contained in the simulating multimeter together with the zero effect $0$ and the unit effect $\one_{K_{\mathrm{octa}}}$. The effects of the four measurements $N_1,N_2,N_3,N_4$ generate the extreme rays of the effect cone and together with the zero and the unit effect, their convex hull is the whole effect space $E(K_{\mathrm{octa}})$. Thus, any dichotomic measurement can be simulated by $N$.
\end{proof}

\begin{lem} \label{lem:simulation-increase-robustness}
    Let $K$ be a state space. Then for any multimeter $M: K \to CS^1_{k,g}$ and any classical simulation $\Phi: CS^1_{k,g} \to CS^1_{l,r}$ we have that $R_m(\Phi \circ M) \geq R_m(M)$.
\end{lem}
\begin{proof}
    Let $R_m(M)=\hat{\lambda}$. Thus, there exists a conditional probability distribution $p = (p_{\cdot|x})_{x \in [g]}$ on $[k]$ such that the multimeter $\hat{M} = \{\hat{\lambda} M_{\cdot |x}+(1-\hat{\lambda})p_{\cdot|x} \one_K\}_{x \in [g]}$ is compatible. By \Cref{cor:sep-com} this means that there exists some measurement $G: K \to S_n$ and a channel $\mu: S_n \to CS^1_{k,g}$ such that $\hat{M} = \mu \circ G$. Now clearly 
    \begin{align}
        (\Phi \circ \mu) \circ G =  \Phi \circ (\mu \circ G) = \Phi \circ \hat{M}
    \end{align}
    so that again by \Cref{cor:sep-com} the multimeter $\Phi \circ \hat{M}: K \to CS^1_{l,r}$ is compatible with a joint measurement $G: K  \to S_n$ and a postprocessing channel $\Phi \circ \mu: S_n \to CS^1_{l,r}$. Now since $p \in CS^1_{k,g}$  we can set $q = \Phi(p)$ and we have that 
    \begin{equation}
        (\Phi \circ \hat{M})_{b|y} = \hat{\lambda} (\Phi \circ M)_{b|y} +(1- \hat{\lambda}) q_{b|y} \one_{K}
    \end{equation}
    for all $b \in [l]$ and $y \in [r]$. Since $\Phi \circ \hat{M}$ is compatible, it follows that $R_m(\Phi \circ M) \geq \hat{\lambda}$.    
\end{proof}

\begin{lem} \label{lem:octahedron-noise-geq-half}
    For any multimeter $E$ on $K_{\mathrm{octa}}$ consisting of dichotomic measurements it holds that $R_m(E) \geq \frac{1}{2}$.
\end{lem}
\begin{proof}
     Consider the dichotomic measurements $N_1, N_2, N_3, N_4$ on the octahedron as visualized in \Cref{cube_n_octa} and let us denote the defining effects for them as $e_1, e_2,e_3,e_4$.
    A joint measurement for the noisy versions of these measurements can be given by the measurement
    \begin{equation}
        G = \left\{\frac{e_1}{2},\,\frac{e_2}{2},\,\frac{e_3}{2},\,\one_{K_{\mathrm{octa}}} - \sum_{i=1}^3 \frac{e_i}{2} \right \}\,,
    \end{equation}
    as it can be used to obtain the noisy measurements $\{\frac{e_i}{2}, \one_{K_{\mathrm{octa}}}- \frac{e_i}{2}\}$ for all $i\in[4]$.
    We rewrite each of these measurements as $\{\frac{1}{2}e_i + \frac{1}{2} \cdot 0 \cdot \one_{K_{\mathrm{octa}}}, \frac{1}{2}(\one_{K_{\mathrm{octa}}} - e_i) + \frac{1}{2} \cdot 1 \cdot \one_{K_{\mathrm{octa}}}\}$ so that they correspond to the noisy measurements $\left\{\lambda N_i + (1-\lambda) T_i\right\}$, where we have $\lambda = \frac{1}{2},  T_{+|i}=0,\;T_{-|i}=\one_{K_{\mathrm{octa}}}$ for all $i \in [4]$. Thus, for the multimeter $N = \{N_1,N_2,N_3,N_4\}$ we have that $R_m(N) \geq \frac{1}{2}$.
    Since by \Cref{lem:simulate-binary-octa} every dichotomic multimeter $E$ can be classically simulated by $N$, and since according to \Cref{lem:simulation-increase-robustness} classical simulation cannot decrease the robustness, we have that $R_m(E)  \geq R_m(N)\geq \frac{1}{2}$ for any dichotomic multimeter $E$.
\end{proof}

Now we can finally provide our counterexample.

\begin{prop}
    Let $K_{\mathrm{octa}}$ be the state space of an octahedron. There exists a steering assemblage $\sigma = \{\sigma_{a|x}\}_{a,x} \subset V(K_{\mathrm{octa}})^+$ which does not have a realization in terms of a multimeter on $K_{\mathrm{octa}}$ and a state in $K_{\mathrm{octa}} \otimes K_{\mathrm{octa}}$ for any tensor product $\otimes$.
\end{prop}
\begin{proof}
    It is a well-known fact that the dual of a cube is an octahedron and vice versa.
    Consequently, $(V(K_{\mathrm{cube}})^+)^* = V(K_{\mathrm{octa}})^+$ and \Cref{cone_facts} implies $(V(K_{\mathrm{octa}})^+)^* = V(K_{\mathrm{cube}})^+$.

    We start by defining the robustness for assemblages $\sigma = \{\sigma_{a|x}\}_{a \in [k], x \in [g]}$, $\sum_{a = 1}^k \sigma_{a|x} = \bar \sigma$, as
    \begin{equation}
        R_s(\sigma) = \max \left\{\mu \in [0,1] \, : \, \exists \,(p_{\cdot|x})_{x\in[g]} \ \text{ s.t.} \ \left\{\mu \sigma_{a|x} + (1-\mu) p_{a|x} \bar \sigma\right\}_{a,x} \text{ has an LHS model}\right\} \, .
    \end{equation}
    Let us consider the assemblage on the octahedron given by the (unnormalized) opposing extreme points embedded in a hyperplane in a 4-dimensional space, i.e., 
     \begin{align}
        \sigma_{1|1}&=\frac{1}{2}(1,0,0,1)\,, & \sigma_{1|2}&=\frac{1}{2}(0,0,1,1)\,, &
        \sigma_{1|3}&=\frac{1}{2}(0,1,0,1)\,,  \\
        \sigma_{2|1}&=\frac{1}{2}(-1,0,0,1)\,, & \sigma_{2|2}&=\frac{1}{2}(0,0,-1,1) \, ,&
        \sigma_{2|3}&=\frac{1}{2}(0,-1,0,1)\,.
    \end{align}
    Here, the notation for the extreme points is as in \Cref{cube_n_octa}.
    We can identify this assemblage on the octahedron with measurements on the cube: if we consider the cube embedded on a hyperplane of a 4-dimensional space as $K_{\mathrm{cube}} = \{(x,y,z,1) \, : \,-1 \leq x,y,z \leq 1\} $ it follows that $\one_{K_{\mathrm{cube}}} = \bar \sigma =(0,0,0,1)$, and also that now $\{\sigma_{a|x}\}_{a\in [2], x \in [3]}$ are the only non-trivial ray-extremal effects on $K_{\mathrm{cube}}$ and that $E(K_{\mathrm{cube}}) = \mathrm{conv}(\left\lbrace \{\sigma_{a|x}\}_{a \in [2], x \in [3]}, 0, \one_{K_{\mathrm{cube}}}) \right\rbrace)$. As the ray-extremal effects (in polytope state spaces) are in one-to-one correspondence with the facets of the state space \cite[Proposition 16]{Heinosaari2019nofreeinformation}, then considering $\sigma$ as a multimeter we can identify $\sigma_{\cdot|1}$ with $M_1$, $\sigma_{\cdot|2}$ with $M_2$ and $\sigma_{\cdot|3}$ with $M_3$ from \Cref{cube_n_octa}. Then from \Cref{lem:cube-robustness-eq-third} it follows that $R_m(\sigma)=\frac{1}{3}$.
    
    Now, if $R_s(\sigma) = \lambda$, this means that there exists a conditional probability distribution $p = (p_{\cdot|x})_{x \in [3]}$ such that $\{\lambda \sigma_{\cdot|x} +(1-\lambda)p_{\cdot|x} \bar \sigma\}_{x \in [3]}$ has an LHS model. By \Cref{def:lhs_model} there exists $n \in \nat$, an ensemble $\varrho = \{\varrho_c\}_{c \in [n]} \subset V(K_{\mathrm{octa}})^+$ with $\bar \varrho :=\sum_{c=1}^n \varrho_c  = \bar \sigma$ and a conditional probability distribution $\nu = (\nu_{\cdot|c,x})_{c \in [n], x \in [3]}$ on [2] such that $\lambda \sigma_{a|x}+(1-\lambda)p_{a|x} \bar \sigma = \sum_{c=1}^n \nu_{a|c,x} \varrho_c$ for all $a \in [2]$ and $x \in [3]$. Now, since $\{\varrho_c\}_{c\in[n]} \subset V(K_{\mathrm{octa}})^+ = (V(K_{\mathrm{cube}})^+)^*$ and $\bar \varrho = \sum_c \varrho_c = \one_{K_{\mathrm{cube}}}$, it follows that $\varrho$ is a measurement on $K_{\mathrm{cube}}$. Thus, the multimeter $\{\lambda \sigma_{\cdot|x} +(1-\lambda)p_{\cdot|x} \one_{K_{\mathrm{cube}}}\}_{x \in [3]}$ is compatible, meaning that $R_m(\sigma)  \geq \lambda$ so that $R_s(\sigma) \leq R_m(M) = \frac{1}{3}$.
    
    Suppose now that $E: K_{\mathrm{octa}} \to \mathrm{CS}_{2,3}^1$ is a multimeter and $w \in K_{\mathrm{octa}} \otimes K_{\mathrm{octa}}$ such that
    \begin{equation}
        \sigma_{a|x} = (E_{a|x} \otimes \id)(w) \qquad \forall a \in [2],~x \in [3] \, .
    \end{equation}
    Then clearly $\bar \sigma = (\one_{K_{\mathrm{octa}}} \otimes \id)(w)$. Let $R_m(E) = \mu$ so that there exists a measurement $F = \{F_d\}_{d=1}^m$ on $K_{\mathrm{octa}}$ and conditional probability distributions $q = (q_{\cdot|x})_{x \in [3]}$ and $\eta = \{\eta_{\cdot|d,x}\}_{d \in [m], x \in [3]}$ on [2] such that
    \begin{equation}
        \mu E_{a|x} +(1-\mu)q_{a|x} \one_{K_{\mathrm{octa}}} = \sum_{d=1}^m \eta_{a|d,x} F_d
    \end{equation}
    for all $a \in [2]$ and $x \in [3]$. Now 
    \begin{align}
        \mu \sigma_{a|x} +(1-\mu) q_{a|x} \bar \sigma  = ((\mu E_{a|x} + (1-\mu)q_{a|x}\one_{K_{\mathrm{octa}}}) \otimes \id)(w) = \sum_{d=1}^m \eta_{a|d,x} (F_d \otimes \id)(w) 
    \end{align}
    for all $a \in [2]$ and $x \in [3]$. If we define $\kappa_d = (F_d \otimes \id)(w) \in V(K_{\mathrm{octa}})^+$, we see that $\bar \kappa:= \sum_d \kappa_d = \bar \sigma$ and thus the assemblage $\{\mu \sigma_{a|x} +(1-\mu) q_{a|x} \bar \sigma\}_{a \in [2],x \in [3]}$ has an LHS model. This shows that $R_s(\sigma) \geq \mu$. Finally, by \Cref{lem:octahedron-noise-geq-half} we have that $R_m(E) = \mu \geq \frac{1}{2}$ so that also $R_s(\sigma) \geq \frac{1}{2}$. This contradicts our earlier observation that $R_s(\sigma) \leq \frac{1}{3}$.

    Thus, the presented assemblage $\sigma$ on $K_{\mathrm{octa}}$ cannot have an realization in terms of a multimeter on $K_{\mathrm{octa}}$ and a state in $K_{\mathrm{octa}} \otimes K_{\mathrm{octa}}$ for any tensor product $\otimes$.
\end{proof}

\subsection{Factorizations of maps preparing ensembles}
For multimeters, we considered different factorizations of the associated map and related them to the compatibility of measurements, $K_B$-simulability and classical simulability. In this section, our aim is similar, but this time we consider a map that prepares the assemblage and then factorize it.

Now assume that you are given a bipartite state space $K_A \otimes K_B$ with an appropriate tensor product $ \otimes$ and a state $\varrho_{AB} \in K_A  \otimes K_B$. We recall that one way to prepare an assemblage from a bipartite state is to apply a multimeter $M$ to one of the systems, e.g., the system $K_A$.  This yields
\begin{equation}
    q_{a|x}\varrho_{a|x} = (M_{a|x} \otimes \id)(\varrho_{AB})
\end{equation}
with $\varrho_{a|x} \in K_B$ and $\sum_{a=1}^k q_{a|x}\varrho_{a|x} =\bar \varrho$ for all $x \in [g]$. Let's assume now that the multimeter $M$ is in addition $K'_A$-simulable, i.e., there exists an instrument $\Phi$ with operations $\Phi_\lambda: V(K_A)^+ \to V(K_A')^+$, and a multimeter $N = \{N_{\cdot|x,\lambda}\}_{x \in [g], \lambda \in [\Lambda]}$ on $K_A'$ such that 
\begin{equation}
    M_{a|x} = \sum_{\lambda =1}^\Lambda\Phi^\ast_\lambda(N_{a|x,\lambda})
\end{equation}
 for all $a \in [k]$ and $x \in [g]$. If we have moreover a family of channels $\Psi_\lambda: K_B \to K_B^\prime$ that can access the outcome $\lambda$ of the instrument $\Phi$, then
\begin{equation}
    q_{a|x}\varrho_{a|x} = \sum_{\lambda = 1}^\Lambda \Psi_\lambda(\sigma_{a, \lambda|x}) \qquad \forall a \in [k], \, x \in [g] \, ,
\end{equation}
where the $\sigma_{a, \lambda|x}$ are subnormalized states defined as $\sigma_{a, \lambda|x} = ((N_{a|x,\lambda}\circ \Phi_\lambda) \otimes \id)(\varrho_{AB})$. Hence, using the properties of the multimeter and the instrument, $\sum_{a=1}^k \sigma_{a, \lambda|x} = \bar \sigma_\lambda$ for all $x \in [g]$, where $\bar \sigma_\lambda$ is again a subnormalized state, and $\sum_{\lambda =1}^\Lambda \bar \sigma_\lambda = \bar \sigma \in K_B^\prime$. This motivates the following definition:

\begin{defi} \label{defi:K-simulable-steering}
    Let $\{q_{a|x}, \varrho_{a|x}\}_{a \in [k], x \in [g]}$ be an assemblage with average state $\bar \varrho \in K_A$. The assemblage is \emph{$K_B$-simulable} if there exists a finite number of outcomes $\Lambda$, a family of channels $\Psi_\lambda: K_B \to K_A$, and a set of subnormalized states $\{\sigma_{a, \lambda|x}\}_{a \in [k], \lambda \in [\Lambda], x \in [g]} \subset V(K_B)^+$ with $\sum_{a=1}^k \sigma_{a, \lambda|x} = \bar \sigma_\lambda \in V(K_B)^+$ for all $x \in [g]$, and $\sum_{\lambda =1}^\Lambda \bar \sigma_\lambda = \bar \sigma \in K_B$, such that 
\begin{equation}
 q_{a|x}\varrho_{a|x} = \sum_{\lambda = 1}^\Lambda \Psi_\lambda(\sigma_{a,  \lambda|x}) \qquad \forall a \in [k], \, x \in [g] \, .
\end{equation}
\end{defi}

As in the case of $K_B$-simulability of multimeters, we can also characterize $K_B$-simulability of assemblages in terms of a particular factorization of the associated map.

\begin{thm} \label{thm:steering-K-simulable}
 Let $\{q_{a|x}, \varrho_{a|x}\}_{a \in [k], x \in [g]}$ be an assemblage with average state $\bar \varrho \in K_A$ and  associated map $\Phi: (CS^\ast_{k,g})^+ \to V(K_A)^+$ such that $\Phi(\one_{CS^1_{k,g}}) = \bar \varrho$. Then this assemblage is $K_B$-simulable if and only if there exists a finite number of outcomes $\Lambda$, a map $\Psi_1: (CS^\ast_{k,g})^+ \to V(K_B \dot{\otimes} S_\Lambda)^+$ with $\Psi_1(\one_{CS^1_{k,g}})\in K_B \tmin S_\Lambda$ and a channel $\Psi_2: K_B \dot \otimes S_\Lambda \to K_A$ such that the following diagram commutes:
\begin{equation}
    \begin{tikzcd}
    (CS^\ast_{k,g})^+ \arrow[rd, "\Psi_1"] \arrow[rr, "\Phi"] & & V(K_A)^+ \\
    & V(K_B \dot{\otimes} S_\Lambda)^+ \arrow[ru, "\Psi_2"] &
    \end{tikzcd}  
\end{equation}
\end{thm}
\begin{proof}
Let us assume the existence of the maps $\Psi_1$, $\Psi_2$ as in the commuting diagram. Let us consider the map $\Psi_1$ first and let us define $\Psi_1(m_a^{(x)}) =: \tilde \sigma_{a|x}$ for all $a \in [k]$ and $x \in [g]$. Since $\Psi_1(\one_{CS^1_{k,g}})=:\hat \sigma \in K_B \tmin S_\Lambda$, it follows that
\begin{equation}
    \hat \sigma = \sum_{a = 1}^k \tilde \sigma_{a|x} \qquad \forall x \in [g] \, .
\end{equation}
We can furthermore decompose uniquely
\begin{equation}
    \tilde \sigma_{a|x} = \sum_{\lambda = 1}^\Lambda \sigma_{a, \lambda |x} \otimes \delta_{\lambda} \qquad \forall a \in [k], \, \forall x \in [g] \, .
\end{equation}
Here, $\sigma_{a, \lambda|x} \in V(K_B)^+$ for all $a \in [k]$, $x \in [g]$, and $\lambda \in [\Lambda]$. We can likewise decompose uniquely $\hat \sigma = \sum_{\lambda = 1}^\Lambda \bar \sigma_\lambda \otimes \delta_{\lambda}$ with $\bar \sigma_\lambda \in V(K_B)^+$ for all $\lambda \in [\Lambda]$. Since both decompositions are unique, we conclude
\begin{equation}
    \forall x \in [g] \qquad \sum_{a = 1}^k \sigma_{a, \lambda|x} = \bar \sigma_\lambda \, , \qquad\sum_{\lambda =1}^\Lambda \bar \sigma_\lambda = (\id \otimes \one_{S_\Lambda})(\hat \sigma) =: \bar \sigma \in K_B \, .
\end{equation}
As $\Psi_2$ is a channel, so is $\Psi_\lambda(\cdot) := \Psi_2(\cdot \otimes \delta_\lambda)$ for all $\lambda \in [\Lambda]$. Combining \cref{eq:Phi-steering-1} and \cref{eq:Phi-steering-2} with $\Phi = \Psi_2 \circ \Psi_1$, we obtain that indeed 
\begin{equation}
 q_{a|x}\varrho_{a|x} = \sum_{\lambda = 1}^\Lambda \Psi_\lambda(\sigma_{a,  \lambda|x}) \qquad \forall a \in [k], \, x \in [g] \, ,
\end{equation}
which proves the converse. 
The other direction follows straightforwardly along the same lines, verifying the properties of $\Psi_1$ and $\Psi_2$. Since $\{\one_{CS^1_{k,g}}, m^{(x)}_a\}_{a \in [k-1], x \in [g]}$ is a basis for $CS_{k,g}^\ast$, the $\tilde \sigma_{a|x}$ and $\hat \sigma$ uniquely define a positive map $\Psi_1$ with $\Psi_1(\one_{CS^1_{k,g}}) \in K_B \otimes S_{\Lambda}$. Finally, we can set 
\begin{equation}
    \Psi_2(x) = \sum_{\lambda = 1}^\Lambda q_\lambda \Psi_\lambda(x_\lambda) \qquad \forall x \in K_B \tmin S_\Lambda \, ,
\end{equation}
since any such $x$ has a unique decomposition as $x = \sum_{\lambda = 1}^\Lambda q_\lambda x_\lambda \otimes \delta_\lambda$ for $(q_\lambda)_{\lambda \in [\Lambda]}$ a probability distribution and $x_\lambda \in K_B$ for all $\lambda \in [\Lambda]$. Thus, $\Psi_2$ defines a channel, which concludes the proof.
\end{proof}

Now by looking at the dual map $\Phi^*$, which we were able to interpret as a multimeter, we can prove a very similar factorization:
\begin{thm} \label{thm:steering-K-simulable-second-version}
    Let $\{q_{a|x}, \varrho_{a|x}\}_{a \in [k], x \in [g]}$ be an assemblage with average state $\bar \varrho \in K_A$  and  associated map $\Phi: (CS^\ast_{k,g})^+ \to V(K_A)^+$ such that $\Phi(\one_{CS^1_{k,g}}) = \bar \varrho$. Let $\bar \varrho \in K_A$ be in the relative interior of $K_A$.  Then, there exists a finite number of outcomes $\Lambda \in \mathds N$, a map $\Psi_1: (CS^\ast_{k,g})^+ \to V(K_B \dot{\otimes} S_\Lambda)^+$ with $\Psi_1(\one_{CS^1_{k,g}})=\bar \sigma \otimes q$ for a probability distribution $(q_\lambda)_{\lambda \in [\Lambda]}$ such that $q_\lambda > 0$ for all $\lambda \in [\Lambda]$ and $\bar \sigma$ in the relative interior of $K_B$, and a map $\Psi_2: V(K_B \dot{\otimes} S_\Lambda)^+ \to V(K_A)^+$ with $\Psi_2(\bar \sigma \otimes q) =\bar \varrho$ such that the following diagram commutes:
\begin{equation}
    \begin{tikzcd}
    (CS^\ast_{k,g})^+ \arrow[rd, "\Psi_1"] \arrow[rr, "\Phi"] & & V(K_A)^+ \\
    & V(K_B \dot{\otimes} S_\Lambda)^+ \arrow[ru, "\Psi_2"] &
    \end{tikzcd}  
\end{equation}
if and only if the multimeter $\Phi^\ast : (K_A)_{\bar \varrho}^\ast \to CS_{k,g}^1$ is $(K_B)_{\bar \sigma}^\ast$-simulable.
\end{thm}
\begin{proof}
    As $\bar \sigma$ and $q$ are both in the relative interior of $K_B$ and $S_\Lambda$, respectively, we can define the state spaces $(K_B)_{\bar \sigma}^\ast$ and $(S_\Lambda)_q^\ast$. By \Cref{ex:dual-state-space-CM}, we can identify $(S_\Lambda)_q^\ast \simeq S_\Lambda$. We observe that $\Psi_1^\ast: (K_B)_{\bar \sigma}^\ast \tmin (S_\Lambda)_q^\ast \to CS_{k,g}^1$ is a channel, since for all probability distributions $(p_i)_i$ and $\alpha_i \in (K_B)_{\bar \sigma}^\ast$, $s_i \in (S_\Lambda)_q^\ast$, 
    \begin{equation}
        \left\langle \one_{CS_{k,g}^1}, \Psi_1^\ast\left(\sum_{i} p_i \alpha_i \otimes s_i\right) \right\rangle =  \sum_{i}\left\langle \Psi_1(\one_{CS_{k,g}^1}),  p_i \alpha_i \otimes s_i\right \rangle = \sum_{i} p_i \alpha_i(\bar \sigma) s_i(q) = 1 \,.
    \end{equation}
    Moreover, as $\Psi_2(\hat \sigma) = \bar \varrho$ with $\hat \sigma = \bar \sigma \otimes q$, we have that for all $\alpha \in (K_A)_{\bar \varrho}^\ast$, 
    \begin{equation}
        \langle \Psi_2^\ast(\alpha), \bar \sigma \otimes q \rangle = \alpha(\Psi_2(\hat \sigma)) = \alpha(\bar \varrho) = 1 \, ,
    \end{equation}
    hence $\Psi_2^\ast: (K_A)_{\bar \varrho}^\ast \to (K_B)_{\bar \sigma}^\ast \tmin (S_\Lambda)_q^\ast$ is a channel. The assertion then follows by \Cref{thm:compression-as-factorization}. It thus follows that the multimeter $\Phi^\ast$ is $(K_B)_{\bar \sigma}^\ast$-simulable.
    
    Conversely, for $(K_B)_{\bar \sigma}^\ast$ to be a state space $\bar \sigma$ has to be in the relative interior of $K_B$.  \Cref{thm:compression-as-factorization} guarantees the factorization $\Phi^\ast = \Psi_1^\ast \circ \Psi_2^\ast$ with channels 
    $\Psi_1^\ast: (K_B)_{\bar \sigma}^\ast \tmin S_\Lambda \to CS_{k,g}^1$ and 
    $\Psi_2^\ast: (K_A)_{\bar \varrho}^\ast \to (K_B)_{\bar \sigma}^\ast \tmin S_\Lambda$. Thus, $\Psi_1$ and $\Psi_2$ are both unital. Picking any $q$ in the relative interior of $S_\Lambda$, we can define $S_\Lambda \simeq (S_\lambda)_q^\ast$ such that $\one_{S_\Lambda} \simeq q$ (see  \Cref{ex:dual-state-space-CM}). Note also that $\one_{(K_B)_{\bar \sigma}^\ast} = \bar \sigma$. We recall that $A((K_A)_{\bar \varrho}^\ast)^+ = V(K_A)^+$ and $A((K_B)_{\bar \sigma}^\ast)^+ = V(K_B)^+$. Hence, $\Psi_1: (CS_{k,g}^\ast)^+ \to V(K_B)^+ \tmin V((S_\Lambda)_q^\ast)^+$ fulfills $\Psi_1(\one_{CS_{k,g}^1})=\bar \sigma \otimes q \in K_B \tmin (S_\Lambda)_q^\ast $ and $\Psi_2: V(K_B)^+ \tmin V((S_\Lambda)_q^\ast)^+ \to V(K_A)^+$ fulfills $\Psi_2(\bar \sigma \otimes q) = \bar \varrho$.
\end{proof}

\begin{remark}
    The main differences between the factorizations in \Cref{thm:steering-K-simulable} and \Cref{thm:steering-K-simulable-second-version} is that in \Cref{thm:steering-K-simulable} the map $\Psi_2$ is a channel, whereas we only require $\Psi_2(\bar \sigma \otimes q) = \bar \varrho$ in \Cref{thm:steering-K-simulable-second-version}. On the other hand, $\Psi_1(\one_{CS_{k,g}^1})=\bar \sigma \otimes q$ in \Cref{thm:steering-K-simulable-second-version}, whereas in \Cref{thm:steering-K-simulable} $\Psi_1(\one_{CS_{k,g}^1}) \in K_B \tmin S_\Lambda$ needs not be a pure tensor product. 
\end{remark}

\begin{remark} \label{rem:steering-simuable-to-quantum}
    Let us take a closer look at  \Cref{thm:steering-K-simulable-second-version} for the case of quantum theory. Seeing $\Phi^\ast=\{M_{\cdot|x}\}_{x \in [g]}$ as a multimeter on $(D(\Cnum^d))^\ast_{\bar \varrho}$, the factorization of $\Phi$ in \Cref{thm:steering-K-simulable-second-version} is equivalent to the existence of a finite number of outcomes $\Lambda$, of an instrument $\Psi: (D(\Cnum^d))^\ast_{\bar \varrho} \to (D(\Cnum^{d^\prime}))^\ast_{\bar \sigma} \dot{\otimes} S_\Lambda$ with operations $\Psi_\lambda: \mathrm{PSD}_d \to \mathrm{PSD}_{d^\prime}$, and of a multimeter $N = \{N_{\cdot|x,\lambda}\}_{x \in [g], \lambda \in [\Lambda]}$ of $g \cdot \Lambda$  measurements with $k$ outcomes on $(D(\Cnum^{d^\prime}))^\ast_{\bar \sigma}$ such that 
\begin{equation}
    M_{a|x} = \sum_{\lambda =1}^\Lambda\Psi^\ast_\lambda(N_{a|x,\lambda})
\end{equation}
for all $a \in [k]$ and $x \in [g]$. We will now relate this to the existence of ordinary quantum multimeters and instruments. \Cref{ex:dual-state-space-QM} tells us the shape of $(D(\Cnum^d))^\ast_{\bar \varrho}$ and $(D(\Cnum^{d^\prime}))^\ast_{\bar \sigma}$. Thus, there exists a multimeter $\{\tilde M_{\cdot|x}\}_{x \in [g]}$ on $D(\Cnum^d)$ such that 
    \begin{equation}
        \tilde M_{a|x}(\tau) = M_{a|x}\left(\bar \varrho^{-\frac{1}{2}} \tau \bar \varrho^{-\frac{1}{2}}\right) \qquad \forall a \in [k], \forall x \in [g], \forall \tau \in D(\Cnum^d) \, .
    \end{equation}
    Moreover, there exist positive maps $\tilde \Psi^\ast_\lambda: PSD_d \to PSD_{d^\prime}$ such that $\sum_{\lambda =1}^\Lambda \tilde \Psi_\lambda^\ast: D(\Cnum^d) \to D(\Cnum^{d^\prime})$ is a channel (although it is not a quantum channel, since it is not necessarily completely positive) and such that 
    \begin{equation}
        \tilde \Psi^\ast_\lambda(\tau) =\bar \sigma^{\frac{1}{2}} \Psi^\ast_\lambda(\bar \varrho^{-\frac{1}{2}} \tau \bar \varrho^{-\frac{1}{2}})\bar \sigma^{\frac{1}{2}} \qquad \forall \lambda \in [\Lambda], \forall \tau \in D(\Cnum^d) \, .
    \end{equation}
    Finally, there exists a multimeter $\{\tilde N_{\cdot|x, \lambda}\}_{x \in [g], \lambda \in [\Lambda]}$ on $D(\Cnum^{d^\prime})$ such that 
    \begin{equation}
        \tilde N_{a|x, \lambda}(\tau) = N_{a|x, \lambda}\left(\bar \sigma^{-\frac{1}{2}} \tau \bar \sigma^{-\frac{1}{2}}\right) \qquad \forall a \in [k], \forall x \in [g], \forall \lambda \in [\Lambda], \forall \tau \in D(\Cnum^d) \, .
    \end{equation}
    Hence, since $\bar \varrho$ and $\bar \sigma$ are invertible, the factorization of $\Phi$ is equivalent to 
    \begin{equation}
    \tilde M_{a|x} = \sum_{\lambda =1}^\Lambda\tilde \Psi^\ast_\lambda(\tilde N_{a|x,\lambda})
\end{equation}
for all $a \in [k]$ and $x \in [g]$.
\end{remark}

Similarly to the discussion that motivated the definition of $K_B$-simulable assemblages, we can instead assume the multimeter preparing the assemblage to be classically simulable. This motivates the following definition:
\begin{defi}
    Let $\{q_{a|x}, \varrho_{a|x}\}_{a \in [k], x \in [g]}$ be an assemblage with average state $\bar \varrho \in K$. We say that the assemblage can be \emph{classically simulated} (or is \emph{classically simulable}) by an assemblage $\{v_{b|y}, \sigma_{b|y}\}_{b \in [l], y \in [r]}$ with $\bar \sigma = \bar \varrho \in K$ if there exist conditional probability distributions $\pi = (\pi_{\cdot|x})_{x \in [g]}$ on $[r]$ and $\nu= (\nu_{\cdot|b,x,y})_{b \in [l], x \in [g], y\in [r]}$ on $[k]$ such that 
\begin{equation}
    q_{a|x} \varrho_{a|x} = \sum_{y=1}^r \pi_{y|x} \sum_{b=1}^l \nu_{a|b,x,y} v_{b|y} \sigma_{b|y} 
\end{equation}
for all $a \in [k]$ and $x \in [g]$.
\end{defi}
We can again find a characterization of classically simulable assemblages in terms of factorizations.

\begin{thm} \label{thm:classical-steering}
 Let $\{q_{a|x}, \varrho_{a|x}\}_{a \in [k], x \in [g]}$ be an assemblage with average state $\bar \varrho \in K$ and  associated map $\Phi: (CS^\ast_{k,g})^+ \to V(K)^+$ such that $\Phi(\one_{CS^1_{k,g}}) = \bar \varrho$. Then this assemblage is classically simulable by some assemblage $\{v_{b|y}, \sigma_{b|y}\}_{b \in [l], y \in [r]}$ with $\bar \sigma = \bar \varrho \in K$  if and only if there exist a map $\Psi_1: (CS^\ast_{k,g})^+ \to (CS_{l, r}^\ast)^+$ with $\Psi_1(\one_{CS^1_{k,g}})=\one_{CS^1_{l,r}}$ and a map $\Psi_2: (CS_{l, r}^\ast)^+ \to V(K)^+$ with $\Psi_2(\one_{CS^1_{l,r}}) \in K$ such that the following diagram commutes:
\begin{equation}
    \begin{tikzcd}
    (CS^\ast_{k,g})^+ \arrow[rd, "\Psi_1"] \arrow[rr, "\Phi"] & & V(K)^+ \\
    & (CS_{l, r}^\ast)^+ \arrow[ru, "\Psi_2"] &
    \end{tikzcd}  
\end{equation}
\end{thm}
\begin{proof}
    Using \Cref{thm:channel-polysimplices-characterization} and duality of maps, every map $\Psi_1: (CS^\ast_{k,g})^+ \to (CS_{l, r}^\ast)^+$ with $\Psi_1(\one_{CS^1_{k,g}})=\one_{CS^1_{l,r}}$ can be written as
    \begin{equation}
        \Psi_1(\alpha) = \alpha(s_{k, \ldots,k})\one_{CS^1_{l,r}}  + \sum_{x=1}^g \sum_{a =1}^{k-1} \alpha(e^{(x)}_a)\left[ \sum_{y=1}^r \pi_{y|x} \sum_{b=1}^l \nu_{a|b,x,y} m^{(y)}_b\right] \, ,
    \end{equation}
    where $\alpha \in (CS^\ast_{k,g})^+$. Now given $\Psi_2: (CS_{l, r}^\ast)^+ \to V(K)^+$, we can define $\Psi_2(m^{(y)}_b) =: v_{b|y} \sigma_{b|y}$ with $\sigma_{b|y} \in K$ and $v_{b|y} :=  \one_{CS^1_{l,r}}(\Psi_2(m^{(y)}_b))$. Setting moreover
    \begin{equation}
        \sum_{b=1}^l v_{b|y} \sigma_{b|y} =: \bar \sigma \in K \qquad \forall y \in [r] \, ,
    \end{equation}
    it follows that $(v_{b|y})_{b \in [l]}$ is a probability distribution for all $y \in [r]$. As $\Psi_1(\one_{CS^1_{k,g}})=\one_{CS^1_{l,r}}$, it follows that $\bar \sigma = \bar \varrho$. Using $\Phi = \Psi_2 \circ \Psi_1$, we find
    \begin{equation}
    q_{a|x} \varrho_{a|x} = \sum_{y=1}^r \pi_{y|x} \sum_{b=1}^l \nu_{a|b,x,y} v_{b|y} \sigma_{b|y} 
\end{equation}
for all $a \in [k-1]$ and $x \in [g]$. For $a = k$, we infer
     \begin{equation}
    q_{k|x} \varrho_{k|x} =\bar \varrho - \underbrace{\sum_{a=1}^{k-1} \sum_{y=1}^r \pi_{y|x} \sum_{b=1}^l \nu_{a|b,x,y} v_{b|y} \sigma_{b|y}}_{\sum_{y=1}^r \pi_{y|x} \sum_{b=1}^l (1-\nu_{k|b,x,y})v_{b|y} \sigma_{b|y}} = \sum_{y=1}^r \pi_{y|x} \sum_{b=1}^l \nu_{k|b,x,y} v_{b|y} \sigma_{b|y} \,.
\end{equation}

    The converse is straightforward defining the map $\Psi_2$ using the assemblage  $\{v_{b|y}, \sigma_{b|y}\}_{b \in [l], y \in [r]}$.
\end{proof}
Instead of seeing classical simulability as a property of the assemblage, we can consider it as a property of the associated multimeter, as the next corollary shows.

\begin{cor}
     Let $\{q_{a|x}, \varrho_{a|x}\}_{a \in [k], x \in [g]}$ be an assemblage with average state $\bar \varrho \in K$ and  associated map $\Phi: (CS^\ast_{k,g})^+ \to V(K)^+$ such that $\Phi(\one_{CS^1_{k,g}}) = \bar \varrho$. Let $\bar \varrho$ be in the relative interior of $K$. Then this assemblage is classically simulable if and only if the multimeter $\Phi^\ast: K_{\bar \varrho}^\ast \to CS_{k,g}^1$ is classically simulable. 
\end{cor}

\begin{proof}
    By \Cref{thm:classical-steering}, the assemblage is classically simulable if and only if there exist a map $\Psi_1: (CS^\ast_{k,g})^+ \to (CS_{l, r}^\ast)^+$ with $\Psi_1(\one_{CS^1_{k,g}})=\one_{CS^1_{l,r}}$ and a map $\Psi_2: (CS_{l, r}^\ast)^+ \to V(K)^+$ with $\Psi_2(\one_{CS^1_{l,r}}) \in K$ such that $\Phi = \Psi_2 \circ \Psi_1$. Therefore, $\Psi_2^\ast: K_{\bar \varrho}^\ast \to CS_{l,r}$ is a multimeter, and as $\Psi_1$ is unital, $\Psi_1^\ast: CS^1_{l, r} \to CS^1_{k,g}$ is a channel as well. As $\Phi = \Psi_2 \circ \Psi_1$, we have $\Phi^\ast = \Psi_1^\ast \circ \Psi^\ast_2$. The assertion follows from \Cref{cor:simulation-factorization} and the converse follows in the same manner.
\end{proof}

Finally, we can recover from our results a characterization of assemblages admitting an LHS model, see also \Cref{prop:LHS-separable}. 

\begin{cor} \label{cor:separbility-LHS}
Let $\{q_{a|x}, \varrho_{a|x}\}_{a \in [k], x \in [g]}$ be an assemblage with average state $\bar \varrho \in K$ and  associated map $\Phi: (CS^\ast_{k,g})^+ \to V(K)^+$ such that $\Phi(\one_{CS^1_{k,g}}) = \bar \varrho$. Then, the following are equivalent:
\begin{enumerate}
    \item $(\id \otimes \Phi)(\chi_{CS^\ast_{k,g}}) \in CS^1_{k,g} \tmin K$ 
    \item $\{q_{a|x}, \varrho_{a|x}\}_{a \in [k], x \in [g]}$ has an LHS model
    \item There exist a finite number $\Lambda$,  a map $\Psi_1: (CS^\ast_{k,g})^+ \to (S_\Lambda^+)^\ast$ with $\Psi_1(\one_{CS^1_{k,g}})=\one_{S_\Lambda}$ and a map $\Psi_2: (S_\Lambda^+)^\ast \to V(K)^+$ with $\Psi_2(\one_{S_\Lambda}) \in K$ such that the following diagram commutes:
\begin{equation}
    \begin{tikzcd}
    (CS^\ast_{k,g})^+ \arrow[rd, "\Psi_1"] \arrow[rr, "\Phi"] & & V(K)^+ \\
    & (S_\Lambda^+)^\ast \arrow[ru, "\Psi_2"] &
    \end{tikzcd}  
\end{equation}
\end{enumerate}
\end{cor}
\begin{proof}
    The equivalence of $(1)$ and $(3)$ follows from \Cref{prop:min-factoring-and-tensor}. The equivalence of $(2)$ and $(3)$ follows from \Cref{thm:classical-steering} with $r=1$ and $l = \Lambda$, because a classical simulation of this kind has the form 
    \begin{equation}
    q_{a|x} \varrho_{a|x} = \sum_{\lambda=1}^\Lambda \nu_{a|\lambda,x} v_{\lambda} \sigma_{\lambda} \, .
\end{equation}
for all $a \in [k]$ and $x \in [g]$.
\end{proof}

\section{Bell nonlocality}\label{sec:bell}

Another form of quantum nonlocality that is closely related to measurement incompatibility is Bell nonlocality, a phenomenon that questions the assumption of locality in an experiment. A no-signaling probability distribution (Bell behavior) is said to be local if it has a local realization (local hidden variable (LHV) model). Then Bell nonlocality can be detected through a violation of a Bell inequality. We can use our formalism to show how no-signaling probability distributions can be seen as tensors and that separability of that tensor corresponds to having a local realization through an LHV model. Again, this will correspond to a factorization of the associated map through a simplex and we see how different factorizations give different simulation processes for the no-signaling distributions. We will also see how Bell inequalities can be formulated in this framework and we will be able to recover some important known results from \cite{Wolf_2009, Barnum_2010, Pironio_2014}.

\subsection{Bell nonlocality in quantum theory}

We begin investigating Bell nonlocality by considering two multimeters $M: D(\Cnum^{d_A}) \rightarrow \CSC$, $N: D(\Cnum^{d_B}) \rightarrow CS_{l,r}^+$ and $\varrho \in D(\Cnum^{d_A d_B})$. 
Both multimeters together yield a map
\begin{equation}
    M \otimes N: \varrho \mapsto (M \otimes N)(\varrho),
\end{equation}
which forms a bipartite Bell behavior with the outcome statistic
\begin{equation}
\label{out_stat}
    \Tr[(M_{a|x} \otimes N_{b|y})\, \varrho] = p_{a,b|x,y} \geq 0 \quad \text{with } a \in [k],\, b \in [l],\, x \in [g],\, y \in [r] \,.
\end{equation}

The following is straightforward to confirm.
\begin{lem}
\label{stat_in_CS_max}
    The image $(M \otimes N)(\varrho)$ is in the maximal tensor product $\CS \tmax CS_{l,r}^1$.
\end{lem}

As a classical explanation to nonlocality one could search for a model that explains nonlocal correlations by local hidden variables. Knowledge of these hidden variables would allow for a complete determination of the measurement outcomes and hence leave the inherit randomness of quantum theory behind.
Such models are consequently called local hidden variable models and will be abbreviated with LHV models from now on.

\begin{defi}
    A Bell behavior $(p_{a_1, \dots, a_N | x_1, \dots, x_N})_{a_j \in [k_j], x_j \in [g_j], j \in [N]}$ with $N$ parties, each performing $m$ measurements with $n$ outcomes, can be described by an LHV model if there is a locally hidden variable $\Lambda \in \mathds N$ and probability distributions $q$ on $[\Lambda]$ and $(p_{\cdot|\lambda,x_j})_{\lambda \in [\Lambda], x_j \in [g_j]}$ on $[k_j]$ for all $j \in [N]$ such that the outcome statistics defined likewise as in \cref{out_stat} can be written as
    \begin{equation}
        p_{a_1, \dots, a_N | x_1 \dots x_N} = \sum_{\lambda \in [\Lambda]} q_\lambda \prod_{j=1}^N p_{a_j |\lambda, x_j}
    \end{equation}
    for all $a_j \in [k_j]$, $x_j \in [g_j]$ and $j \in [N]$.
\end{defi}

The idea is partially similar to LHS models, and as with LHS models, there is a quite similar connection between separability and LHV models to be found.
However, in contrast to steering, where every bipartite quantum assemblage has a realization (\Cref{assem_real}), not all outcome statistics can be achieved by quantum theory.
Formally, this means $\exists P \in \CS \tmax CS_{l,r}^1$ st. $P \neq (M \otimes N)(\varrho)$, where $M, \, N$ are quantum multimeters and $\varrho$ is a density matrix.
An example of such objects $P$ are Popescu-Rohrlich nonlocal boxes \cite{Popescu_Rohrlich_94}. 

Typically, in scenarios where more than one party performs sets of measurements on a shared system, the parties are considered spatially separated.
For example, the parties are in distinct laboratories performing their operations and measurements.
Thus, it is fair to assume that the outcomes of one party are not affected by the measurement choice of another party, since that would require instantaneous or superluminal signal transmission.
Hence, the assumption is well motivated by special relativity and is widely known as no-signaling.

\begin{defi}
\label{ns_behavior}
    A conditional probability distribution $(p_{\cdot, \cdot |xy})_{x \in [g],y \in [r]}$ on $[k] \times [l]$ is said to be no-signaling if it satisfies:
    \begin{enumerate}
        \item $\sum_{a=1}^k p_{a,b|x,y} = \sum_{a=1}^k p_{a,b|x',y} \quad \forall x,x' \in [g]$, $b \in [l]$
        \item $\sum_{b=1}^l p_{a,b|x,y} = \sum_{b=1}^l p_{a,b|x,y'} \quad \forall y,y' \in [r]$, $a \in [k]$
    \end{enumerate}
    In the context of experiments, these probability distributions are also called no-signaling behaviors.
\end{defi}

The concept of no-signaling is more general and weaker than local hidden variables.
As the next lemma shows, all physical statistics are covered in the set $\CS \tmax CS^1_{l, r}$ (see \Cref{prop:NS-tensor} for proof, originally shown in  \cite{jencova2018incompatible}). 
\begin{lem} \label{lem:NS}
    The set of all no-signaling behaviors is in one-to-one correspondence with $\CS \tmax CS_{l, r}^1$.
\end{lem}

As we saw with the steering assemblages, the separable assemblages correspond to those that are classically explainable with LHS models. We can make a similar conclusion here with the separable no-signaling behaviors being classically explainable with LHV models (see \Cref{cor:separability-LHV-simulation} for proof, originally shown in  \cite{plavala2017conditions, jencova2018incompatible}):

\begin{lem}
\label{sep-LHV}
    Let $P \in \CS \tmax CS_{l,r}^1$. Then $P \in \CS \tmin CS_{l,r}^1$ if and only if $P$ has an LHV model.
\end{lem}

In particular, if a behavior $P$ has a realization as $P = (M\otimes N)(\varrho)$ for some multimeters $M$ and $N$ and some shared state $\varrho$, we get the following:

\begin{cor}
    Consider $M \otimes N : D(\Cnum^{d_A d_B}) \rightarrow \CS \tmax CS_{l, r}^1$ and $\varrho \in D(\Cnum^{d_A d_B})$. Then $(M \otimes N)(\varrho)$ has a LHV model if and only if $(M \otimes N)(\varrho) \in \CS \tmin CS_{l, r}^1$.
\end{cor}

It is possible to detect nonlocality in Bell behaviors with Bell inequalities.
Statistics from quantum theory are able to exceed bounds to statistics stemming from local realist behaviors, that is, statistics that have an LHV model.
There are also bounds to the statistics of quantum theory, for instance the Tsirelson bound, which in turn is exceeded by the already mentioned Popescu-Rohrlich boxes. Those are consequently also nonlocal.
A necessary condition for violating a Bell inequality is the inseparability of the quantum states and incompatibility of the measurements in the Bell experiment \cite{Fine2}. As we saw earlier, incompatibility can also be seen as a form of entanglement. Therefore, the Bell inequality acts as an entanglement witness in a generalized sense.

\begin{defi}
    A Bell inequality is a witness-type map defined as $W : \CS \tmax CS_{l,r}^1 \rightarrow \Rnum$ such that $W(P) \geq 0$  for all $P \in \CS \tmin CS_{l, r}^1$.
\end{defi}

This definition can easily be extended to every finite number of parties.
We call functionals witnesses if they are positive on separable elements.
Moreover, we will also often make use of the fact that for every inseparable (entangled) element, there is always a witness that detects it, i.e., maps it to a negative number.
For more details, the reader is asked to see \Cref{appendix:tensor-cones}.
The claim from above that Bell inequalities are generalized entanglement witnesses shall now be formalized for two parties. 

\begin{lem}
\label{BI_maxtp}
    Every Bell inequality is an element of the set $\CSCd \tmax (CS_{l,r}^*)^+$ and vice versa, every functional in this set is a Bell inequality.
\end{lem}

The statement follows directly from the fact that $(CS^*_{k,g})^+ \tmax (CS^*_{l,r})^+ = (CS_{k,g}^+ \tmin CS_{l,r}^+)^*$ (see \cref{appendix:tensor-cones}). 

As mentioned before, it is a well-known fact that measurement incompatibility in the Bell setup is a necessity for violating Bell inequalities. 
In fact, no Bell inequality with $N$ parties can be violated if $N-1$ parties perform compatible measurements \cite{Fine2}.
In the bipartite case, this reduces to one set of compatible POVMs.
Both entanglement and incompatibility are a form of inseparability, which can be witnessed by a Bell inequality.

Popescu-Rohrlich nonlocal boxes are able to achieve the algebraic maximum of Bell inequalities, succeeding the Tsirelson bounds of quantum theory.
A question arises whether normalized elements of $D(\Cnum^{d_A}) \tmax D(\Cnum^{d_B})$ lead to statistics unachievable by quantum states.
That is, we could relax the condition on the joint system and instead of considering density matrices, let it be a more general object, the elements of which get mapped to probabilities by POVMs.
In this way, we might get more than quantum correlations in $\CSC \tmax CS_{l,r}^+$.
However, in \cite{Barnum_2010} this possibility was already ruled out.
Here we are able to write down a different proof and also give another point of view on this result.

\begin{thm}
    Let $d_B \geq d_A$. For all multimeters $M: D(\Cnum^{d_A}) \to \CSmn^1$, $N: D(\Cnum^{d_B}) \to CS_{l,r}^1$ and $\Gamma \in D(\Cnum^{d_A}) \tmax D(\Cnum^{d_B})$ there exist multimeter $\tilde N: D(\Cnum^{d_B}) \to CS_{l, r}^1$ and $\varrho \in D(\Cnum^{d_A d_B})$ such that
\begin{equation}
    (M \otimes N)(\Gamma) = (M \otimes \Tilde N)(\varrho)\, .
\end{equation}
\end{thm}

\begin{proof}
    For this $\Gamma$ we can write
    \begin{equation}
        (\id \otimes N)(\Gamma) \in D(\Cnum^{d_A}) \tmax \CS, \mspace{8mu} (M \otimes N)(\Gamma) \in \CS \tmax CS^1_{l, r} \, .
    \end{equation}
    These objects are still positive and we get an assemblage $\sigma$ from $\Gamma$.
    According to \Cref{assem_real} the obtained assemblage $\sigma$ finds a quantum realization 
    \begin{equation}
        (\id \otimes N)(\Gamma) = \sigma =(\id \otimes \Tilde N)(\varrho) \, .
    \end{equation}
    Now we are already done since for all $a \in [k],\,b \in [l],\,x \in [g], \,y \in [r]$ it will hold that
    \begin{equation}
        \Tr[(M_{a|x} \otimes N_{b|y}) \, \Gamma] = \Tr[(M_{a|x} \otimes \Tilde N_{b|y}) \, \varrho] \, .
    \end{equation}
\end{proof}
It follows that there cannot be a Bell inequality to detect if the statistics in $\CS \tmax CS^1_{l, r}$ stemmed from a density matrix in $D(\Cnum^{d_A d_B})$ or a block-positive, unit trace matrix $\Gamma \in D(\Cnum^{d_A}) \tmax D(\Cnum^{d_B})$.
\begin{cor}
    For every $\Gamma \in D(\Cnum^{d_A}) \hat \otimes D(\Cnum^{d_B})$ and multimeters $M: D(\Cnum^{d_A}) \to \CSmn^1$, $N: D(\Cnum^{d_B}) \to CS_{l,r}^1$ there is no Bell inequality $W \in \CSCd \tmax (CS^\ast_{l,r})^+$ with $W((M \otimes N)(\Gamma)) < 0$ and $W((\Tilde M \otimes \Tilde N)(\varrho)) \geq 0$ for all multimeters $\Tilde M: D(\Cnum^{d_A}) \to \CSmn^1$, $ \tilde N: D(\Cnum^{d_B}) \to CS_{l,r}^1$ and density matrices $\varrho \in D(\Cnum^{d_A d_B})$.
\end{cor}
Note that the correctness of this theorem is not clear anymore when extending to the general multipartite case.
The technique of the given proof breaks down since there is no quantum realization for every steering assemblage in case of more than two parties as it was shown in \cite{postQ_steering} by considering PR-boxes again.

\subsection{CHSH inequalities as incompatibility witnesses}

This section is about proving a theorem linking CHSH inequalities and isomorphisms between cones with square basis, and about applying it in order to retrieve interesting results in quantum theory.
The CHSH scenario, in which two parties perform two dichotomic (also binary or two-outcome) measurements, therefore plays a special role among Bell experiments.
Most commonly, a CHSH inequality is represented by quantum correlations, which are expectation values of the products of outcomes $\pm 1$. A lower and an upper bound together with 4 positions of a negative sign lead to 8 CHSH inequalities.
\begin{equation}
    \begin{split}
        &-2 \leq E_{1,1} + E_{1,2} + E_{2,1} - E_{2,2} \leq 2 \\
        &-2 \leq E_{1,1} + E_{1,2} - E_{2,1} + E_{2,2} \leq 2 \\
        &-2 \leq E_{1,1} - E_{1,2} + E_{2,1} + E_{2,2} \leq 2 \\
        &-2 \leq  - E_{1,1} + E_{1,2} + E_{2,1} + E_{2,2} \leq 2
    \end{split}
\end{equation}
where $E_{x,y} = p_{1,1|x,y} - p_{1,2|x,y} - p_{2,1|x,y} + p_{2,2|x,y}$.
Here, the first index corresponds to the first party and the second index to the second party, whereas the value indicates which measurement is being used.
Those inequalities hold for all classical correlations and can be violated by nonlocal theories such as quantum mechanics or theories containing PR-boxes.

\begin{thm}
\label{thm:CHSH-is-linear-isom}
    The tensors $\xi_L \in (CS^*_{2,2})^+ \tmax (CS^*_{2,2})^+$ corresponding to linear isomorphisms  $L: \CStwo \rightarrow \CStwodu$ are in one-to-one correspondence with CHSH inequalities.
\end{thm}
This was also observed in \cite[Example 14]{jencova2018incompatible}. Here, we give a more elementary proof. We first provide an outline of the proof.
\begin{enumerate}
    \item We show that the cone $\CStwo$ is generated by matrices $\{X_i\}_{i\in[4]}$ satisfying the square condition and we map them to vectors $\{x_i\}_{i \in [4]}$ in the Euclidean space $\Rnum^3$.
    \item For the vectors $\{x_i\}_{i \in [4]} \subset \Rnum^3$ we construct dual vectors $\{\vp_j\}_{j \in [4]} \subset  (\Rnum^3)^* \cong \Rnum^3$ and we make sure that both sets of vectors lie in a square and span the corresponding vector space.
    \item A linear isomorphism $L$ maps vectors $\{x_i\}_{i \in [4]}$ to dual vectors $\{\vp_j\}_{j \in [4]}$ in $(\Rnum^3)^*$.
    Using the one-to-one correspondence of linear maps and maximal tensor products of cones, we can decompose the tensor $\xi_L$.
    \item We find the dual matrices $\{F_j\}_{j \in [4]}$ corresponding to the set of dual vectors and specify $\xi_L$.
    \item The decomposition enables us to show that $\xi_L$ is a CHSH inequality.
    \item We show that all CHSH inequalities correspond to such  a tensor $\xi_L$, or equivalently, to such an isomorphism $L$.
\end{enumerate}
\begin{proof}
Any matrix from $CS_{2,2}^1$ can be written in the following form, with $a$, $b \in [0,1]$:
\begin{align}
    &\begin{bmatrix}
        a & b \\
        1-a &1-b
    \end{bmatrix}
    = a
    \begin{bmatrix}
        1 & b \\
        0 & 1-b
    \end{bmatrix}
    + (1-a)
    \begin{bmatrix}
        0 & b \\
        1 & 1-b
    \end{bmatrix} \\
    &= a \left(b
    \begin{bmatrix}
    1 & 1 \\
    0 & 0 
    \end{bmatrix} 
    +(1-b)
    \begin{bmatrix}
        1 & 0 \\
        0 & 1
    \end{bmatrix}
    \right)
    +(1-a) \left(b
    \begin{bmatrix}
        0 & 1 \\
        1 & 0
    \end{bmatrix}
    +(1-b)
    \begin{bmatrix}
        0 & 0 \\
        1 & 1
    \end{bmatrix}  
    \right)
\end{align}
So we can define the generating matrices
\begin{equation}
    X_1 = \begin{bmatrix}
            1 & 1 \\
            0 & 0 
        \end{bmatrix}, \,
    X_2 = \begin{bmatrix}
        1 & 0 \\
        0 & 1
        \end{bmatrix}, \,
    X_3 = \begin{bmatrix}
        0 & 0 \\
        1 & 1
        \end{bmatrix}, \,
    X_4 = \begin{bmatrix}
        0 & 1 \\
        1 & 0
        \end{bmatrix},
\end{equation}
which form the basis of the cone $\CStwo$.
Note that:
\begin{equation}
\label{square_condition}
    X_1 + X_3 = X_2 + X_4 \, ,
\end{equation}
which will be referred to as the ``square condition" because it says that two diagonals intersect in the center of the shape the 4 points form.
Consequently the cone $\CStwo$ is generated by a square.
We now use an isomorphism $\iota: \CStwo \rightarrow \Rnum^3, \, X_i \mapsto x_i$ to map the four matrices $X_i$ to vectors in Euclidean space $\Rnum^3$ that, satisfying \cref{square_condition}, lie in a square.
\begin{equation}
    x_1 = \begin{pmatrix} 1 \\ -1 \\ 1 \end{pmatrix}, \,
    x_2 = \begin{pmatrix} 1 \\ 0 \\ 0 \end{pmatrix}, \,
    x_3 = \begin{pmatrix} 0 \\ 1 \\ 0 \end{pmatrix}, \,
    x_4 = \begin{pmatrix} 0 \\ 0 \\ 1 \end{pmatrix}
\end{equation}
The choice was arbitrary except that the vectors $\{x_i\}_{i \in [4]}$ have to satisfy \cref{square_condition} and as such they can be seen as the extremal points of $\CStwo$.
From here we want to find dual vectors $\{\vp_j\}_{j \in [4]}$ that fulfill the constraint
\begin{equation}
\label{eq:duality-const-CHSH}
    \vp_j(x_i ) \geq 0 \, .
\end{equation}
To find the dual vectors, fix one $\vp_j = (a \mspace{5mu} b \mspace{5mu} c)$ and with \cref{eq:duality-const-CHSH} one can solve for the components of $\vp_j$.
\begin{equation}
    \vp_1 = (1 \mspace{5mu} 1 \mspace{5mu} 0), \, \vp_2 = (0 \mspace{5mu} 1 \mspace{5mu} 1), \, \vp_3 = (0 \mspace{5mu} 0 \mspace{5mu} 1), \, \vp_4 = \vp_1 - \vp_2 + \vp_3 = (1 \mspace{5mu} 0 \mspace{5mu} 0) 
\end{equation}
Also these dual vectors form a square and they span $(\Rnum^3)^*$.
As both cones are generated by a square, they are isomorphic to each other $\CStwo \cong \CStwodu$.
The linear isomorphism is $L : \CStwo \rightarrow \CStwodu$.
We can map a square to a square by mapping a vertex to a vertex.
Since the isomorphism must be linear, neighboring vertices must remain being neighbors.
Therefore we are left with $4 \cdot 2 = 8$ isomorphisms $L$.
For instance, one could choose $\hat L$ such that:
\begin{equation}
        x_1 \overset{h}{\mapsto} \vp_2, \, x_2 \overset{h}{\mapsto} \vp_3, \, x_3 \overset{h}{\mapsto} \vp_4, \, x_4 \overset{h}{\mapsto} \vp_1\, ,
\end{equation}
where $h = (\iota^*)^{-1} \circ \hat L \circ \iota^{-1}:\Rnum^3 \rightarrow (\Rnum^3)^*$, but up to the isomorphisms this can be thought of as just $\hat L$.
For an isomorphism $L$ we define the tensor $\xi_L \in (CS^*_{2,2})^+ \tmax (CS^*_{2,2})^+$:
\begin{equation}
\label{respawn}
    \vp_{j_L}(x_i) = \langle L(x_j),x_i\rangle = \langle L(\iota^{-1}(X_j)),\iota^{-1}(X_i)\rangle = \langle\xi_L,X_j \otimes X_i \rangle\, ,
\end{equation}
Let $\{e_i\}_{i \in [3]}$ and $\{f_j\}_{j \in [3]}$ be the standard bases of $\Rnum^3$ and $\Rnum^{3^*}$ respectively, such that $f_j( e_i ) = \delta_{ji}$ with the usual Kronecker delta.
We write down the action of $\xi_L$ using its basis decomposition $\xi_L = \sum_{k,l = 1}^3 \alpha_{kl} \mu_k \otimes \mu_l$, where $\mu_k = \iota^*(f_k)$, so that 
\begin{align}
    \langle\xi_L, X_i \otimes X_j \rangle &= \left\langle \sum_{k,l = 1}^3 \alpha_{kl} \mu_k \otimes \mu_l, \iota^{-1}(x_i)  \otimes \iota^{-1}(x_j) \right\rangle  \\
    &= \sum_{k,l = 1}^3 \alpha_{kl} \, \langle \mu_k, \iota^{-1}(x_i) \rangle \langle \mu_l,\iota^{-1}(x_j) \rangle  \\
    &=  \sum_{k,l = 1}^3 \alpha_{kl} \langle (\iota^{-1})^*(\mu_k), x_i \rangle \langle (\iota^{-1})^*(\mu_l) x_j \rangle \\
    &= \langle \sum_{k,l = 1}^3 \alpha_{kl} \, f_k \otimes f_l, x_i \otimes x_j \rangle
\end{align}
and the coefficients are found via $\langle \xi_L, \iota^{-1}(e_k) \otimes \iota^{-1}(e_l)\rangle = ( L(e_k))(e_l ) = \alpha_{kl}$.
Since the set of dual vectors $\{\vp_j\}_{j \in [4]}$ span $(\Rnum^3)^*$, the standard basis $\{f_j\}_{j \in [3]}$ can be expressed through them and vice versa.
Expressing $\xi_{\hat L}$ with the $\vp$'s yields for the above example:
\begin{equation} \label{eq:decomp-psi-hatL}
    \langle \xi_{\hat L}, \, \cdot \, \rangle =\langle \vp_3 \otimes \vp_2 + \vp_2 \otimes \vp_4 + (\vp_4 - \vp_3) \otimes \vp_3, \iota \otimes \iota(\,\cdot\,) \rangle
\end{equation}

Due to the isomorphisms $\iota$ and $\xi_L$ there must be a correspondence between the functionals $\{\vp_j\}_{j \in [4]}$ and dual matrices $\{F_j\}_{j \in [4]}$ via:
\begin{equation}
\label{eq:dual-matrix-correspondance}
    \Tr(F_j^\intercal \, X_i) = \langle F_j,X_i\rangle = \langle F_j \circ \iota^{-1} \circ \iota, X_i \rangle = \langle \vp_j , \iota (X_i)\rangle =\vp_j (x_i)
\end{equation}
Recall that the matrices $X$  contain conditional probabilities $X_{a|x}=p_{a|x}$ for $a,x \in [2]$.
So, the result of \cref{eq:dual-matrix-correspondance} has to be a combination of (conditional) probabilities.
Let $q$ be an arbitrary vector in convex hull of $\{x_1,x_2,x_3,x_4\}$ in $\Rnum^3$ corresponding to some column stochastic matrix
\begin{equation}
    \begin{bmatrix} p_{1|1} & p_{1|2} \\ p_{2|1} & p_{2|2} \end{bmatrix} \, .
\end{equation}
Then, there is a convex combination of matrices $\{X_i\}_{i \in [4]}$, such that:
\begin{equation}
\label{decomp_phiq}
    \vp_j( q ) = \left \langle \vp_j, \iota \left(\sum_{i=1}^4 \lambda_i X_i \right) \right \rangle=\Tr \left(F_j^\intercal \, \begin{bmatrix} p_{1|1} & p_{1|2} \\ p_{2|1} & p_{2|2} \end{bmatrix} \right)
\end{equation}
This can be written down explicitly for every $q$ by calculating $\vp_j(x_i)$ because $\{x_i\}_{i \in [4]}$ span $\Rnum^3$.
This is explained best by showing one instance, that is, we fix one index $j$, say $\vp_1$.
We have
\begin{equation}
    \vp_1(x_1) = \vp_1(x_4)=0, \, \vp_1(x_2) = \vp_1(x_3) = 1
\end{equation}
and the matrices $\{X_i\}_{i \in [4]}$ were given above.
By solving \cref{eq:dual-matrix-correspondance} for every $i$ we can deduce that $F_1^\intercal = \begin{bmatrix} 1 & 1 \\ -1 & 0 \end{bmatrix}$ and so 
\begin{equation}
    \vp_1(q) = \Tr \left(F_1^\intercal \, \begin{bmatrix} p_{1|1} & p_{1|2} \\ p_{2|1} & p_{2|2} \end{bmatrix} \right) = p_{1|1} + p_{2|1} - p_{1|2} = 1 - p_{1|2} = p_{2|2} \, .   
\end{equation}
Note that the construction of $F_j$ is not unique although matrices, vectors and dual vectors are fixed.
We repeat the same procedure for the remaining $\vp_j$ to compute $\vp_j(q)$ and find:
\begin{equation}\label{eq:CHSH-action-of-vp}
    \vp_1(q) = p_{2|2}, \, \vp_2(q) = p_{2|1}, \, \vp_3(q) = p_{1|2}, \, \vp_4(q) = p_{1|1}
\end{equation}
The results we obtained again satisfy the square condition \cref{square_condition}.
By construction (recall \cref{eq:duality-const-CHSH}) $\langle \xi_L,P \rangle \geq 0$ for any $P$ that is a convex combination of product elements, i.e., for any separable $P \in \CStwo \tmin \CStwo$ and hence $\xi_L \in \CStwodu \tmax \CStwodu$ for any $L$.
It follows from \Cref{BI_maxtp} that $\xi_L$ is a Bell inequality.

For $P = \sum_k \lambda_k Q_k \otimes \tilde{Q}_k \in \CStwo \tmin \CStwo$, we use the above positivity and plug the obtained expressions of $\vp_j(q_k)$ and $\vp_j(\tilde{q}_k)$ for the corresponding vectors $q_k$ and $\tilde{q}_k$ of the matrices $Q_k$ and $\tilde{Q}_k$ into the decomposition \cref{eq:decomp-psi-hatL} to find  that
\begin{align}
    \langle \xi_{\hat L}, P \rangle &=\langle (\iota^* \otimes \iota^*)(\vp_3 \otimes \vp_2 + \vp_2 \otimes \vp_4 + (\vp_4 - \vp_3) \otimes \vp_3), P \rangle\\
    &= \sum_{k=1}^n \lambda_k \, \langle F_3 \otimes F_2 + F_2 \otimes F_4 + (F_4-F_3)\otimes F_3), Q_k \otimes \Tilde Q_k\rangle\\
    &= \sum_{k=1}^n \lambda_k \, \left(p_{1|2}^{(k)} \cdot \tilde{p}_{2|1}^{(k)} + p_{2|1}^{(k)} \cdot \tilde{p}_{1|1}^{(k)} + p_{1|1}^{(k)} \cdot \tilde{p}_{1|2}^{(k)} - p_{1|2}^{(k)} \cdot \tilde{p}_{1|2}^{(k)} \right)\\
    &= p_{1,2|2,1} + p_{2,1|1,1} + p_{1,1|1,2} - p_{1,1|2,2} \geq 0 ,\label{eq:CHSH-ineq-1}
\end{align}
where we have set $p_{a,b|x,y}=\sum_{k=1}^n \lambda_k \,  p_{a|x}^{(k)} \cdot \tilde{p}_{b|y}^{(k)}  $, which is a LHV model for $P$.
We want to rewrite this in terms of expectation values in order to show that this is indeed a CHSH inequality.
Using the normalization of the probabilities and the marginal probabilities as well as standard correlations
\begin{equation}\label{eq:CHSH-ineq-2}
    E_{x,y} = p_{1,1|x,y} - p_{1,2|x,y} - p_{2,1|x,y} + p_{2,2|x,y} \, ,
\end{equation}
the inequality obtained from \cref{eq:CHSH-ineq-1} is equivalent to
\begin{equation}
    E_{1,1} + E_{2,1} + E_{2,2} - E_{1,2} \leq 2 \, .
\end{equation}
This is a CHSH inequality. Constructing the remaining CHSH inequalities now becomes easier as we can obtain them by applying the symmetries of a square to $\hat L$.

\begin{figure}
    \centering
    \begin{tikzpicture}

    \draw[thick, blue] (0,0) rectangle (4,4);

    \node[below] at (0,0) {$\varphi_1$};
    \node[below] at (4,0) {$\varphi_2$};
    \node[above] at (4,4) {$\varphi_3$};
    \node[above] at (0,4) {$\varphi_4$};

    \draw[<->, thick,green] (.3,.3) -- (3.7,3.7);
    \node[right, green] at (2.2,2){$Q_1$};
    
    \draw[<->, thick,orange] (.3,3.7) -- (3.7,.3);
    \node[left, orange] at (1.8,2){$Q_2$};

    \draw[<->, thick] (-0.8,0) -- (-0.8,4);
    \node at (-1.2,2) {$S_2$};
    \draw[<->, thick] (4.8,0) -- (4.8,4);
    \node at (5.2,2) {$S_2$};
    \draw[<->, thick] (0,-0.8) -- (4, -0.8);
    \node at (2,-.5) {$S_1$};
    \draw[<->, thick] (0,4.8) -- (4, 4.8);
    \node at (2,4.5) {$S_1$};
\end{tikzpicture}
    \caption{The isomorphisms $S_1, \, S_2, \, Q_1$ and $Q_2$ are reflections of the square and can also be used to get counterclockwise rotations by $\frac{\pi}{2}, \pi \text{ and } \frac{3 \pi}{2}$. The arrows indicate the effects of the maps, that is, they reflect across the symmetry axes of the square. As a consequence, measurement choices and outcomes get swapped.}
    \label{isom_square}
\end{figure}

By symmetries of a square we mean reflections along the symmetry axes and rotations as indicated in \cref{isom_square}, i.e.\ the elements of the dihedral group $D_4$, which is known to have 8 of them.
By comparing the effects of the symmetries in \cref{isom_square} with the explicit expressions in \cref{eq:CHSH-action-of-vp}, it can readily be verified that the symmetries have physical implications, e.g.\ $\vp_1(q) = p_{2|2} \overset{S_1}{\mapsto} \vp_2(q) = p_{2|1}$.
\begin{enumerate}
    \item id: identity
    \item $S_1$: reflection along vertical axis (flip measurement choice)
    \item $S_2$: reflection along horizontal axis (flip measurement and outcome)
    \item $Q_1$: reflection along diagonal $\vp_1 - \vp_3$ axis (flip outcomes of second measurement)
    \item $Q_2$: reflection along diagonal $\vp_2 - \vp_4$ axis (flip outcomes of first measurement)
    \item $R_{\frac{\pi}{2}}=S_1 \circ Q_2$: rotation by $\frac{\pi}{2}$
    \item $R_\pi = S_1 \circ S_2$: rotation by $\pi$
    \item $R_{\frac{3\pi}{2}}=S_2 \circ Q_2$: rotation by $\frac{3\pi}{2}$
\end{enumerate}
A key step in obtaining the inequality corresponding to $\hat L$ was to find the decomposition of $\xi_{\hat L}$ in terms of $\{\vp_j\}_{j \in [4]}$ (see \cref{eq:decomp-psi-hatL}).
By concatenating $\hat L$ with an element from $D_4$, that is, another isomorphism between squares, we get a decomposition of a new tensor $\xi_L$.
This allows for recycling the previously calculated $\alpha_{kl}$, e.g.\ 
\begin{equation}
    L(e_1) = L(x_2) = R_\pi \circ \hat L(x_2) = R_\pi(\vp_3) = \vp_1 = \hat L(x_4) = \hat L(e_3) \Rightarrow \alpha_{1l} = \alpha^{(\hat L)}_{3l} \, \forall l \in [3] \, .
\end{equation}
This leads to another CHSH inequality by repeating the same procedure as before and in the end, to all CHSH inequalities.
\begin{align}
	&\langle \xi_{R_{\frac{\pi}{2}} \circ \hat L}, P \rangle = p_{1,2|2,1} + p_{1,1|1,1} - p_{1,2|1,2} + p_{2,2|2,2} \geq 0 \Leftrightarrow -2 \leq E_{1,1} - E_{2,1} + E_{2,2} + E_{1,2} \\
	&\langle \xi_{S_1 \circ \hat L}, P \rangle = p_{2,2|1,1} + p_{1,2|1,2} - p_{2,2|2,2} + p_{2,1|2,1} \geq 0 \Leftrightarrow -E_{1,1} + E_{2,1} + E_{2,2} + E_{1,2} \leq 2 \\
	&\langle \xi_{R_{\pi \circ \hat L}}, P \rangle = p_{1,1|1,1} + p_{2,2|2,1} - p_{2,1|2,2} + p_{2,1|1,2} \geq 0 \Leftrightarrow -2 \leq E_{1,1} + E_{2,1} + E_{2,2} - E_{1,2} \\
	&\langle \xi_{S_2 \circ \hat L}, P \rangle = p_{1,2|1,1} + p_{1,2|1,2} - p_{1,2|2,2} + p_{1,1|2,1} \geq 0 \Leftrightarrow -2 \leq - E_{1,1} + E_{2,1} + E_{2,2} + E_{1,2} \\
	&\langle \xi_{R_{\frac{3\pi}{2}} \circ \hat L}, P \rangle = p_{1,2|1,1} + p_{2,1|1,2} - p_{1,1|2,2} + p_{1,1|2,1} \geq 0 \Leftrightarrow E_{1,1} - E_{2,1} + E_{2,2} + E_{1,2} \leq 2 \\
	&\langle \xi_{Q_1 \circ \hat L}, P \rangle = p_{1,2|1,2} + p_{2,1|1,1} - p_{1,1|2,1} + p_{1,1|2,2} \geq 0 \Leftrightarrow E_{1,1} + E_{2,1} - E_{2,2} + E_{1,2} \leq 2 \\
	&\langle \xi_{Q_2 \circ \hat L}, P \rangle = p_{1,1|1,2} + p_{2,2|1,1} - p_{1,2|2,1} + p_{1,2|2,2} \geq 0 \Leftrightarrow -2 \leq E_{1,1} + E_{2,1} - E_{2,2} + E_{1,2}
\end{align}
This concludes that all 8 CHSH inequalities can be obtained by such an isomorphism $L$ from above.
\end{proof}

Of course, there are easier ways to derive the CHSH inequality and the above result was already discovered.
So, let us immediately make use of the obtained isomorphism.
In \cite{Wolf_2009} the authors showed the equivalence of dichotomic POVMs being incompatible and them allowing for violation of the CHSH inequality.
We reformulate this result in our terms and prove it using our established formalism.
\begin{thm}
\label{result1}
    A multimeter $M: D(\Cnum^{d_A}) \to CS^1_{2,2}$ is incompatible if and only if there exists another multimeter $N: D(\Cnum^{d_B}) \to CS^1_{2,2}$ and a bipartite entangled state $\varrho \in D(\Cnum^{d_A d_B})$ such that the resulting behavior violates the CHSH inequality.
\end{thm}
\begin{proof}
    To begin with, let $M: D(\Cnum^{d_A}) \to CS^1_{2,2}$ be a multimeter and let us consider the corresponding tensor $\xi_M \in \Posa^* \tmax \CStwo$. From \Cref{witness_exists} we get a witness-type functional $W_M$ such that $\langle W_M,\xi_M\rangle < 0$ if $\xi_M$ is not separable, i.e., $M$ is incompatible.
    \begin{equation}
    \label{witness}
    W_M \in \Posa \tmax \CStwodu = (\id \otimes L) (\Posa \tmax \CStwo),
    \end{equation}
    where the linear isomorphism $L: \CStwo \rightarrow \CStwodu$ was used. We are free to choose a normalization here, so we can take $W_M \in (\id \otimes L)(D(\mathds C^{d_A}) \tmax CS_{2,2}^1)$.  
    The set $D(\mathds C^{d_A}) \tmax CS_{2,2}^1$ is the already known set of assemblages $\sigma$.
    Recall that for each assemblage $\sigma \in D(\mathds C^{d_A}) \tmax CS_{2,2}^1$ there is a multimeter $N: D(\Cnum^{d_B}) \to CS^1_{2,2}$ and a bipartite state $\varrho \in D(\Cnum^{d_A d_B})$ that reproduce it as $\sigma = (\id \otimes N)(\varrho)$. Thus, there exists such $\sigma$ such that $W_M = (\id \otimes L)(\sigma)$. 
    By expressing $\varrho = \sum_a \alpha_a \otimes \beta_a$ we have that
    \begin{align}
        0 > \langle W_M,\xi_M\rangle
        &= \langle (\id \otimes L)(\sigma), \xi_M\rangle = \sum_a \langle \alpha_a \otimes L(N(\beta_a)), \xi_M\rangle \\
        &= \sum_a \langle M(\alpha_a), L(N(\beta_a)) \rangle = \sum_a \langle\xi_L, M(\alpha_a) \otimes N(\beta_a) \rangle \\
        &= \langle \xi_L, (M\otimes N)(\varrho)\rangle,
    \end{align}
    where $\xi_L \in (CS^*_{2,2})^+ \tmax (CS^*_{2,2})^+$ is the tensor corresponding to the map $L$ which by \cref{thm:CHSH-is-linear-isom} can be interpreted as a CHSH inequality.

    For the converse direction, one just needs to read the previous equations backward and conclude $\langle (\id \otimes L)(\sigma), \xi_M\rangle <0$. Now $W_M = (\id \otimes L)(\sigma)$ is an incompatibility witness for $M$. As an image of the assemblage $\sigma \in D(\mathds C^{d_A}) \tmax CS_{2,2}^1$ under the positive map $\id \otimes L: \Posa \tmax CS^+_{2,2} \to \Posa \tmax (CS^*_{2,2})^+$ we have that $W_M \in \Posa \tmax \CStwodu$. By the definition of the maximal tensor product we then have that   $\langle W_M, \xi \rangle \geq 0$ for all $\xi \in \Posa^* \tmin \CStwo$. Thus, $W_M$ detects the incompatibility of $M$.
\end{proof}

The more well-known direction is the implication that if the set of measurements $M$ enables the violation of a Bell inequality, then it must be incompatible.
Usually, this is shown via contraposition: If in a bipartite scenario a Bell inequality is violated, then all measurements must have been incompatible because the measurement incompatibility on both parties is necessary.
As soon as one party has more than two measurements, incompatibility does not imply Bell nonlocality anymore \cite{noJM-noB,Bene2018}.

We go a bit further and relax the conditions on one party.
So now suppose that the second party has an arbitrary number of measurements $g$ as well as arbitrary many outcomes $k$, whereas without loss of generality $k$ can be set to an equal constant for every measurement by adding zero outcomes.
This gives some measurement statistics, which could arise from any (nonlocal) model. Then the same witness-type argument can be made as before.
Again, all nontrivial facets of the Bell polytope can be constructed from CHSH, which was found in \cite{Pironio_2014}.

\begin{thm}
\label{result2}
    For any Bell inequality $W \in (CS^*_{2,2})^+ \tmax (CS^*_{k,g})^+$ there is a positive map $\Xi^*: CS^+_{k,g} \to CS^+_{2,2}$ such that $W$ can be expressed through the CHSH inequality as $W=(\id \otimes \Xi^*)(\xi_L)$, where $\xi_L \in (CS^*_{2,2})^+ \tmax (CS^*_{2,2})^+$ is the tensor corresponding to $L$ in \cref{thm:CHSH-is-linear-isom}.
\end{thm}
The channel $\Xi^*$ can be understood as a post-processing, enabling the reduction of the witness $W$ to a CHSH inequality. 
\begin{proof}    
    We can rewrite $W \in (CS^*_{2,2})^+ \tmax (CS^*_{k,g})^+$ by using the isomorphism $L: CS^+_{2,2} \to (CS^*_{2,2})^+$, which is the already familiar CHSH inequality, simply as $W = (L \otimes \id)(L^{-1} \otimes \id)(W)= (L \otimes \id)(\xi_\Xi)$, where we have defined $\xi_\Xi = (L^{-1} \otimes \id)(W) \in \CStwo \tmax \CSCd$. Now $\xi_\Xi$ corresponds to a positive map $\Xi : \CStwodu \rightarrow \CSCd$.

    Now if we write $P = \sum_i A_i \otimes B_i$ we see that 
    \begin{align}
        0 > \langle W, P\rangle &= \sum_i \langle (L \otimes \id)(\xi_\Xi), P \rangle = \sum_i \langle \xi_\Xi, L^*(A_i) \otimes B_i \rangle = \sum_i \langle L^*(A_i) , \Xi^*(B_i) \rangle \\
        &= \sum_i \langle \xi_L, A_i \otimes \Xi^*(B_i) \rangle = \langle \xi_L, (\id \otimes \Xi^*)(P) \rangle, 
    \end{align}
    where $\Xi^*: CS^+_{k,g} \to CS^+_{2,2}$. The map $\Xi^*$ is positive since its dual map $\Xi$ is.
\end{proof}

Up to normalization, we see that $\Xi^*$ is exactly the type of map that was characterized in \cref{thm:channel-polysimplices-characterization} as classical simulations. The way $\Xi^*$ acts on $P$ is by reducing the $k \cross g$ column stochastic matrix to a $2 \cross 2$ column stochastic matrix. 

\subsection{Bell inequalities in general probabilistic theories}
In this section, we will again look at Bell inequalities from a factorization point of view. The following in particular proves \Cref{lem:NS}.
\begin{prop} \label{prop:NS-tensor}
    $\xi \in CS^1_{k,g} \tmax CS^1_{l, r}$ if and only if there exist no-signaling probability distributions $(p_{a,b|x,y})_{a \in [k],b \in [l], x \in [g], y \in [r]}$ such that
    \begin{align}
    \xi &= \sum_{x = 1}^g \sum _{y = 1}^r \sum_{a = 1}^{k-1} \sum_{b=1}^{l-1} p_{a,b|x,y} e_a^{(x)} \otimes e_b^{(y)} + \sum_{x=1}^g \sum_{a = 1}^{k-1} \left(\sum_{b=1}^{l} p_{a,b|x,1}\right) e_a^{(x)} \otimes s_{l, \ldots, l} \\ 
    & + \sum_{y=1}^r \sum_{b=1}^{l-1} \left(\sum_{a=1}^{k} p_{a,b|1,y}\right) s_{k, \ldots, k} \otimes e_b^{(y)} +  s_{k, \ldots, k} \otimes s_{l, \ldots, l}\, .
    \end{align}
\end{prop}
\begin{proof}
Let $\xi \in CS^1_{k,g} \tmax CS^1_{l, r}$. Then, we can use the basis of the polysimplex to write
\begin{equation}
    \xi = \sum_{j = 1}^g \sum_{q = 1}^r \sum_{i = 1}^{k-1} \sum_{p=1}^{l-1} \lambda_{i,p|j,q} e_i^{(j)} \otimes e_p^{(q)} + \sum_{j = 1}^g \sum_{i = 1}^{k-1} \mu_{i|j} e_i^{(j)} \otimes s_{l, \ldots, l} + \sum_{q = 1}^r \sum_{p=1}^{l-1} \mu^\prime_{p|q} s_{k, \ldots, k} \otimes e_p^{(q)} + \eta  s_{k, \ldots, k} \otimes s_{l, \ldots, l}\, .
\end{equation}
We infer $\langle \mathds 1_{CS^1_{k,g}} \otimes  \mathds 1_{CS^1_{l,r}}, \xi \rangle = \eta$, thus $\eta = 1$. Moreover, using that the $m_i^{(j)}$ are positive functionals for all $i$ and $j$, we obtain
\begin{align}
    \langle m_i^{(j)} \otimes m_p^{(q)}, \xi \rangle = \lambda_{i,p|j,q} \geq 0 \quad & \quad \forall j \in [g], q \in [r], i \in [k-1], p \in [l-1] \, , \\
    \langle m_i^{(j)} \otimes m_{l}^{(q)}, \xi \rangle = \mu_{i|j} - \sum_{p=1}^{l-1} \lambda_{i,p|j,q} \geq 0 \quad & \quad \forall j \in [g], q \in [r], i \in [k-1] \, , \\
    \langle m_{k}^{(j)} \otimes m_p^{(q)}, \xi \rangle = \mu^\prime_{p|q} - \sum_{i = 1}^{k-1} \lambda_{i,p|j,q} \geq 0 \quad & \quad \forall j \in [g], q \in [r], p \in [l-1] \, , \\
    \langle m_i^{(j)} \otimes \mathds 1_{CS^1_{l, r}}, \xi \rangle = \mu_{i|j} \geq 0 \quad & \quad \forall j \in [g], i \in [k-1] \, , \\
    \langle \mathds 1_{CS_{g,k}^1} \otimes m_p^{(q)}, \xi \rangle = \mu^\prime_{p|q} \geq 0 \quad & \quad \forall q \in [r], p \in [l-1] \, 
\end{align}
and 
\begin{equation}
     \langle m_{k}^{(j)} \otimes m_{l}^{(q)}, \xi \rangle = 1 +\sum_{i = 1}^{k-1} \sum_{p=1}^{l-1} \lambda_{i,p|j,q} -\sum_{i = 1}^{k-1} \mu_{i|j} - \sum_{p=1}^{l-1} \mu_{p|q}^\prime \geq 0 \, .
\end{equation}
We can therefore define $p_{a,b|x,y} =  \langle m_a^{(x)} \otimes m_b^{(y)}, \xi \rangle$ for all $x \in [g]$, $y \in [r]$, $a \in [k]$, $b \in [l]$, i.e.,
\begin{equation*}
    p_{a,b|x,y} = \begin{cases}\lambda_{a,b|x,y} & x \in [g], y \in [r], a \in [k-1], b \in [l-1] \\
    \mu^\prime_{b|y} - \sum_{a = 1}^{k-1} \lambda_{a,b|x,y} & x\in [k], y \in [r], a =k, b \in [l-1] \\
    \mu_{a|x} - \sum_{b = 1}^{l-1} \lambda_{a,b|x,y} & x\in [k], y \in [r], b =l, a \in [k-1] \\
    1 +\sum_{a = 1}^{k-1} \sum_{b=1}^{l-1} \lambda_{a,b|x,y} -\sum_{a = 1}^{k-1} \mu_{a|x} - \sum_{b=1}^{l-1} \mu_{b|y}^\prime & x\in [k], y \in [r], a=k, b=l
    \end{cases} 
\end{equation*}
From the inequalities above, we can verify that the $(p_{a,b|x,y})_{a \in [k],b \in [l]}$ defined in this way form valid probability distributions. As $\sum_{i} m_i^{(j)} = \mathds 1_{CS_{g,k}^1}$ for all $j$,  they satisfy the no-signaling condition, as
\begin{align}
    \sum_{b = 1}^{l} p_{a,b|x,y} = \mu_{a|x} \quad & \quad \forall y \in [r] \, , \\
    \sum_{a = 1}^{k} p_{a,b|x,y} = \mu^\prime_{b|y} \quad & \quad \forall x \in [g] \, .
\end{align}
The converse follows as the $m^{(j)}_i$ form the extreme rays of $(CS^+_{g,k})^*$ and $(CS^+_{l,r})^*$.
\end{proof}

By the general correspondence between tensors and positive maps (see \Cref{sec:maps-between-GPTs}) we can also equate the no-signaling probability distributions with some positive maps.

\begin{cor} \label{cor:Phi-NS}
    $\Phi: (CS_{k, g}^+)^\ast \to CS_{l, r}^+$ with $\Phi(\mathds 1_{CS^1_{k,g}}) \in CS^1_{l, r}$ if and only if and there exist no-signaling probability distributions $(p_{a,b|x,y})_{a \in [k],b \in [l]}$  such that
    \begin{align}
        \langle m_b^{(y)}, \Phi(m_a^{(x)})\rangle = p_{a,b|x,y} \quad & \quad \forall a \in [k], \forall b \in [l],\forall x \in [g],\forall y \in [r] \\
        \langle \mathds 1_{CS^1_{l, r}}, \Phi(m_a^{(x)})\rangle = \sum_{b=1}^{l} p_{a,b|x,1} \quad & \quad \forall a \in [k], \forall x \in [g] \\
        \langle m_b^{(y)}, \Phi(\one_{CS^1_{k,g}})\rangle = \sum_{a=1}^{k} p_{a,b|1,y} \quad  & \quad \forall b \in [l], \forall y \in [r] \\
        \langle \one_{CS^1_{l, r}}, \Phi(\one_{CS^1_{k,g}})\rangle = 1
    \end{align}
\end{cor}
\begin{proof}
    This follows from directly from \Cref{lem:xi-to-phi} and \Cref{prop:NS-tensor}. 
\end{proof}

Analogously to $K$-simulability of multimeters and assemblages, we can look into $K$-simulability of no-signaling distributions. It turns out this relates to how the no-signaling distribution can be realized in terms of multimeters and state assemblages on the simulating state space.

\begin{thm} \label{thm:Bell-compression}
    Let $\Phi: (CS^+_{k,g})^\ast \to CS_{l, r}^+$ with $\Phi(\one_{CS^1_{k,g}}) \in CS^1_{l, r}$ and let $(V(K_A), V(K)^+, \mathds 1_K)$ be a GPT. Furthermore, let $\Psi: (CS^+_{k,g})^\ast \to V(K)^+ \tmin S_\Lambda^+$ such that $\Psi(\one_{CS^1_{k,g}}) \in K \tmin S_\Lambda$ and let $F: K \tmin S_\Lambda \to CS^1_{l, r}$. Then, the following diagram commutes
    \begin{equation}
        \begin{tikzcd}
        (CS^\ast_{k,g})^+ \arrow[rd, "\Psi"] \arrow[rr, "\Phi"] & & CS^+_{l, r} \\
        & V(K)^+ \tmin S_\Lambda^+\arrow[ru, "F"] &
        \end{tikzcd}  
    \end{equation}
if and only the no-signaling probability distributions defining $\Phi$ (see \Cref{cor:Phi-NS}) satisfy
    \begin{equation} \label{eq:Bell-compression}
        p_{a,b|x,y} = \sum_{\lambda=1}^\Lambda w_\lambda  \langle F_{b|y, \lambda}, q_{a|x, \lambda} \varrho_{a|x, \lambda} \rangle \qquad \forall a \in [k], \forall b \in [l], \forall x \in [g], \forall y \in [r] \, .
    \end{equation}
    Here, $(w_\lambda)_{\lambda \in [\Lambda]}$ is a probability distribution,  $(q_{a|x, \lambda})_{a\in[k]}$ are some conditional probability distributions, $F_{b|y, \lambda} \in A(K)^+$ are such that $\sum_{b=1}^l F_{b|y, \lambda} = \mathds 1_K$ for all $y \in [r]$ and $\lambda \in [\Lambda]$; in other words, $F$ is a multimeter. The $\varrho_{a|x, \lambda} \in K$ are such that $\sum_{a = 1}^k q_{a|x, \lambda} \varrho_{a|x, \lambda} = \bar \varrho_\lambda$ for all $x \in [g]$ and $\lambda \in [\Lambda]$ with $\bar \varrho_\lambda \in K$. Finally, $\sum_{\lambda \in [\Lambda]} w_\lambda \bar \varrho_\lambda = \bar \varrho \in K$, i.e., $\{q_{a|x, \lambda}, \varrho_{a|x, \lambda}\}_{a \in [k], x \in [g]}$ is an assemblage for all $\lambda \in [\Lambda]$.
\end{thm}
\begin{proof}
    We start by realizing that any map $F: K \tmin S_{\Lambda} \to CS^1_{l, r}$ is a multimeter and can therefore be written as $F= \{F_{\cdot|y, \lambda}\}_{y \in [r], \lambda \in [\Lambda]}$ (see also \Cref{lemma:1}). In particular, $m_b^{(y)}\circ F(\cdot \otimes \delta_\lambda) = F_{b|y, \lambda}$. Moreover, we find that 
    \begin{align}
        \Psi(m_a^{(x)}) &=:  \tilde \sigma_{a|x} \in V(K)^+ \tmin S^+_\Lambda\qquad \forall a \in [k], \forall x \in [g] \\
        \Psi(\mathds 1_{CS^1_{k,g}}) &=: \bar \sigma \in K \tmin S_\Lambda
    \end{align}
    by definition and $\sum_{a=1}^k \tilde \sigma_{a|x} = \bar \sigma$ as $\sum_{a=1}^k m_a^{(x)} = \mathds 1_{CS^1_{k,g}}$ for all $x \in [g]$. We can furthermore decompose 
    \begin{equation} \label{eq:Bell-1}
        \tilde \sigma_{a|x} = \sum_{\lambda \in [\Lambda]} \tilde \varrho_{a, \lambda|x} \otimes \delta_\lambda
    \end{equation}
    and
    \begin{equation} \label{eq:Bell-2}
    \bar \sigma = \sum_{\lambda = 1}^\Lambda w_\lambda \varrho_\lambda  \otimes \delta_\lambda
    \end{equation}
    We can set $\bar \varrho := (\id \otimes \one_{S_\Lambda})\bar \sigma$ and $\tilde q_{a, \lambda|x} = \one_K(\tilde \varrho_{a,\lambda|x})$ such that $\tilde \varrho_{a, \lambda|x} = \tilde q_{a, \lambda|x} \varrho_{a|x, \lambda}$ with $\varrho_{a|x, \lambda} \in K$ for all $a \in [k]$, $\lambda \in [\Lambda]$ and $x \in [g]$. The equation $\sum_{a = 1}^k \tilde q_{a, \lambda|x} \varrho_{a |x, \lambda} = w_\lambda \bar \varrho_\lambda$ can be easily inferred from a comparison of \cref{eq:Bell-1} and \cref{eq:Bell-2}, which also implies that the $(\tilde q_{a, \lambda|x})_{a\in[k], \lambda \in [\Lambda]}$ are conditional probability distributions. We can define a probability distribution
    \begin{equation}
        q_{a|x, \lambda} = \begin{cases}
            q_{a, \lambda|x}/w_\lambda & w_\lambda \neq 0 \\
            1/k & w_\lambda = 0
        \end{cases} \,.
    \end{equation}
    Note that for $w_\lambda = 0$ we can just choose $\bar \varrho_\lambda$ as we please in \cref{eq:Bell-2}, thus enforcing the condition $\sum_{a = 1}^k q_{a|x, \lambda} \varrho_{a|x, \lambda} = \bar \varrho_\lambda$ also in this case.
    The assertion follows then from the relations in \Cref{cor:Phi-NS}, as $\Psi$ can easily defined from an assemblage in the way shown above.
\end{proof}

\begin{remark} \label{rem:Bell-behaviors}
    The roles of $CS_{l, r}$ and $CS_{k,g}$ in \Cref{thm:Bell-compression} can be reversed. In quantum theory, we can find a bipartite state and a multimeter such that $q_{a|x, \lambda} \, \varrho_{a|x, \lambda} = \Tr_A[(E_{a|x, \lambda} \otimes \mathds 1_B) \varrho_{AB}]$. Thus, \cref{eq:Bell-compression} can be written as
    \begin{equation}
        p_{a,b|x,y} = \sum_{\lambda =1}^\Lambda w_\lambda\Tr[(E_{a|x, \lambda} \otimes F_{b|y, \lambda}) \varrho_{AB}] \qquad \forall a \in [k], \forall b \in [l], \forall x \in [g], \forall y \in [r] \, .
    \end{equation}
\end{remark}

Similarly we can consider classical simulability of no-signaling probability distributions and characterize this in terms of a particular factorization. We note that since we want to consider classically simulating no-signaling distributions with no-signaling distributions, we cannot simply consider any joint classical simulations $\Phi: CS^1_{k_1 \cdot l_1,g_1 \cdot h_1} \to CS^1_{k_2 \cdot l_2,g_2 \cdot h_2}$ since they do not necessarily preserve the no-signaling property. Instead, we consider a separate classical simulation process in which we classically simulate the first input-output pair and the second input-output pair separately. It turns out that this kind of classical simulation results in more of a sequential factorization process. This is formalized in the following result:

\begin{thm} \label{thm:Bell-classical-simulation}
    Let $\Phi: (CS_{k,g}^+)^\ast \to CS_{l, r}^+$ with $\Phi(\one_{CS_{k,g}^1}) \in CS_{l,r}^1$. Then, the no-signaling probability distributions defining $\Phi$ (see \Cref{cor:Phi-NS}) satisfy
    \begin{equation}
        p_{a,b|x,y} = \sum_{w=1}^u \pi_{w|x} \sum_{c=1}^s \nu_{a|c,x,w}  \left[ \sum_{z=1}^v \theta_{z|y} \sum_{d=1}^t \mu_{b|d,y,z} q_{c,d|w,z}\right] \, .
    \end{equation}
    where $(\pi_{w|x})_{w \in [u]}$, $(\theta_{z|y})_{z \in [v]}$, $(\nu_{a|c,x,w})_{a \in [k]}$ and $(\mu_{b|d,y,z})_{b \in [l]}$ are conditional probability distributions and $(q_{c, d|w,z})_{c \in [s],d \in [t]}$ are no-signaling probability distributions if and only if there exist maps  $\Psi_1:(CS^\ast_{k,g})^+ \to (CS^\ast_{s,u})^+$ with $\Psi_1(\one_{CS^1_{k,g}})=\one_{CS^1_{s,u}}$, $\Psi_2:(CS^\ast_{s,u})^+ \to CS^+_{t, v}$ with $\Psi_2(\one_{CS^1_{s,u}}) \in CS^1_{t,v}$, and a channel $\Psi_3: CS^1_{t, v} \to CS^1_{l, r}$ such that the following diagram commutes:
     \begin{equation}
        \begin{tikzcd}
        (CS^\ast_{k,g})^+ \arrow[rd, "\Psi_1"] \arrow[rrr, "\Phi"]  & & & CS^+_{l, r}  \\
        & (CS^\ast_{s,u})^+  \arrow[r, "\Psi_2"]&CS^+_{t, v} \arrow[ru, "\Psi_3"]&
        \end{tikzcd}
\end{equation}
\end{thm}
\begin{proof}
    Let us define $\langle m_d^{(z)}, \Psi_2(m^{(w)}_c) \rangle =: q_{c,d|w,z}$ for all $c \in [s]$, $d \in [t]$, $w \in [u]$, $z \in [v]$ using \Cref{cor:Phi-NS}. We will use \Cref{thm:channel-polysimplices-characterization} again to obtain for any $\alpha \in (CS^\ast_{k,g})^+$ 
    \begin{equation}
        \Psi_1(\alpha) = \alpha(s_{k, \ldots,k})\one_{CS^1_{s,u}}  + \sum_{x=1}^g \sum_{a =1}^{k-1} \alpha(e^{(x)}_a)\left[ \sum_{w=1}^u \pi_{w|x} \sum_{c=1}^s \nu_{a|c,x,w} m^{(w)}_c\right]
    \end{equation}
    and for any $X \in CS_{t, v}$
    \begin{equation}
    \Psi_3(X) = \one_{CS^1_{t,v}}(X) s_{l, \ldots,l} + \sum_{y=1}^r \sum_{b =1}^{l-1} \left[ \sum_{z=1}^v \theta_{z|y} \sum_{d=1}^t \mu_{b|d,y,z} m^{(z)}_d(X)\right] e^{(y)}_b \, .
\end{equation}
    Hence, we infer for $a \in [k-1]$, using \Cref{cor:Phi-NS},
    \begin{align}
        \Psi_3(\Psi_2( \Psi_1(m_a^{(x)})))=&\sum_{w=1}^u \pi_{w|x} \sum_{c=1}^s \nu_{a|c,x,w} \one_{CS^1_{t,v}}(\Psi_2(m^{(w)}_c))s_{l, \ldots,l} \\&+ \sum_{w=1}^u \pi_{w|x} \sum_{c=1}^s \nu_{a|c,x,w} \sum_{y=1}^r \sum_{b =1}^{l-1} \left[ \sum_{z=1}^v \theta_{z|y} \sum_{d=1}^t \mu_{b|d,y,z} m^{(z)}_d(\Psi_2(m^{(w)}_c))\right] e^{(y)}_b \\
        =&\sum_{w=1}^u \pi_{w|x} \sum_{c=1}^s \nu_{a|c,x,w}\sum_{d=1}^t q_{c,d|w,1} s_{l, \ldots,l} \\&+ \sum_{w=1}^u \pi_{w|x} \sum_{c=1}^s \nu_{a|c,x,w} \sum_{y=1}^r \sum_{b =1}^{l-1} \left[ \sum_{z=1}^v \theta_{z|y} \sum_{d=1}^t \mu_{b|d,y,z} q_{c,d|w,z}\right] e^{(y)}_b
    \end{align}
    Moreover,
    \begin{align}
        \Psi_3(\Psi_2( \Psi_1(\one_{CS^1_{k,g}})))&= \one_{CS^1_{t,v}}(\Psi_2(\one_{CS_{s,u}^1}))s_{l, \ldots, l}+\sum_{y=1}^r \sum_{b =1}^{l-1} \left[ \sum_{z=1}^v \theta_{z|y} \sum_{d=1}^t \mu_{b|d,y,z} m^{(z)}_d(\Psi_2(\one_{CS_{s,u}^1}))\right] e^{(y)}_b\\&= s_{l, \ldots, l}+\sum_{y=1}^r \sum_{b =1}^{l-1} \left[ \sum_{z=1}^v \theta_{z|y} \sum_{d=1}^t \mu_{b|d,y,z} \sum_{c=1}^s q_{c,d|1,z}\right] e^{(y)}_b \, .
    \end{align}
    Thus, for $b \in [l-1]$ and $a \in [k-1]$,
    \begin{equation}
        p_{a,b|x,y} = \sum_{w=1}^u \pi_{w|x} \sum_{c=1}^s \nu_{a|c,x,w}  \left[ \sum_{z=1}^v \theta_{z|y} \sum_{d=1}^t \mu_{b|d,y,z} q_{c,d|w,z}\right] \, .
    \end{equation}
    For $b = l$ and $a \in [k-1]$, we have using the no-signaling of $q_{c,d|w,z}$
    \begin{align}
        p_{a,l|x,y} =& \sum_{w=1}^u \pi_{w|x} \sum_{c=1}^s \nu_{a|c,x,w}\sum_{d=1}^t q_{c,d|w,1} \\&- \sum_{w=1}^u \pi_{w|x} \sum_{c=1}^s \nu_{a|c,x,w} \sum_{b =1}^{l-1} \left[ \sum_{z=1}^v \theta_{z|y} \sum_{d=1}^t \mu_{b|d,y,z} q_{c,d|w,z}\right] \\
        =& \sum_{w=1}^u \pi_{w|x} \sum_{c=1}^s \nu_{a|c,x,w} \sum_{z=1}^v \theta_{z|y} \sum_{d=1}^t q_{c,d|w,z} \underbrace{\left(1-\sum_{b =1}^{l-1}\mu_{b|d,y,z}\right)}_{\mu_{l|d,y,z}}
        \, .
    \end{align}
    For $b \in [l-1]$ and $a=k$, it follows that, using the no-signaling of $q_{c,d|w,z}$,
    \begin{align}
        p_{k,b|x,y} =& m^{(y)}_b[\Psi_3(\Psi_2( \Psi_1(\one_{CS^1_{k,g}})))] - \sum_{a =1}^{k-1}m^{(y)}_b[\Psi_3(\Psi_2( \Psi_1(m^{(x)}_a)))] \\ =&  \sum_{z=1}^v \theta_{z|y} \sum_{d=1}^t \mu_{b|d,y,z} \sum_{c=1}^s q_{c,d|1,z}  - \sum_{a =1}^{k-1}\sum_{w=1}^u \pi_{w|x} \sum_{c=1}^s \nu_{a|c,x,w}  \left[ \sum_{z=1}^v \theta_{z|y} \sum_{d=1}^t \mu_{b|d,y,z} q_{c,d|w,z}\right] \\ =&\sum_{z=1}^v \theta_{z|y} \sum_{d=1}^t \mu_{b|d,y,z} \sum_{w=1}^u \pi_{w|x} \sum_{c=1}^s q_{c,d|w,z} \underbrace{\left(1-\sum_{a =1}^{k-1} \nu_{a|c,x,w}\right)}_{\nu_{k|c,x,w}} \, .
    \end{align}
    Finally, we find, using again the no-signaling of $q_{c,d|w,z}$,
    \begin{align}
        p_{k,l|x,y} =& \left(\one_{CS_{l, r}^1} - \sum_{b=1}^{h-1} m_b^{(y)}\right) \left[\Psi_3(\Psi_2( \Psi_1(\one_{CS^1_{k,g}}))) - \sum_{a =1}^{k-1}\Psi_3(\Psi_2( \Psi_1(m^{(x)}_a)))\right] \\
        =& 1- \sum_{b=1}^{h-1} \sum_{z=1}^v \theta_{z|y} \sum_{d=1}^t \mu_{b|d,y,z} \sum_{w=1}^u \pi_{w|x} \sum_{c=1}^s q_{c,d|w,z} \nu_{k|c,x,w}\\& - \sum_{a=1}^{k-1}\sum_{w=1}^u \pi_{w|x} \sum_{c=1}^s \nu_{a|c,x,w}\underbrace{\sum_{d=1}^t q_{c,d|w,1}}_{\sum_{z=1}^v \theta_{z|y}\sum_{d=1}^t q_{c,d|w,z}} \underbrace{1}_{\sum_{b=1}^l \mu_{b|d,y,z}} \\
        =& \sum_{z=1}^v \theta_{z|y} \sum_{d=1}^t \mu_{l|d,y,z} \sum_{w=1}^u \pi_{w|x} \sum_{c=1}^s q_{c,d|w,z} \nu_{k|c,x,w} \,.
    \end{align}
    The converse is straightforward as we can define the maps $\Psi_1$, $\Psi_2$ and $\Psi_3$ via the conditional and no-signaling probabilities.
    \end{proof}

    Finally, as with multimeters and assemblages, for no-signaling probability distributions we can equate the special case of factoring through a simplex and being separable with having a classical realization, i.e., having an LHV model. 

    \begin{cor} \label{cor:separability-LHV-simulation}
        Let $\Phi: (CS^\ast_{k,g})^+ \to CS_{l,r}^+$ with $\Phi(\one_{CS^1_{k,g}}) \in CS_{l,r}^1$. Then, the following are equivalent:
        \begin{enumerate}
            \item $(\Phi \otimes \id)(\chi_{CS_{k,g}}) \in CS_{l,r}^1 \tmin CS_{k,g}^1$.
            \item The associated no-signaling probability distributions $(p_{a,b|x,y})_{a \in [k], b \in [l]}$ by \Cref{cor:Phi-NS} admit an LHV model.
            \item There exist a finite number of outcomes $\Lambda$, a map $\Psi_1: (CS^\ast_{k,g})^+ \to S_\Lambda^+$ such that $\Psi(\one_{CS^1_{k,g}}) \in S_\Lambda$, and a channel $\Psi_2: S_\Lambda \to CS_{l,r}^1$ such that the following diagram commutes:
        \begin{equation}
            \begin{tikzcd}
            (CS^\ast_{k,g})^+ \arrow[rd, "\Psi_1"] \arrow[rr, "\Phi"] & &  CS_{l,r}^+ \\
            & S_{\Lambda}^+ \arrow[ru, "\Psi_2"] &
            \end{tikzcd}  
        \end{equation}
        \end{enumerate}
    \end{cor}
    \begin{proof}
        The equivalence of $(1)$ and $(3)$ follows from \Cref{prop:min-factoring-and-tensor}. The equivalence of $(2)$ and $(3)$ follows from \Cref{thm:Bell-compression} with $K=S_1$ as follows: We obtain that $(3)$ is equivalent to 
        \begin{equation}
        p_{a,b|x,y} = \sum_{\lambda=1}^\Lambda w_\lambda q_{a|x, \lambda} f_{b|y, \lambda}  \qquad \forall a \in [k], \forall b \in [l], \forall x \in [g], \forall y \in [r] \, ,
    \end{equation}
    where $(f_{b|y, \lambda})_{b \in [l]}$ is a conditional probability distribution and $\varrho_{a|x, \lambda}=1$ for all $a \in [k]$, $x \in [g]$, $\lambda \in [\Lambda]$ in this case. This gives an LHV model for $(p_{a,b|x,y})_{a\in[k],b \in [l], x\in [g], y \in [r]}$.
    \end{proof}

\section{Separability and extendability} \label{sec:extendability}

The notion of \emph{extendability} has been recognized to be central in the context of quantum entanglement \cite{doherty_complete_2004}. In this section relate the extendability of relevant tensors to intermediate notion of classicality in the setting of measurement compatibility (\cref{sec:ext-comp}) and existence of LHV models for correlations (\cref{sec:ext-LHV}).
In order to introduce the concept of extendability, we recall a known instance from quantum theory.
Consider a separable bipartite quantum state $\varrho_{AB} = \sum_\lambda p(\lambda) \, \varrho_{\lambda}^{(A)} \otimes \varrho_{\lambda}^{(B)}$.
We can construct a possible extension of this state with finitely many parties as follows:
\begin{equation} \label{eq:extendability-symemtricExtensionRho}
    \sigma_{A B_1 B_2 \ldots B_n} = \sum_\lambda p(\lambda) \, \varrho_{\lambda}^{(A)} \otimes \Big(\varrho_{\lambda}^{(B)}\Big)^{\otimes n}\, ,
\end{equation}
such that
\begin{enumerate}
    \item $\Tr_{B_2 \ldots B_n}[\sigma_{AB_1 B_2 \ldots B_n}] = \varrho_{AB}$
    \item $\pi_{B_1 B_2 \ldots B_n}(\sigma_{AB_1 B_2 \ldots B_n}) = \sigma_{AB_1 B_2 \ldots B_n}$, where $\pi$ permutes $B_1, \dots, B_n$
\end{enumerate}
In \cite{doherty_complete_2004} it was found that if a quantum state
$\sigma_{AB_1 B_2 \ldots B_n}$ having properties (1) and (2) exists for all $n \geq 2$, then $\varrho_{AB}$ is separable.
On the other hand, if there is some $n \geq 2$ with no state $\sigma_{AB_1 B_2 \ldots B_n}$ satisfying above properties, then $\varrho_{AB}$ must be entangled.

We will first of all argue that factorizability can be solved using the generalization of the aforementioned methods of symmetric extensions. Consider the case where we are given state spaces $K_A$, $K_B$ and a map $\Phi: K_A \to K_B$. We want to determine whether given the state space $K_C$ there are maps $\Psi_1: K_A \to K_C \dot{\otimes} S_\Lambda$ and $\Psi_2: K_C \dot{\otimes} S_\Lambda \to K_B$ such that the following diagram commutes:
\begin{equation} \label{eq:extendability-generalWithSL}
    \begin{tikzcd}
    V(K_A)^+ \arrow[rd, "\Psi_1"] \arrow[rr, "\Phi"] & & V(K_B)^+ \\
    & V(K_C \dot{\otimes} S_\Lambda)^+ \arrow[ru, "\Psi_2"] &
    \end{tikzcd}  
\end{equation}
Here, as before, $S_\Lambda$ represents a classical side information. For every value $\lambda \in \Lambda$ we define the partial maps $\Psi_{1, \lambda}: V(K_A)^+ \to V(K_C)^+$ and $\Psi_{2, \lambda}: V(K_C)^+ \to V(K_B)^+$ by fixing the value of $\lambda$ in the output or input respectively. Then, commutativity of the diagram \eqref{eq:extendability-generalWithSL} is equivalent to
\begin{equation}
    \Phi = \sum_\lambda \Psi_{2, \lambda} \circ \Psi_{1, \lambda}.
\end{equation}
Defining the evaluation functional $\chi_{V(K_C)}$ on $V(K_C) \otimes A(K_C)$ via
\begin{equation}
    \chi_{V(K_C)} (x \otimes f) = f(x)
\end{equation}
for $x \in K_C$, $f \in A(K_C)$, we get that the commutativity of \eqref{eq:extendability-generalWithSL} is equivalent to
\begin{equation} \label{eq:extendability-PhiSeparability}
    \Phi = (\id \otimes \chi_{V(K_C)} \otimes \id) \left(\sum_\lambda \Psi_{1, \lambda} \otimes \Psi_{2, \lambda} \right).
\end{equation}
It thus follows that commutativity of \eqref{eq:extendability-generalWithSL} is equivalent to the existence of a separable tensor $\sum_\lambda \Psi_{2, \lambda} \otimes \Psi_{1, \lambda}$ such that \eqref{eq:extendability-PhiSeparability} holds. Thus one can use the hierarchy of symmetric extensions \cite{monogamy} that generalizes the idea of \eqref{eq:extendability-symemtricExtensionRho} to rule out the existence of suitable maps $\Psi_{1, \lambda}$ and $\Psi_{2, \lambda}$. We will investigate this in detail in the case of joint measurability bellow, but before that we have to address the case without classical side information. 

Consider again state spaces $K_A$, $K_B$ and a map $\Phi: K_A \to K_B$. We want to determine whether given the state space $K_C$ there are maps $\Psi_1: K_A \to K_C$ and $\Psi_2: K_C \to K_B$ such that the following diagram commutes:
\begin{equation} \label{eq:extendability-generalNoSL}
    \begin{tikzcd}
    V(K_A)^+ \arrow[rd, "\Psi_1"] \arrow[rr, "\Phi"] & & V(K_B)^+ \\
    & V(K_C)^+ \arrow[ru, "\Psi_2"] &
    \end{tikzcd}  
\end{equation}
Repeating the steps as above, we get that this is equivalent to
\begin{equation}
    \Phi = (\id \otimes \chi_{V(K_C)} \otimes \id) (\Psi_1 \otimes \Psi_2).
\end{equation}
In this case $\Psi_1 \otimes \Psi_2$ is a product tensor rather than a separable tensor and one clearly cannot just replace $\Psi_1 \otimes \Psi_2$ with an arbitrary separable tensor $\sum_\lambda \Psi_{2, \lambda} \otimes \Psi_{1, \lambda}$ as this would change the definition of the problem. Nevertheless even in this case one can construct a hierarchy of symmetric extensions that could be used to rule out the existence of $\Psi_1: K_A \to K_C$ and $\Psi_2: K_C \to K_B$ such that \eqref{eq:extendability-generalNoSL} commutes using the methods outlined in \cite{plavala2025polarization}.

\subsection{Extendability of measurements}\label{sec:ext-comp}

Let us now investigate the case of extendability and joint measurability in depth. We will start by recalling the first two theorems from \cite{monogamy}. Let us consider proper cones $V_A^+$ and $V_B^+$ in vector spaces $V_A$ and $V_B$. First we define a reduction map of $n$-th order:
\begin{equation}
    \gamma_n^\Phi: V_B^{\otimes n} \rightarrow V_B, \quad \gamma_n^\Phi= \frac{1}{n} \sum_{i=1}^n \Phi^{\otimes i - 1} \otimes \id \otimes \Phi^{\otimes n-i},
\end{equation}
where $\Phi$ is a linear form in the interior of $(V_B^+)^*$, e.g., the trace function $\Tr$ in the case of quantum theory. We make the following definition:
\begin{defi}
    Let $\Phi$ be a linear form in the interior of $(V_B^+)^*$. A tensor $\xi \in V^+_A \tmax V^+_B$ is $n$-extendable with respect to $\Phi$ if there exists $\xi^{(n)} \in V^+_A \tmax (V^+_B)^{\otimes n}$ such that $(id \otimes \gamma^\Phi_n)(\xi^{(n)}) = \xi$.
\end{defi}

As shown in \cite{monogamy}, if a tensor is $n$-extendable then it is also $n-1$ extendable so that we get a hierarchy. In particular, Theorem 1 of \cite{monogamy} states the following:
\begin{lem}[\cite{monogamy}]
    A tensor $\xi \in V^+_A \tmax V^+_B$ is separable, i.e., $\xi \in V^+_A \tmin V^+_B$, if and only if it is $n$-extendable for all $n \in \nat$ with respect to some linear form $\Phi$ in the interior of $(V_B^+)^*$, i.e.,
    \begin{equation}
\label{monogamy_thm1}
    \acone \tmin \bcone = \bigcap_{n \geq 1} \left(\id \otimes \gamma_n^\Phi \right) \left(\acone \tmax (\bcone)^{\tmax n} \right)
\end{equation}
\end{lem}

Thus, the consequence is similar as in quantum theory, we get a separable vector via reducing all finite extensions of it. The procedure in \cref{monogamy_thm1} can be thought as cutting the set of all vectors closer to the subset of separable vectors with increasing $n$, eventually converging. Although in practice this is not very applicable unless one wants to check a small set of $n$-extensions where the reduction fails in order to find an entangled vector; otherwise one would need to check infinitely many $n$ which is naturally not feasible.

In special cases, however, \cite[Theorem 2]{monogamy} provides a remedy. Under a certain condition there is a threshold, i.e., one must only check finitely many extensions up to some fixed $n$.
\begin{lem}[\cite{monogamy}]
\label{monogamy_thm2}    
    For $n \geq 1$ the following are equivalent:
    \begin{enumerate}
        \item[(i)] $\acone \tmin \bcone = \left(\id \otimes \gamma_n^\Phi \right) \left(\acone \tmax (\bcone)^{\tmax n} \right)$
        \item[(ii)] The base $K_\Phi = \bcone \cap \Phi^{-1}(1)$ is affinely equivalent to a Cartesian product of at most $n$ simplices, that is, a polysimplex with at most $n$ simplices.
    \end{enumerate}    
\end{lem}

We want to apply this result to multimeters in order to find a criterion for compatibility. To do this, let $M = \{M_{\cdot|x}\}_{x \in [g]}$ be a multimeter that we identify with the tensor
$\xi_M \in A(K)^+ \tmax \CSC$. From now on, we fix the reduction map $\Phi$ as the equal sums of the columns, that is, $\Phi(C) = \one_{\CS}(C) = \sum_{i=1}^k C_{ij}$ for $j \in [g]$ and $C \in \CSC$, and when we consider $n$-extendability we mean $n$-extendability with respect to this choice of $\Phi$. Now the base $K_\Phi = \CSC \cap \Phi^{-1}(1)$ is the set of column stochastic matrices $\CS$. As already discussed, the state space $\CS$ is a special case of a polysimplex and in particular it is affinely isomorphic to a Cartesian product of $g$ simplices $S_k$.

Thus, by combining \cref{prop:min-factoring-and-tensor} and \cref{monogamy_thm2} we get the following characterization for compatibility:

\begin{cor}
    A multimeter $M: K \to CS^1_{k,g}$ is compatible if and only $\xi_M \in A(K)^+ \tmax \CSC$ is $g$-extendable.
\end{cor}

The next question that arises is can we conclude something about $M$ by considering its $n$-extensions for $n < g$? It turns out that we can and that it is related to the concept of $n$-wise compatibility.
\begin{defi}
    A multimeter $M = \{M_{\cdot|x}\}_{x \in [g]}$ is said to be $n$-wise compatible with $n \leq g$ if for all $\{x_1, \dots, x_n\} \subset [g]$ the measurements $\{M_{\cdot|x_i}\}_{i \in [n]}$ are compatible.
\end{defi}
Here, the choice of $n$ measurements from $M$ is arbitrary, i.e., if $n < g$, every subset of $n$ measurements from $M$ is compatible. Of course, if $n = g$ then $M$ is (fully) compatible. The key to connecting $n$-extendability to $n$-wise compatibility is to see that applying the reduction map $\gamma_n^{\Phi}$ to an $n$-extension $\xi^{(n)}$ of $\xi_M$ corresponds to obtaining $M$ as a marginal of the map corresponding to $\xi^{(n)}$.

In order to show the connection, we first need some structural results on the extensions. The extensions of $\xi_M \in A(K)^+ \tmax \CSC$ are tensors $A(K)^+ \tmax (\CSC)^{\tmax n}$. Equivalently we can consider the extensions as positive maps from $V(K)^+$ to $(\CSC)^{\tmax n}$. In fact, if $M: K \to CS^1_{k,g}$ is a multimeter, and $\xi^{(n)} \in A(K)^+ \tmax (\CSC)^{\tmax n}$ is its $n$-extension, then actually the corresponding map $N: V(K)^+ \to (\CSC)^{\tmax n}$ is a channel (we will prove this in \cref{thm:extendability-compatibility}). Thus, we will look into properties of normalized elements in $(\CSC)^{\tmax n}$ similarly as we did in \cref{prop:NS-tensor}.

Recall that the symmetric group $\mathfrak{S}_n$ is the set of permutations of the set $[n]$, i.e., the set of bijective functions $\pi: [n] \to [n]$. For a permutation $\pi \in \mathfrak{S}_n$, we denote the \emph{descent set} of $\pi$ by $\mathrm{desc}(\pi)$ as the set of indices $i \in [n-1]$ such that $\pi(i) > \pi(i+1)$. Note that for the identity permutation $e$ we have $\mathrm{desc}(e) = \emptyset$ and it is the only permutation with this property. For $m \in [n-1]$, we denote $\Pi_m := \{\pi \in \mathfrak{S}_n \, : \, \mathrm{desc}(\pi) \subseteq \{m\}\}$, i.e., it is the set of permutations with at most one descent at index $m$. Equivalently, $\Pi_m = \{e\} \cup \{\pi \in \mathfrak{S}_n \, : \, \mathrm{desc}(\pi) = \{m\}\}$. We note that $|\Pi_m| = {n \choose m}$ as it characterizes in how many ways we can pick the elements $\pi(1) < \cdots < \pi(m)$ from the set $[n]$.
We will also make use of the fact that for all $g \in \nat$ the symmetric group $\mathfrak{S}_n$ acts naturally on $[g]^n$ by the linear map $\alpha_\pi: [g]^n \to [g]^n$ defined as $\alpha_\pi(y_1, \ldots,y_n) = (y_{\pi(1)}, \ldots, y_{\pi(n)})$ for all $(y_1, \ldots, y_n) \in [g]^n$ for all $\pi \in \mathfrak{S}_n$. Similarly $\mathfrak{S}_n$ also acts on $(CS_{k,g})^{\otimes n}$ by the linear map $\mathcal{U}_\pi: (CS_{k,g})^{\otimes n} \to (CS_{k,g})^{\otimes n} $ defined as $\mathcal{U}_\pi(X_1 \otimes \cdots \otimes X_n) = X_{\pi(1)} \otimes \cdots \otimes X_{\pi(n)}$ for all $X_i \in CS_{k,g}$ for all $\pi \in \mathfrak{S}_n$.

\begin{prop} \label{prop:extended-NS-tensor}
    $\xi \in (\CS)^{\tmax n}$ with $2 \leq n \in \mathds{N}$ if and only if there exists a no-signalling conditional probability distribution $p =(p_{\cdot,\ldots,\cdot|x_1,\ldots,x_n})_{x_1,\ldots,x_n \in [g]}$ on $[k]^n$ such that
    \begin{equation}\label{eq:xi-tensor-no-signaling}
        \begin{split}
            &\xi = (s_{k,\ldots,k})^{\otimes n} + \sum_{x_1 = 1}^{g} \cdots \sum_{x_n = 1}^g \sum_{a_1 = 1}^{k-1} \cdots \sum_{a_n = 1}^{k-1} p_{a_1,\ldots,a_n|x_1,\ldots,x_n} \bigotimes_{i=1}^n e_{a_i}^{(x_i)} \\
            & \quad +\sum_{m=1}^{n-1} \sum_{\pi \in \Pi_m} \sum_{x_{\pi(1)} = 1}^{g} \cdots \sum_{x_{\pi(m)} = 1}^g  \sum_{a_{\pi(1)} = 1}^{k-1} \cdots \sum_{a_{\pi(m)} = 1}^{k-1} \sum_{a_{\pi(m+1)} = 1}^{k} \cdots \sum_{a_{\pi(n)} = 1}^{k} p_{a_{1},\ldots,a_{n}|\alpha_\pi(x_{\pi(1)},\ldots,x_{\pi(m)},1, \ldots, 1)} \\
            & \quad \quad \times \mathcal{U}_{\pi}\left(e_{a_{\pi(1)}}^{(x_{\pi(1)})} \otimes \cdots \otimes e_{a_{\pi(m)}}^{(x_{\pi(m)})} \otimes (s_{k,\ldots,k})^{\otimes n-m}\right) \, .
        \end{split}
    \end{equation}
\end{prop}
\begin{proof}
    We start by writing a given tensor $\xi \in (\CS)^{\tmax n}$ in the tensor basis of $CS_{k,g}$:
        \begin{equation}\label{eq:xi-tensor}
        \begin{split}
            &\xi = \eta (s_{k,\ldots,k})^{\otimes n} + \sum_{x_1 = 1}^{g} \cdots \sum_{x_n = 1}^g \sum_{a_1 = 1}^{k-1} \cdots \sum_{a_n = 1}^{k-1} N_{a_1,\ldots,a_n|x_1,\ldots,x_n} \bigotimes_{i=1}^n e_{a_i}^{(x_i)} +\sum_{m=1}^{n-1} \sum_{\pi \in \Pi_m} \sum_{x_{\pi(1)} = 1}^{g} \cdots \sum_{x_{\pi(m)} = 1}^g \\
            &\sum_{a_{\pi(1)} = 1}^{k-1} \cdots \sum_{a_{\pi(m)} = 1}^{k-1} N^{(\pi)}_{a_{\pi(1)},\ldots,a_{\pi(m)}|x_{\pi(1)},\ldots,x_{\pi(m)}} \, \mathcal{U}_{\pi}\left(e_{a_{\pi(1)}}^{(x_{\pi(1)})} \otimes \cdots \otimes e_{a_{\pi(m)}}^{(x_{\pi(m)})} \otimes (s_{k,\ldots,k})^{\otimes n-m}\right) \, .
        \end{split}
    \end{equation}
    To see that we get all the basis elements exactly once by applying the permutations in $\Pi_m$ for all $m \in [n-1]$ we argue the following: Let us fix $m \in [n-1]$. Let us consider the tensor product as a list of length $n$ and let us consider filling it with the $m$ terms of $e^{(x_1)}_{a_1}, \ldots, e^{(x_m)}_{a_m}$ and the $m-n$ copies of $s_{k,\ldots,k}$. (Note that in \cref{eq:xi-tensor} we also relabel the indices $i$ of $x_i$ and $a_i$ but just to get all the basis elements we do not need to consider that). In general there are exactly $n!$ ways to do this, but in exactly $m!$ ways the $e^{(x_i)}_{a_i}$ terms are in the same $m$ positions, just differently permuted. We only need one of these $m!$ products and so we select the products where the relative order of the terms $e^{(x_i)}_{a_i}$ remains the same, i.e., where  $e^{(x_j)}_{a_j}$ appears before $e^{(x_{j+1})}_{a_{j+1}}$. Thus, there are now $\frac{m!}{n!}$ of these products. To actually get these products, let us start from the product $e_{a_{1}}^{(x_{1})} \otimes \cdots \otimes e_{a_{m}}^{(x_{m})} \otimes (s_{k,\ldots,k})^{\otimes n-m}$. Since we want to preserve the relative order of the $e^{(x_i)}_{a_i}$'s we can only apply permutations $\pi \in \mathfrak{S}_n$ which have an ascending run of $[m]$, i.e., the set of indices $i$ where $\pi(i) > \pi(i+1)$ is $[m]$. By taking all such permutations we get all $\frac{n!}{m!}$ products that we wanted. However, in these products there are still multiplicities of the basis elements where the $s_{k,\ldots,k}$ terms are permuted among themselves. Since there are $(m-n)!$ ways to permute the $s_{k,\ldots,k}$ terms, in the end we want to end up with ${n \choose m} = \frac{n!}{(n-m)!m!}$ products corresponding to the basis elements. One way to achieve this is to take only those permutations which also preserve the relative ordering of the $s_{k,\ldots,k}$, i.e., those permutations which have an ascending run of $\{m+1, \ldots, n\}$. Thus, in the end we are left with permutations with one ascending run $[n]$, which corresponds to the identity permutation, or two ascending runs $[m]$ and $[n] \setminus [m]$. This is equivalent to saying that we have permutations with either zero descends (the identity permutation) or exactly one descent at position $m$. But this is exactly how we defined $\Pi_m$ so in the end we thus get all the basis elements exactly once for all $m \in [n-1]$.
    
    Since $\xi$ is a state, it is clear that $\langle \one_{\CS}^{\otimes n}, \xi \rangle = 1 = \eta$. We make use of the positivity of the measurements $m_{a}^{(x)}$ for all $a \in [k]$ and $x \in [g]$ to define the conditional probability distribution $p =(p_{\cdot,\ldots,\cdot|x_1,\ldots,x_n})_{x_1,\ldots,x_n \in [g]}$ on $[k]^n$ as follows:
    \begin{equation}
        p_{a_1,\ldots,a_n|x_1,\ldots,x_n} = \langle \bigotimes_{i=1}^n m_{a_i}^{(x_i)}, \xi \rangle \quad \forall a_i \in [k], x_i \in [g], i \in [n]
    \end{equation}
    Clearly, since $\sum_{a=1}^k m^{(x)}_a = \one_{CS^1_{k,g}}$ for all $x \in [g]$, we have:
    \begin{equation}
        \sum_{a_1=1}^{k} \cdots \sum_{a_n=1}^{k} p_{a_1,\ldots,a_n|x_1,\ldots,x_n}=\langle \sum_{a_1=1}^{k} \cdots \sum_{a_n=1}^{k} m_{a_1}^{(x_1)} \otimes \cdots \otimes m_{a_n}^{(x_n)}, \xi\rangle = \langle \one_{\CS}^{\otimes n}, \xi\rangle = 1
    \end{equation}
    so that $p =(p_{\cdot,\ldots,\cdot|x_1,\ldots,x_n})_{x_1,\ldots,x_n \in [g]}$ is indeed a conditional probability distribution on $[k]^n$ for all $x_1, \ldots, x_n \in [g]$.

    What remains to see is that $p$ is no-signaling and that the decompositions in \cref{eq:xi-tensor-no-signaling} and \cref{eq:xi-tensor} match. To see that $p$ is no-signaling, we fix some $\{i_1, \dots, i_v\} \subset [n]$ for some $v \in [n-1]$. Let $\sigma \in \mathfrak{S}_n$ be the unique permutation with at most one descent at index $l$ such that $\sigma(j) = i_j$ for all $j \in [l]$. Thus, $\sigma \in \Pi_l$. We fix some $a_{i_j} \in [k]$ and $x_{i_j} \in [g]$ for all $j \in [l]$. Now for all $x_i \in [g]$, $i \in [n]\setminus \{i_1, \ldots, i_l\} = \sigma([n] \setminus [l])$ we have that
    \begin{align}
                & \sum_{a_{\sigma(l+1)}=1}^{k}  \cdots \sum_{a_{\sigma(n)}=1}^{k} p_{a_1,\ldots,a_n|x_1,\ldots,x_n} \\
                &= \sum_{i \in \sigma([n] \setminus [l])} \sum_{a_i=1}^{k} \langle m^{(x_1)}_{a_1} \otimes \cdots \otimes m^{(x_{i_1})}_{a_{i_1}} \otimes \cdots \otimes m^{(x_{i_l})}_{a_{i_l}} \otimes\cdots \otimes m^{(x_{n})}_{a_{n}}, \xi \rangle \\
                &= \langle \one_{CS^1_{k,g}} \otimes \cdots \otimes  \one_{CS^1_{k,g}} \otimes  m^{(x_{\sigma(1)})}_{a_{\sigma(1)}} \otimes \one_{CS^1_{k,g}} \otimes \cdots \otimes \one_{CS^1_{k,g}} \otimes m^{(x_{\sigma(l)})}_{a_{\sigma(l)}} \otimes \one_{CS^1_{k,g}} \otimes \cdots \otimes \one_{CS^1_{k,g}}, \xi \rangle \\
                &= \left\langle \mathcal{U}^*_{\sigma^{-1}} \left(  m_{a_{\sigma(1)}}^{(x_{\sigma(1)})} \otimes \cdots \otimes m_{a_{\sigma(l)}}^{(x_{\sigma(l)})} \otimes \one_{\CS}^{\otimes n-l} \right), \xi \right\rangle,
    \end{align}
    where $\mathcal{U}^*_{\sigma^{-1}}: (CS_{k.g}^*)^{\otimes n} \to (CS_{k.g}^*)^{\otimes n}$ is the adjoint of $\mathcal{U}_{\sigma^{-1}}$ and arranges the tensor product of the vectors in $CS^*_{k,g}$ in the same order as $\mathcal{U}_{\sigma}$ arranges the tensor product of the vectors in $CS_{k,g}$. Clearly the result is independent of what we choose as $x_i \in [g]$ for $i \in \sigma([n] \setminus [l])$ so that the probability distribution is no-signaling. Thus, in this case we may choose $x_i =1$ for all $i \in \sigma([n] \setminus [l])$. Furthermore, if $a_{i_j} \in [k-1]$ for all $j \in [l]$, then we have that
    \begin{align}
                 \sum_{a_{\sigma(l+1)}=1}^{k}  \cdots \sum_{a_{\sigma(n)}=1}^{k} p_{a_1,\ldots,a_n|\alpha_\sigma(x_{\sigma(1)},\ldots,x_{\sigma(l)},1, \ldots,1)} &= \left\langle \mathcal{U}^*_{\sigma^{-1}} \left(  m_{a_{\sigma(1)}}^{(x_{\sigma(1)})} \otimes \cdots \otimes m_{a_{\sigma(l)}}^{(x_{\sigma(l)})} \otimes \one_{\CS}^{\otimes n-l} \right), \xi \right\rangle  \\
                &= N^{(\sigma)}_{a_{\sigma(1)}, \ldots, a_{\sigma(l)} | x_{\sigma(1)}, \ldots, x_{\sigma(l)}} \, ,
    \end{align}
    which gives us exactly the decomposition in \cref{eq:xi-tensor-no-signaling}.

    The converse follows as the $m^{(x)}_a$ form the extreme rays of $(CS^+_{g,k})^*$.
\end{proof}
\begin{remark}\label{remark:NS-extension}
    With the procedure in the above proof we also could have shown the same statement with $\xi \in CS_{k_1,g_1}^1 \tmax CS_{k_2,g_2}^1 \tmax \ldots \tmax CS_{k_n,g_n}^1$, i.e.\ , extending with arbitrary state cones. One just needs to label the indices accordingly and make sure the measurements act on the correct spaces. The challenge is that the notation becomes even more challenging, see \cref{prop:NS-tensor} for the $n=2$ case.
    However, if we want to make use of the extension as described in the beginning of this section, it is important to have copies of the same state cone.
\end{remark}
At this point we are able to link the extensions with joint measurements.

\begin{thm} \label{thm:extendability-compatibility}
    Let $M: K \rightarrow \CS$ be a multimeter and let $\xi_M \in A(K)^+ \tmax \CSC$ be the corresponding tensor. Then $M$ is $n$-extendable if and only if it is $n$-wise compatible such that for all measurements $\{x_1, \ldots, x_n\} \subset [g]$ there exists joint measurements $G^{(x_1, \ldots, x_n)} = \{G^{(x_1, \ldots, x_n)}_{a_1, \ldots,a_n}\}_{a_1, \ldots, a_n=1}^k$ which satisfy the following no-signaling constraints
    \begin{equation}\label{eq:joint-meas-ns}
        \sum_{a_i=1}^k G^{(x_1, \ldots,x_{i-1}, x_i, x_{i+1}, \ldots x_n)}_{a_1, \ldots,a_n} = \sum_{a_i=1}^k G^{(x_1, \ldots,x_{i-1}, \tilde{x}_i, x_{i+1}, \ldots x_n)}_{a_1, \ldots,a_n}
    \end{equation}
    for all $a_j \in [k]$, $x_j \in [g]$, $j \in [n]\setminus\{i\}$ and $x_i, \tilde{x}_i \in [g]$ for all $i \in [n]$.
\end{thm}
\begin{proof}
    Let us start with an $n$-extendable $\xi_M$. Then there exists $N^{(n)}:V(K)^+ \rightarrow (\CSC)^{\otimes n}$, that is, $\xi_{N^{(n)}} \in A(K)^+ \tmax (\CSC)^{\otimes n}$ such that $(\id \otimes \gamma^\Phi_n)(\xi_{N^{(n)}}) = \xi_M$ for $\Phi = \one_{\CS}$.
    Let us first see how we can get a given measurement $M_{\cdot|x}$ from the extension:
    \begin{align}
        M_{a|x}(\varrho) &= \langle M(\varrho), m_{a}^{(x)} \rangle = \langle \xi_M, \varrho \otimes m_{a}^{(x)}\rangle \\
        &= \langle (\id \otimes \gamma_n^{\Phi})(\xi_{N^{(n)}}), \varrho \otimes m_{a}^{(x)}\rangle = \langle \xi_{N^{(n)}}, \varrho \otimes (\gamma_n^\Phi)^*(m_{a}^{(x)}) \rangle \\
        &=\langle N^{(n)}(\varrho), (\gamma_n^\Phi)^*(m_{a}^{(x)}) \rangle  =  \langle \gamma^\Phi_n(N^{(n)}(\varrho)), m^{(x)}_{a} \rangle
    \end{align}
    for all $a \in [k]$, $x \in [g]$ and $\varrho \in K$. In particular we see that $M = \gamma^\Phi_n \circ N^{(n)}$.
    
    Now, for $\varrho \in K$ we see that $N^{(n)}(\varrho) \in (\CS)^{\tmax n}$ so that $N^{(n)}$ is a channel. This follows as $M$ is a channel for all $\varrho \in K$:
    \begin{equation}
        \langle \one_{\CS}^{\otimes n}, N^{(n)}(\varrho) \rangle= \langle (\gamma_n^{\Phi})^*(\one_{\CS}), N^{(n)}(\varrho) \rangle = \langle \one_{\CS}, \gamma_n^{\Phi}(N^{(n)}(\varrho))\rangle = \langle \one_{\CS}, M(\varrho) \rangle = \langle \one_K , \varrho \rangle = 1,
    \end{equation}
    where according to the results of \cite{monogamy} we can write $(\gamma^\Phi_n)^*(X) = P_{\mathrm{Sym}_n(CS^*_{k,g})}(X \otimes \Phi^{\otimes (n-1)})$ for all $X \in CS^*_{k,g}$, where 
    \begin{equation}
        P_{\mathrm{Sym}_n(CS^*_{k,g})} = \frac{1}{n!}\sum_{\sigma \in \mathfrak{S}_n} \mathcal{U}^*_{\sigma}
    \end{equation}
    is the symmetric projection onto the symmetric subspace of $(CS^*_{k,g})^{\tmax n}$. Clearly now for $\Phi= \one_{\CS}$ we have $(\gamma_n^{\Phi})^*(\one_{\CS}) = P_{\mathrm{Sym}_n(CS^*_{k,g})}(\one_{\CS}^{\otimes n}) = \one_{\CS}^{\otimes n} $ and the normalization of $N^{(n)}$ follows.
    
    On the other hand, by using the same expression for $(\gamma^\Phi_n)^*$, we see that 
    \begin{align}
        M_{a|x}(\varrho) &= \langle N^{(n)}(\varrho), (\gamma_u^\Phi)^*(m_{a}^{(x)}) \rangle =\langle N^{(n)}(\varrho), P_{\mathrm{Sym}_n(CS^*_{k,g})}(m_{a}^{(x)} \otimes \one_{\CS}^{\otimes (n-1)}) \rangle
    \end{align}
    for all $a \in [k]$, $x \in [g]$ and $\varrho \in K$.
    For any $\{x_1, \ldots, x_n\} \subset [g]$, we can define the joint measurement $G^{(x_1, \ldots x_n)}$ for $\{M_{\cdot|x_1}, \ldots, M_{\cdot|x_n}\}$ as
    \begin{equation}
        G_{a_1,\ldots,a_n}^{(x_1,\ldots,x_n)}(\varrho) = \langle N^{(n)}(\varrho), P_{\mathrm{Sym}_n(CS^*_{k,g})}(m^{(x_1)}_{a_1} \otimes \cdots \otimes m^{(x_n)}_{a_n})\rangle 
    \end{equation}
    for all $a_1, \ldots, a_n \in [k]$ and $\varrho \in K$. Clearly $G^{(x_1, \ldots x_n)}$ is a positive map and it correctly gives the measurements $\{M_{\cdot|x_1}, \ldots, M_{\cdot|x_n}\}$ as marginals:
    \begin{align}
        \sum_{\substack{j=1 \\ j\neq i}}^n \sum_{a_j=1}^k G_{a_1,\ldots,a_n}^{(x_1,\ldots,x_n)}(\varrho) &= \sum_{\substack{j=1 \\ j\neq i}}^n \sum_{a_j=1}^k  \langle N^{(n)}(\varrho), P_{\mathrm{Sym}_n(CS^*_{k,g})}(m^{(x_1)}_{a_1} \otimes \cdots \otimes m^{(x_n)}_{a_n})\rangle \\
        &= \langle N^{(n)}(\varrho), P_{\mathrm{Sym}_n(CS^*_{k,g})}(m_{a_i}^{(x_i)} \otimes \one_{\CS}^{\otimes (n-1)}) \rangle \\
        &= M_{a_i|x_i}(\varrho)
    \end{align}
    for all $a_i \in [k]$, $x_i \in [g]$, $i \in [n]$ and $\varrho \in K$. Thus, $M$ is $n$-wise compatible. Finally, by \cref{prop:extended-NS-tensor} for all $\varrho \in K$ we have that $G_{a_1,\ldots,a_n}^{(x_1,\ldots,x_n)}(\varrho)$ is of the form
   \begin{align}
       G_{a_1,\ldots,a_n}^{(x_1,\ldots,x_n)}(\varrho) = \frac{1}{n!} \sum_{\sigma \in \mathfrak{S}_n} p^{N^{(n)}(\varrho)}_{a_{\sigma(1)}, \ldots, a_{\sigma(n)}|x_{\sigma(1)}, \ldots, x_{\sigma(n)}},
   \end{align} 
    where $ p^{N^{(n)}(\varrho)}$ is the no-signaling probability distribution related to $N^{(n)}(\varrho) \in (CS^1_{k,g})^{\otimes n}$.
    Alternatively, one sees that $G_{a_1,\ldots,a_n}^{(x_1,\ldots,x_n)}(\varrho) = p^{\tilde{N}^{(n)}(\varrho)}_{a_{1}, \ldots, a_{n}|x_{1}, \ldots, x_{n}}$, where $\tilde{N}^{(n)} = P_{\mathrm{Sym}_n(CS_{k,g})} \circ N^{(n)}$. One then directly sees that the joint measurements inherit the no-signaling constraints from the extension so that \cref{eq:joint-meas-ns} is satisfied. This also means that the map $\tilde{N}^{(n)} = P_{\mathrm{Sym}_n(CS_{k,g})} \circ N^{(n)}$ is also a channel.

    On the other hand, if $M$ is $n$-wise compatible and we are given the joint measurements which satisfy the no-signaling conditions in \cref{eq:joint-meas-ns}, we can explicitly construct the $n$-extension $\xi_{N^{(n)}}$ of $\xi_M$. We will do this by using \cref{prop:extended-NS-tensor} by defining the no-signaling probability distributions $p^{N^{(n)}(\varrho)}$ for the elements $N^{(n)}(\varrho) \in (CS^{1}_{k,g})^{\tmax n}$ of the channel $N: K \to (CS^{1}_{k,g})^{\tmax n}$ corresponding to the tensor $\xi_{N^{(n)}} \in A(K)^+ \tmax (CS^{+}_{k,g})^{\tmax n}$.

    Let us start with the joint measurements $G^{(x_1,\ldots,x_n)}$ for all measurement choices $\{x_1, \ldots, x_n\}\subset [g]$. Note that for any $m<n$ we can also construct a joint measurement of the measurements $\{M_{\cdot|x_1}, \ldots, M_{\cdot|x_m}\}$ from the joint measurement $G^{(x_1,\ldots,x_n)}$ of any set of $n$ measurements $\{M_{\cdot|x_1}, \ldots, M_{\cdot|x_n}\}$ containing $\{M_{\cdot|x_1}, \ldots, M_{\cdot|x_m}\}$ as
    \begin{equation}
        \tilde{G}^{(x_1, \ldots, x_m, x_{m+1}, \ldots, x_{n})}_{a_1, \ldots, a_m} := \sum_{a_{m+1}=1}^k \cdots \sum_{a_{n}=1}^k G^{(x_1, \ldots,x_m,x_{m+1}, \ldots, x_n)}_{a_1, \ldots, a_m,a_{m+1}, \ldots, a_n}
    \end{equation}
    for all $a_1, \ldots, a_m \in [k]$. Furthermore, since the joint measurements satisfy the no-signaling conditions in \cref{eq:joint-meas-ns}, we see that if we choose some different measurement settings $\{\tilde{x}_{m+1}, \ldots, \tilde{x}_n\}$ for the remaining $n-m$ measurements and use a joint measurement $G^{(x_1, \ldots, x_m, \tilde{x}_{m+1}, \ldots, \tilde{x}_n)}$ for the settings $\{x_1, \ldots, x_m, \tilde{x}_{m+1}, \ldots, \tilde{x}_n\}$ to define the joint measurement of the measurements $\{M_{\cdot|x_1}, \ldots, M_{\cdot|x_m}\}$, we have that
        \begin{align}
        \tilde{G}^{(x_1, \ldots, x_m, \tilde{x}_{m+1}, \ldots, \tilde{x}_n)}_{a_1, \ldots, a_m} &:= \sum_{a_{m+1}=1}^k \cdots \sum_{a_{n}=1}^k G^{(x_1, \ldots, x_m, \tilde{x}_{m+1}, \ldots, \tilde{x}_n)}_{a_1, \ldots, a_m,a_{m+1}, \ldots, a_n} \\
        &=  \sum_{a_{m+1}=1}^k \cdots \sum_{a_{n}=1}^k G^{(x_1, \ldots,x_m,x_{m+1}, \ldots, x_n)}_{a_1, \ldots, a_m,a_{m+1}, \ldots, a_n} \\
        &= \tilde{G}^{(x_1, \ldots, x_m, x_{m+1}, \ldots, x_{n})}_{a_1, \ldots, a_m}
    \end{align}
    for all $a_1, \ldots, a_m \in [k]$. Thus, the two joint measurements for the measurements $\{M_{\cdot|x_1}, \ldots, M_{\cdot|x_m}\}$ are actually the same so that it only depends on the measurement settings $\{x_1, \ldots, x_m\}$.
    
    In the joint measurements the measurement choices are always different from each other but to construct the $n$-extensions we need to define it also in the cases when some of the measurement choices are identical. We will do this by giving the equal measurements always the same outcome. Next we will do this more formally.
    
    Let us take $x_i \in [g]$ for all $i \in [n]$. Let $m$ be the number of different measurement settings in $\{x_1, \ldots, x_n \}$ and let $y_{1}, \ldots, y_{m}$ be those different measurement settings. Let $l_j$ be the number of measurement settings in $\{x_1, \ldots, x_n \}$ which are equal to $y_j$. Let us partition $[n]$ into intervals of length $l_i$ such that $[n] = \{1, \ldots, l_1, l_1 +1, \ldots, l_1+l_2, l_1+ l_2+1, \ldots, \sum_{j=1}^m l_j\}$. Let us define $i_r =  \sum_{j=1}^r l_j \in [n]$ for all $r \in [m]$ so that now we have $[n] = \{1, \ldots, i_1, i_1 +1, \ldots, i_2, \ldots, i_m\}$, where $i_m=n$. Let us define a permutation $\sigma \in \mathfrak{S}_n$ as the permutation which orders the measurement settings such that 
    \begin{align}
        &x_{\sigma^{-1}(1)} = \ldots = x_{\sigma^{-1}(i_1)}=y_1, \\
        &x_{\sigma^{-1}(i_1+1)} = \ldots = x_{\sigma^{-1}(i_2)}=y_2, \\
        &\vdots\\
        &x_{\sigma^{-1}(i_{m-1}+1)} = \ldots = x_{\sigma^{-1}(i_m)} =y_m, 
    \end{align}
    and $\sigma^{-1}(i) < \sigma^{-1}(i')$ for all $i<i'$, $i,i' \in \{i_{r-1}+1, \ldots,i_{r}\}$ for all $r \in [n]$, where we set $i_0 = 0$.
    Now we can define the conditional probability distribution $(p^{N^{(n)}(\varrho)}_{\cdot,\ldots, \cdot|x_1, \ldots, x_n})$ as
        \begin{align}\label{eq:no-signaling-distr-joint-meas}
        p^{N^{(n)}(\varrho)}_{a_1,\ldots, a_n|x_1, \ldots, x_n} = \sum_{b_1=1}^k \cdots \sum_{b_m=1}^k \prod_{z=1}^m \delta_{a_{\sigma^{-1}(i_{z-1}+1)},\ldots,a_{\sigma^{-1}(i_z)},b_{z}}G^{(y_1, \ldots, y_m)}_{b_1, \ldots, b_m}(\varrho)
    \end{align}
    for all $a_1, \ldots, a_n \in [k]$ and $\varrho \in K$, where we are using the Kronecker delta as
    \begin{equation}
        \delta_{a_1, \ldots,a_u} = \begin{cases}
            1, \quad a_1 = \ldots = a_u, \\
            0, \quad \text{else,}
        \end{cases} \quad \text{for some } u \in \mathds{N} \, .
    \end{equation}
    Clearly $\sum_{i=1}^n \sum_{a_i=1}^k p^{N^{(n)}(\varrho)}_{a_1,\ldots, a_n|x_1, \ldots, x_n} =1$ so that it is a valid probability distribution for all $x_1, \ldots,x_n \in [g]$ and $\varrho \in K$.
    As an example, to see how the construction looks like explicitly, say we want to define $p^{N^{(3)}(\varrho)}_{a_1,a_2,a_3|x_1,x_1,x_3}$ where the first and second measurement choices are equal. Then, we have $y_1=x_1$ and $y_2=x_3$ and we pick the joint measurement between the first and the third party:
    \begin{equation}
        p^{N^{(3)}(\varrho)}_{a_1, a_2,a_3|x_1,x_1,x_3} = \sum_{b_1=1}^k \sum_{b_2=1}^k \delta_{a_1, a_2,b_1} \delta_{a_3,b_2} G_{b_1,b_2}^{(y_1,y_2)}(\varrho) = \delta_{a_1,a_2} G^{(x_1, x_3)}_{a_1,a_3}(\varrho)
    \end{equation}
    for all $a_1, a_3 \in [k]$ and $\varrho \in K$. Thus, we have just set the measurement outcomes the same for the equal measurements.
    
    We now observe that if we fix $i \in [n]$ so that $x_i = y_{j_i}$ for some $j_i \in [m]$ then we have that
    \begin{align}\label{eq:prob_distr_marginal}
        \sum_{\substack{t=1 \\ t\neq i}}^n \sum_{a_t=1}^k p^{N^{(n)}(\varrho)}_{a_1,\ldots, a_n|x_1, \ldots, x_n} = \sum_{\substack{s=1 \\ s \neq j_i}}^m \sum_{b_s=1}^k G^{(y_1, \ldots, y_{j_i}, \ldots, y_n)}_{b_1, \ldots, b_{j_i-1},a_i,b_{j_i+1}, \ldots, b_m }(\varrho) = M_{a_i|y_{j_i}}(\varrho) = M_{a_i|x_i}(\varrho)
    \end{align}
    for all $a_i \in [k]$ and $\varrho \in K$.
    
    According to \cref{prop:extended-NS-tensor}, in order to define the channel $N^{(n)}: K \to (CS^1_{k,g})^{\tmax n}$ via the conditional probabilities $p^{N^{(n)}(\varrho)} = (p^{N^{(n)}(\varrho)}_{\cdot,\ldots, \cdot|x_1, \ldots, x_n})_{x_1, \ldots, x_n \in [g]}$ for all $\varrho \in K$ we still need to check that the no-signaling constraints are satisfied. Thus, let $r \in [n]$. Now $x_r = y_{j_r}$ for some $j_r \in [m]$ and thus $r \in \sigma^{-1}(\{i_{j_{r-1}}+1, \ldots, i_{j_r}\})$. Then 
    \begin{align}
        \sum_{a_r=1}^k &p^{N^{(n)}(\varrho)}_{a_1,\ldots, a_n|x_1, \ldots,x_r, \ldots, x_n} = \sum_{a_r=1}^k \sum_{b_1=1}^k \cdots \sum_{b_m=1}^k \prod_{z=1}^m \delta_{a_{\sigma^{-1}(i_{z-1}+1)},\ldots,a_{\sigma^{-1}(i_z)},b_{z}}G^{(y_1, \ldots, y_m)}_{b_1, \ldots, b_m}(\varrho)\\
        &=  \sum_{b_1=1}^k \cdots \sum_{b_m=1}^k \left( \sum_{a_r=1}^k \delta_{a_{\sigma^{-1}(i_{j_r-1}+1)},\ldots,a_{\sigma^{-1}(i_{j_r})},b_{j_r}} \right) \prod_{\substack{z=1 \\ z \neq j_r}}^m \delta_{a_{\sigma^{-1}(i_{z-1}+1)},\ldots,a_{\sigma^{-1}(i_z)},b_{z}}G^{(y_1, \ldots, y_m)}_{b_1, \ldots, b_m}(\varrho)
    \end{align}
    On the other hand, let us now change the measurement setting $x_r$ to $\tilde{x}_r$ for some $\tilde{x}_r \in [g]$, $x_r \neq \tilde{x}_r$. Thus, we have the total measurement settings $\{x_1, \ldots, x_{r-1}, \tilde{x}_r, x_{r+1}, \ldots, x_n\}$. Let first $\tilde{x}_r \in \{y_1, \ldots, y_m\}$ so that there exists $\tilde{j}_r \in [m]$ such that $\tilde{x}_r = y_{\tilde{j}_r}$. Just as before, there is a permutation $\tilde{\sigma}$ which groups the same measurement settings together. The positions where the measurement settings change are now denoted by $\tilde{i}_1, \ldots,\tilde{i}_m$. Thus, now $r \in \tilde{\sigma}^{-1}(\{\tilde{i}_{\tilde{j}_{r-1}}+1, \ldots, \tilde{i}_{\tilde{j}_r}\})$. Let now $\pi \in \mathfrak{S}_n$ be a permutation such that $\pi \sigma = \tilde{\sigma}$. We note that now 
    \begin{align}
        \pi(\{i_{z-1}+1, \ldots,i_{z}\}) = \{\tilde{i}_{z-1}+1, \ldots, \tilde{i}_{z} \} 
    \end{align}
     for all $z \in [m]$. Then we have that
        \begin{align}
        \sum_{a_r=1}^k &p^{N^{(n)}(\varrho)}_{a_1,\ldots, a_n|x_1, \ldots,\tilde{x}_r, \ldots, x_n} = \sum_{a_r=1}^k \sum_{b_1=1}^k \cdots \sum_{b_m=1}^k \prod_{z=1}^m \delta_{a_{\tilde{\sigma}^{-1}(\tilde{i}_{z-1}+1)},\ldots,a_{\tilde{\sigma}^{-1}(\tilde{i}_z)},b_{z}}G^{(y_1, \ldots, y_m)}_{b_1, \ldots, b_m}(\varrho) \label{eq:thm-7.8-proof-1}\\
        &=  \sum_{b_1=1}^k \cdots \sum_{b_m=1}^k \left( \sum_{a_r=1}^k \delta_{a_{\tilde{\sigma}^{-1}(\tilde{i}_{\tilde{j}_r-1}+1)},\ldots,a_{\tilde{\sigma}^{-1}(\tilde{i}_{\tilde{j}_r})},b_{\tilde{j}_r}} \right) \prod_{\substack{z=1 \\ z \neq \tilde{j}_r}}^m \delta_{a_{\tilde{\sigma}^{-1}(\tilde{i}_{z-1}+1)},\ldots,a_{\tilde{\sigma}^{-1}(\tilde{i}_z)},b_{z}}G^{(y_1, \ldots, y_m)}_{b_1, \ldots, b_m}(\varrho) \\
        &=  \sum_{b_1=1}^k \cdots \sum_{b_m=1}^k \left( \sum_{a_r=1}^k \delta_{a_{\sigma^{-1}(\pi^{-1}(\tilde{i}_{\tilde{j}_r-1}+1))},\ldots,a_{\sigma^{-1}(\pi^{-1}(\tilde{i}_{\tilde{j}_r}))},b_{\tilde{j}_r}} \right) \nonumber  \\
        & \quad \times \prod_{\substack{z=1 \\ z \neq \tilde{j}_r}}^m \delta_{a_{\sigma^{-1}(\pi^{-1}(\tilde{i}_{z-1}+1))},\ldots,a_{\sigma^{-1}(\pi^{-1}(\tilde{i}_z))},b_{z}} G^{(y_1, \ldots, y_m)}_{b_1, \ldots, b_m}(\varrho) \\
        &=  \sum_{b_1=1}^k \cdots \sum_{b_m=1}^k \left( \sum_{a_r=1}^k \delta_{a_{\sigma^{-1}(i_{\tilde{j}_r-1}+1)},\ldots,a_{\sigma^{-1}(i_{\tilde{j}_r})},b_{\tilde{j}_r}} \right) \prod_{\substack{z=1 \\ z \neq \tilde{j}_r}}^m \delta_{a_{\sigma^{-1}(i_{z-1}+1)},\ldots,a_{\sigma^{-1}(i_z)},b_{z}}G^{(y_1, \ldots, y_m)}_{b_1, \ldots, b_m}(\varrho) \\
        &=\sum_{a_r=1}^k \sum_{b_1=1}^k \cdots \sum_{b_m=1}^k \prod_{z=1}^m \delta_{a_{\sigma^{-1}(i_{z-1}+1)},\ldots,a_{\sigma^{-1}(i_z)},b_{z}}G^{(y_1, \ldots, y_m)}_{b_1, \ldots, b_m}(\varrho) \\
        &=  \sum_{a_r=1}^k p^{N^{(n)}(\varrho)}_{a_1,\ldots, a_n|x_1, \ldots,x_r, \ldots, x_n} 
    \end{align} 
    so that the no-signaling constraints are satisfied.

    Above we considered the case when $\tilde{x}_r \in \{y_1, \ldots,y_m\}$. However, if $\tilde{x}_r \notin \{y_1, \ldots,y_m\}$ we can just repeat the whole procedure by defining $p^{N^{(n)}(\varrho)}_{a_1,\ldots, a_n|x_1, \ldots,x_r, \ldots, x_n}$ and $p^{N^{(n)}(\varrho)}_{a_1,\ldots, a_n|x_1, \ldots,\tilde{x}_r, \ldots, x_n}$ through the joint measurement of the measurement settings $\{y_1, \ldots, y_m, y_{m+1}\}$, where now $y_{m+1} = \tilde{x}_r$. Then in \cref{eq:no-signaling-distr-joint-meas} and \cref{eq:thm-7.8-proof-1} if some setting in $\{y_1, \ldots, y_m, y_{m+1}\}$ is not used in the settings $\{x_1, \ldots, x_r, \ldots, x_n\}$ or $\{x_1, \ldots, \tilde{x}_r, \ldots, x_n\}$, we just sum over its outcomes in the joint measurement $G^{(y_1, \ldots, y_{m+1})}$. By the same procedure we see that the no-signaling constraints hold also in this case.

    Thus, the conditional probability distributions $(p^{N^{(n)}(\varrho)}_{\cdot,\ldots,\cdot|x_1, \ldots, x_n})$ satisfy the no-signaling conditions for all $x_1, \ldots, x_n \in [g]$ and $\varrho \in K$. Then according to \cref{prop:extended-NS-tensor} we have that the tensors $N^{(n)}(\varrho) \in (CS^+_{k,g})^{\tmax n}$ defined by the no-signaling probability distributions $p^{N^{(n)}(\varrho)}$ are actually normalized, i.e, $N^{(n)}(\varrho) \in (CS^1_{k,g})^{\tmax n}$ for all $\varrho \in K$.  Thus, if we define a map $N^{(n)}: V(K)^+ \to (CS^+_{k,g})^{\tmax n}$ by $\varrho \mapsto N^{(n)}(\varrho)$ for all $\varrho \in K$, we see that it is a channel, i.e., $N^{(n)}: K \to (CS^1_{k,g})^{\tmax n}$.  As was observed before, if $N^{(n)}: K \to (CS^1_{k,g})^{\tmax n}$ is a channel, so is $\tilde{N}^{(n)} := P_{\mathrm{Sym}_n(CS_{k,g})} \circ N^{(n)}$. In fact, we have that
    \begin{align}
        p^{\tilde{N}^{(n)}(\varrho)}_{a_1, \ldots,a_n|x_1, \ldots,x_n} = \frac{1}{n!} \sum_{\sigma \in \mathfrak{S}_n} p^{N^{(n)}(\varrho)}_{a_{\sigma(1)}, \ldots, a_{\sigma(n)}|x_{\sigma(1)}, \ldots,x_{\sigma(n)}}
    \end{align}
    for all $a_1, \ldots, a_n \in [k]$ and $x_1, \ldots, x_n \in [g]$. Using \cref{eq:prob_distr_marginal} we see that 
    \begin{align}
        \sum_{\substack{t=1 \\ t\neq i}}^n \sum_{a_t=1}^k p^{\tilde{N}^{(n)}(\varrho)}_{a_1,\ldots, a_n|x_1, \ldots, x_n} = M_{a_i|x_i}(\varrho)
    \end{align}
    for all $a_i \in [k]$, $x_i \in [g]$, $i \in [n]$ and $\varrho \in K$.
    
   Finally we have that
    \begin{align}
        \langle \xi_M, \varrho \otimes m^{(x_i)}_{a_i} \rangle &= \langle M(\varrho), m^{(x_i)}_{a_i} \rangle = M_{a_i|x_i}(\varrho) = \sum_{\substack{t=1 \\ t\neq i}}^n \sum_{a_t=1}^k p^{\tilde{N}^{(n)}(\varrho)}_{a_1,\ldots, a_n|x_1, \ldots, x_n} \\
        &= \langle \tilde{N}^{(n)}(\varrho), \one_{CS^1_{k,g}}^{\otimes (i-1)} \otimes m^{(x_i)}_{a_i} \otimes \one_{CS^1_{k,g}}^{\otimes (n-i)} \rangle \\
        &= \langle P_{\mathrm{Sym}_n(CS_{k,g})} (N^{(n)}(\varrho)), \one_{CS^1_{k,g}}^{\otimes (i-1)} \otimes m^{(x_i)}_{a_i} \otimes \one_{CS^1_{k,g}}^{\otimes (n-i)} \rangle \\
        &=  \langle  N^{(n)}(\varrho), P_{\mathrm{Sym}_n(CS^*_{k,g})}(\one_{CS^1_{k,g}}^{\otimes (i-1)} \otimes m^{(x_i)}_{a_i} \otimes \one_{CS^1_{k,g}}^{\otimes (n-i)}) \rangle \\
        &=  \langle  N^{(n)}(\varrho), P_{\mathrm{Sym}_n(CS^*_{k,g})}(m^{(x_i)}_{a_i} \otimes \one_{CS^1_{k,g}}^{\otimes (n-1)}) \rangle \\
        &=\langle N^{(n)}(\varrho), (\gamma_n^\Phi)^*(m_{a_i}^{(x_i)}) \rangle \\
        &= \langle \xi_{N^{(n)}}, \varrho \otimes (\gamma_n^\Phi)^*(m_{a_i}^{(x_i)}) \rangle \\
        &= \langle (\id \otimes \gamma_n^{\Phi})(\xi_{N^{(n)}}), \varrho \otimes m_{a_i}^{(x_i)}\rangle  
    \end{align}
    for all $a_i \in [k]$, $x_i \in [g]$, $i \in [n]$ and $\varrho \in K$. Since the states in $K$ and the vectors $m_{a}^{(x)}$ span the corresponding vector spaces, we have that $\xi_M = (\id \otimes \gamma_n^{\Phi})(\xi_{N^{(n)}})$ which shows that $\xi_M$ is $n$-extendable.
\end{proof}

\subsection{Extendability of no-signaling behaviors}\label{sec:ext-LHV}

Going back to the setting of \cref{sec:bell}, we can also consider the concept of extendability of the general no-signaling statistics $P \in CS^1_{k,g} \tmax CS^1_{l,r}$, where separability corresponds to the corresponding no-signaling distribution having an LHV-model. As $n$-wise compatibility can be seen as a restricted version of compatibility, we can also consider a no-signaling behavior to have a restricted LHV-model.

\begin{defi}
    Let $g,r,k,l \in \nat$ and $m \in [g]$ and $n \in [r]$. A no-signaling probability distribution $p= (p_{\cdot,\cdot|x,y})_{x \in [g], y \in [r]}$ on $[k]\times[l]$ has an $(m,n)$-LHV model if for all $(x_1, \ldots,x_m) \subseteq [g]$ and $(y_1, \ldots, y_n) \subseteq [r]$ the restricted no-signaling probability distribution $p^{(x_1, \ldots, x_m,y_1, \ldots,y_n)} = (p_{\cdot,\cdot|x,y})_{x \in \{x_1, \ldots, x_m\}, y \in \{y_1, \ldots, y_n\}}$ on $[k] \times [l]$ has an LHV model.
\end{defi}

We note that $p= (p_{\cdot,\cdot|x,y})_{x \in [g], y \in [r]}$ has an LHV model if and only if it has an $(g,r)$-LHV model. 
Before connecting $n$-extendability to these restricted LHV models, let us first take a closer look on the set of LHV probability distributions. Let us denote by $\mathcal{P}^{NS}_{(k,g),(l,r)}$ the set of no-signaling probability distributions $p= (p_{\cdot,\cdot|x,y})_{x \in [g], y \in [r]}$ on $[k]\times[l]$; by \cref{prop:NS-tensor} they are exactly the conditional probability distributions related to tensors $P \in CS^1_{k,g} \tmax CS^1_{l,r}$. We denote by $\mathcal{P}^{LHV}_{(k,g),(l,r)}$ the subset of $\mathcal{P}^{NS}_{(k,g),(l,r)}$ that have an LHV-model. Thus, each $p \in \mathcal{P}^{LHV}_{(k,g),(l,r)}$ can be written as
\begin{align}
    p_{a,b|x,y} = \sum_{\lambda \in [\Lambda]} q_\lambda \, p^A_{a|\lambda,x} p^B_{b|\lambda,y}
\end{align}
for all $a \in [k]$, $b \in [l]$, $x \in [g]$ and $y \in [r]$ for some probability distributions $q$ on $[\Lambda]$, $p^A=(p^A_{\cdot|\lambda, x})_{\lambda \in [\Lambda], x \in [g]}$ on $[k]$ and $p^B=(p^B_{\cdot|\lambda, y})_{\lambda \in [\Lambda], y \in [r]}$ on $[l]$ for some $\Lambda \in \nat$. It is clear that $\mathcal{P}^{LHV}_{(k,g),(l,r)}$ is a polytope so that it can be expressed as the convex hull of its finite number of extreme points. The extreme points are then probability distributions $(p^A_{\cdot|x}p^B_{\cdot|y})_{x \in [g], y \in [r]}$ on $[k]\times [l]$ such that $p^A_{a|x}$ is non-zero only for one outcome $a_x \in [k]$ for all $x \in [g]$ and $p^B_{b|y}$ is non-zero only for one outcome $b_y \in [l]$ for all $y \in [r]$. We can thus label the extreme points $p^{\vec{a}, \vec{b}}$ by the outcome vectors $\vec{a}=(a_1, \ldots, a_g) \subset [k]^g$ and $\vec{b}=(b_1, \ldots, b_r) \subset [l]^r$ such that 
\begin{align}
    p^{\vec{a}, \vec{b}}_{a,b|x,y} = \delta_{a, a_x} \delta_{b, b_y}
\end{align}
for all $a \in [k]$, $b \in [l]$, $x \in [g]$ and $y \in [r]$. Now if we can express a probability distribution $p \in \mathcal{P}^{NS}_{(k,g),(l,r)}$ as a mixture of these extreme points, i.e., 
\begin{align}\label{eq:LHV-model-extreme}
    p_{a,b|x,y} = \sum_{\vec{a} \in [k]^g} \sum_{\vec{b} \in [l]^r} q_{\vec{a}, \vec{b}} \,p^{\vec{a}, \vec{b}}_{a,b|x,y}
\end{align}
for some probability distribution $q$ on $[k]^{g}\times [l]^r$, then we know that $p \in \mathcal{P}^{LHV}_{(k,g),(l,r)}$. Also, vice versa, if $p \in \mathcal{P}^{LHV}_{(k,g),(l,r)}$ then it has a convex decomposition as in \cref{eq:LHV-model-extreme}. We call a decomposition, or simply the probability distribution $q$ in \cref{eq:LHV-model-extreme} an \emph{LHV-decomposition} of $p$. We note that such a decomposition is not always unique. 

Similarly for the restricted LHV models, we see that $p= (p_{\cdot,\cdot|x,y})_{x \in [g], y \in [r]}$ has an $(n,m)$-LHV model if and only if an LHV-decomposition exists for all restricted distributions $p^{(x_1, \ldots, x_m,y_1, \ldots,y_n)} = (p_{\cdot,\cdot|x,y})_{x \in \{x_1, \ldots, x_m\}, y \in \{y_1, \ldots, y_n\}}$ on $[k] \times [l]$ for all $(x_1, \ldots,x_m) \subseteq [g]$ and $(y_1, \ldots, y_n) \subseteq [r]$.

In the next result we will take $l=k$ and $r=g$ since we will apply \cref{prop:extended-NS-tensor} where the cones match. As was noted in \cref{remark:NS-extension}, the same result could be obtained also if the cones are different, but we will not show it explicitly in this case. This restriction can also be avoided by simply embedding the smaller cone inside the larger one. In the proof we draw inspiration from \cite{Terhal2003symmetric}, where the authors showed that in a bipartite setting where the parties are both measuring $g$ measurements on a shared quantum state, then if the quantum state has a symmetric $g$-quasiextension (the extension need not be a quantum state but only in the maximum tensor product of $\Pos$ cones), then it can be only used to generate behaviors which have an LHV model.

\begin{thm}
    Let $P \in CS^1_{k,g} \tmax CS^1_{k,g}$ with the corresponding no-signaling probability distribution $p= (p_{\cdot,\cdot|x,y})_{x,y \in [g]}$ on $[k]\times[k]$ and let $n \in [g]$. Then $P$ is $n$-extendable if and only if $p$ has an $(g,n)$-LHV model such that for all settings
    $(y_1, \ldots,y_n)\subseteq [g]$ there exists some $(g,n)$-LHV decompositions $q^{(y_1, \ldots, y_n)}$ satisfying the following no-signaling constraints
    \begin{align}\label{eq:LHV-decomp-ns}
        \sum_{\substack{z=1 \\ z\neq i}}^g\sum_{a_z=1}^k \sum_{b_j=1}^k q^{(y_1, \ldots,y_j, \ldots , y_n)}_{a_1, \ldots,a_i, \ldots, a_g, b_1, \ldots,b_j, \ldots,b_n} = \sum_{\substack{z=1 \\ z\neq i}}^g\sum_{a_z=1}^k \sum_{b_j=1}^k q^{(y_1, \ldots,\tilde{y}_j, \ldots , y_n)}_{a_1, \ldots,a_i, \ldots, a_g, b_1, \ldots,b_j, \ldots,b_n}
    \end{align}
    for all $b_1, \ldots,b_n \in [k]$, $y_j,\tilde{y}_j \in [g]$ and $i \in [g]$, $j \in [n]$. 
    If $n = g$, then $P \in  CS^1_{k,g} \tmin CS^1_{k,g}$.
\end{thm}
\begin{proof}
    Let first $P$ be $n$-extendable. Thus, there exists $P^{(n)} \in CS^1_{k,g} \tmax (CS^1_{k,g})^{\tmax n}$ s.t. $(\id \otimes \gamma^\Phi_n)(P^{(n)}) = P$ for $\Phi = \one_{CS^1_{k,g}}$. By \cref{prop:extended-NS-tensor} we have that for all $a, b_i \in [k]$ and $x,y_i \in [g]$ 
    \begin{align}
        p_{a,b_i|x,y_i} &= \langle P, m^{(x)}_a \otimes m^{(y_i)}_{b_i} \rangle \\
        &= \langle (\id \otimes \gamma^\Phi_n)(P^{(n)}), m^{(x)}_a \otimes m^{(y_i)}_{b_i} \rangle \\
        &= \langle P^{(n)}, m^{(x)}_a \otimes P_{\mathrm{Sym}_n(CS^*_{k,g})}(m^{(y_i)}_{b_i} \otimes \one_{CS^1_{k,g}}^{\otimes n}) \rangle \\
        &= \sum_{\substack{j=1 \\ j \neq i}}^{n} \sum_{b_j=1}^k  \frac{1}{n!} \sum_{\sigma \in \mathfrak{S}_n} p^{(n)}_{a,\alpha_\sigma(b_1, \ldots, b_{n})|x,\alpha_\sigma(1, \ldots,1,y_i,1, \ldots, 1)} \\
        &=  \sum_{\substack{j=1 \\ j \neq i}}^{n} \sum_{b_j=1}^k  \tilde{p}^{(n)}_{a,b_1, \ldots, b_{n}|x,1, \ldots,1,y_i,1, \ldots, 1}\, ,
    \end{align}
    where $\tilde{p}^{(n)} = (\tilde{p}^{(n)}_{\cdot, \ldots, \cdot|z_1, \ldots,z_{n+1}})_{z_1, \ldots,z_{n+1} \in [g]}$ is the no-signaling probability distribution on $[k]^{n+1}$ related to the tensor $(\id \otimes \gamma^\Phi_n)(P^{(n)}) \in (CS^1_{k,g})^{\tmax (n+1)}$.

    Let us now fix $\{y_1, \ldots, y_n\} \subset [g]$. We set $r^{(y_1, \ldots,y_n)}_{b_1, \ldots, b_n} = \sum_{a=1}^k \tilde{p}^{(n)}_{a,b_1, \ldots, b_n|1,y_1, \ldots,y_n}$ for all $b_i \in [k]$ for all $i \in [n]$ and define a probability distribution $q^{(y_1, \ldots,y_n)}$ on $[k]^{gn}$ by setting
    \begin{align}
        q^{(y_1, \ldots,y_n)}_{a_1, \ldots, a_g, b_1, \ldots, b_n} = \frac{1}{\left(r^{(y_1, \ldots,y_n)}_{b_1, \ldots, b_n}\right)^{g-1}} \prod_{x=1}^g \tilde{p}^{(n)}_{a_x,b_1, \ldots, b_{n}|x,y_1, \ldots,y_n}
    \end{align}
    for all $a_1, \ldots,a_g,b_1, \ldots,b_n \in [k]$. Clearly, since $\tilde{p}^{(n)}$ is also a no-signaling distribution, and in  \cref{eq:LHV-decomp-ns} we are summing over all $a_i$'s except for one, we see that the no-signaling constraints of \cref{eq:LHV-decomp-ns} are satisfied for $q^{(y_1, \ldots,y_n)}$.

    Now for all $a, b \in [k]$, $x \in [g]$ and $i \in [n]$ we have that
    \begin{align}
        \sum_{z=1}^{g} & \sum_{a_z=1}^k  \sum_{j =1 }^n \sum_{b_j=1}^k q^{(y_1, \ldots,y_n)}_{a_1, \ldots, a_g, b_1, \ldots, b_n} p^{(a_1, \ldots,a_g), (b_1, \ldots,b_n)}_{a,b|x,y_i} \label{eq:g-n-LHV}\\
        &= \sum_{z=1}^{g}  \sum_{a_z=1}^k  \sum_{j =1 }^n \sum_{b_j=1}^k q^{(y_1, \ldots,y_n)}_{a_1, \ldots, a_g, b_1, \ldots, b_n} \delta_{a,a_x} \delta_{b,b_i} \\
        &= \sum_{\substack{z=1 \\ z\neq x}}^{g} \sum_{a_z=1}^k  \sum_{\substack{j =1 \\ j\neq i}}^n \sum_{b_j=1}^k  q^{(y_1, \ldots,y_n)}_{a_1, \ldots,a_{x-1},a,a_{x+1}, \ldots, a_g, b_1,\ldots, b_{i-1}, b, b_{i+1}, \ldots, b_n} \\
        &= \sum_{\substack{j =1 \\ j\neq i}}^n \sum_{b_j=1}^k \tilde{p}^{(n)}_{a,b_1,\ldots, b_{i-1}, b, b_{i+1}, \ldots, b_n|x,y_1, \ldots,y_n} \\
        &= p_{a,b|x,y_i} \, .
    \end{align}
    Hence, the restricted no-signaling distribution $p^{(y_1, \ldots,y_n)} = (p_{\cdot,\cdot|x,y})_{x \in [g], y \in \{y_1, \ldots, y_n\}}$ on $[k] \times [k]$ for all $\{y_1, \ldots, y_n\} \subset [g]$ has a LHV-model which means that $p$ has a $(g,n)$-LHV model.

    On the other hand, let now $p$ have a $(g,n)$-LHV model s.t. for all settings
    $(y_1, \ldots,y_n)\subseteq [g]$ there exists some $(g,n)$-LHV decompositions $q^{(y_1, \ldots, y_n)}$ satisfying the no-signaling constraints in \cref{eq:LHV-decomp-ns}. We can define a conditional probability distribution $p'^{(n)} = (p'^{(n)}_{\cdot, \ldots,\cdot|x,y_1, \dots,y_{n}})_{x,y_1, \ldots,y_{n} \in [g]}$ on $[k]^{n+1}$ by setting 
    \begin{equation}
        p'^{(n)}_{a,b_1, \ldots,b_n|x,y_1, \dots,y_{n}} = \sum_{\substack{i=1 \\ i \neq x}}^g \sum_{a_i=1}^k q^{(y_1, \ldots,y_{n})}_{a_1, \ldots,a_{x-1}, a, a_{x+1}, \ldots, a_g, b_1, \ldots, b_n} \, .
    \end{equation}
    Clearly 
    \begin{align}
        \sum_{a=1}^k p'^{(n)}_{a,b_1, \ldots,b_n|x,y_1, \dots,y_{n}} = \sum_{a=1}^k p'^{(n)}_{a,b_1, \ldots,b_n|\tilde{x},y_1, \dots,y_{n}}
    \end{align}
    for all $b_1, \ldots, b_n \in [k]$, $x,\tilde{x},y_1, \ldots,y_n \in [g]$. Also, since the no-signaling constraints in \cref{eq:LHV-decomp-ns} are satisfied, we also have that 
    \begin{align}\label{eq:ns-ext-eq0}
        \sum_{b_j=1}^k p'^{(n)}_{a,b_1, \ldots,b_n|x,y_1, \dots,y_j, \ldots, y_{n}} = \sum_{b_j=1}^k p'^{(n)}_{a,b_1, \ldots,b_n|x,y_1,\ldots,\tilde{y}_j, \ldots,y_{n}}
    \end{align}
    for all $a, b_1, \ldots,b_n \in [k]$, $x,y_1, \ldots,y_n,\tilde{y}_j \in [g]$ for all $j \in [n]$. Hence, $p'^{(n)}$ is a no-signaling probability distribution and thus by \cref{prop:extended-NS-tensor} it corresponds to a tensor $P'^{(n)} \in (CS^1_{k,g})^{\tmax(n+1)}$.
    We see that now for all $a,b \in [k]$ and $x,y_i \in [g]$:
    \begin{align}
        \langle P, m^{(x)}_a \otimes m^{(y_i)}_{b} \rangle &= p_{a,b|x,y_i} = \sum_{z=1}^{g} \sum_{a_z=1}^k  \sum_{j =1 }^n \sum_{b_j=1}^k q^{(y_1, \ldots,y_n)}_{a_1, \ldots, a_g, b_1, \ldots, b_n} p^{(a_1, \ldots,a_g), (b_1, \ldots,b_n)}_{a,b|x,y_i} \\
        &= \sum_{\substack{z=1 \\ z\neq x}}^{g} \sum_{a_z=1}^k  \sum_{\substack{j =1 \\ j\neq i}}^n \sum_{b_j=1}^k  q^{(y_1, \ldots,y_n)}_{a_1, \ldots,a_{x-1},a,a_{x+1}, \ldots, a_g, b_1,\ldots, b_{i-1}, b, b_{i+1}, \ldots, b_n} \\
        &= \sum_{\substack{z=1 \\ z\neq x}}^{g} \sum_{a_z=1}^k  \sum_{\substack{j =1 \\ j\neq i}}^n \sum_{b_j=1}^k  q^{(1, \ldots,1,y_i,1, \ldots,1)}_{a_1, \ldots,a_{x-1},a,a_{x+1}, \ldots, a_g, b_1,\ldots, b_{i-1}, b, b_{i+1}, \ldots, b_n} \\
        &= \sum_{\substack{j =1 \\ j\neq i}}^n \sum_{b_j=1}^k p'^{(n)}_{a,b_1, \ldots,b_n|x,1, \dots,1, y_{i}, 1, \ldots,1} \label{eq:ns-ext-eq1} \\
        &= \sum_{\substack{j=1 \\ j \neq i}}^{n} \sum_{b_j=1}^k  \frac{1}{n!} \sum_{\sigma \in \mathfrak{S}_n} p'^{(n)}_{a,\alpha_\sigma(b_1, \ldots, b_{n})|x,\alpha_\sigma(1, \ldots,1,y_i,1, \ldots, 1)} \label{eq:ns-ext-eq2} \\
        &= \langle P'^{(n)}, m^{(x)}_a \otimes P_{\mathrm{Sym}_n(CS^*_{k,g})}( m^{(y_i)}_b \otimes \one_{CS^1_{k,g}}^{\otimes (n-1)} \rangle \\
        &= \langle (\id \otimes \gamma^\Phi_n)(P'^{(n)}), m^{(x)}_a \otimes m^{(y_i)}_b \rangle,
    \end{align}
    where the equality between \cref{eq:ns-ext-eq1} and \cref{eq:ns-ext-eq2} follows from the no-signaling constraints in \cref{eq:ns-ext-eq0}. Since the vectors $m^{(x)}_a$ span $CS^*_{k,g}$ we have that $P = (\id \otimes \gamma^\Phi_n)(P'^{(n)})$ so that $P$ is $n$-extendable.
\end{proof}

\section{Robustness results} \label{sec:robustness}

In the previous sections, we have seen that whether a positive linear map between two cones factorizes through a third cone as a composition of positive maps depends heavily on the structure of the three cones and on the types of tensor product one uses in constructing the cones. In this section we would like to understand what type of limitations on factorization one can derive directly from the structure of the state space of column stochastic matrices. 

As a first case study, let us suppose that the identity channel $\mathrm{id}: CS^1_{l,r} \to  CS^1_{l,r}$ can be factorized through $CS^1_{k,g}$. That means that there are channels $\Psi: CS^1_{l,r} \to CS^1_{k,g}$ and $\Phi: CS^1_{k,g} \to CS^1_{l,r}$ such that $\mathrm{id}= \Phi \circ \Psi$. Operationally this means that the action of a particular simulation could always be reversed. Namely, if $M: {K} \to CS^1_{l,r}$ is any multimeter on a state space ${K}$, then from the factorization of $\mathrm{id}$ we have that $M = \Phi \circ \Psi \circ M$ so that the simulation of $M$ by the channel $\Psi$ can be reversed by applying the channel $\Phi$. Naturally, there are cases when such factorization cannot be feasible:  if $g=1$ so that $\Psi \circ M$ is a single measurement meaning that $M = \Phi \circ \Psi \circ M$ implies the compatibility of all of the measurements in $M$ even when $M$ would contain incompatible measurements. On the other hand, in some cases this kind of reversibility is natural, for example when $k\geq l$ and $g\geq r$ so that $CS^1_{l,r}$ can be trivially embedded in $CS^1_{k,g}$. We can use \Cref{prop:sim-irr-polysimplices} to show that actually this is the only case when the identity channel can be factorized like that.

\begin{prop}
Let $l,r \in \nat$. The identity channel $\mathrm{id}: CS^1_{l,r} \to CS^1_{l,r}$ can be factorized through $CS^1_{k,g}$ for some $k,g \in \mathds{N}$ if and only if $k \geq l$ and $g \geq r$.
\end{prop}
\begin{proof}
If $k \geq l$ and $g \geq r$, we can trivially consider a column stochastic matrix of size $l \times r$ as a column stochastic matrix of size $k \times g$ by adding zero rows and adding extra columns with element 1 on some row (this would correspond to adding trivial measurement $\one_K$ to your multimeter).

We first note that the identity channel $\mathrm{id}: CS^1_{l,r} \to CS^1_{l,r}$ can be seen as a multimeter on $CS^1_{l,r}$. Let us write this multimeter explicitly in the basis of $CS_{l,r}$. In fact, one can confirm that actually
\begin{equation}
    \mathrm{id}(X) = \one_{CS^1_{l,r}}(X) s_{l, \ldots, l} + \sum_{y =1}^r \sum_{b=1}^{l-1} m^{(y)}_b(X) e^{(y)}_b
\end{equation}
for all $X \in CS_{l,r}$. Thus, $\mathrm{id}$ corresponds to a multimeter consisting of the measurements $\{m^{(y)}\}_{y\in [r]}$ on $CS^1_{l,r}$.

If now the identity channel $\mathrm{id}: CS^1_{l,r} \to CS^1_{l,r}$ is factorizable through $CS^1_{k,g}$ then there exists channels $\Psi: CS^1_{l,r} \to CS^1_{k,g}$ and $\Phi: CS^1_{k,g} \to CS^1_{l,r}$ such that $\mathrm{id}= \Phi \circ \Psi$. We can now interpret $\Psi = \{\Psi_{\cdot|x}\}_{x \in [g]}$ as a multimeter on $CS^1_{l,r}$ which thus consists of $g$ measurements with $k$ outcomes. Thus, by  \Cref{cor:simulation-factorization} the factorization $\mathrm{id}= \Phi \circ \Psi$ can be interpreted as classically simulating the multimeter $\mathrm{id}$ by the multimeter $\Psi$. However, since by \Cref{prop:sim-irr-polysimplices} the multimeter $\mathrm{id}$ consists only of simulation irreducible measurements  $\{m^{(y)}\}_{y\in [r]}$, it means that for all $y \in [r]$, there exists $x_y \in [g]$ such that $\Psi_{\cdot | x_y}$ is post-processing equivalent with $m^{(y)}$. Also, because none of the measurements $m^{(1)}, \ldots, m^{(r)}$ are post-processing equivalent, we must have that $x_y \neq x_{y'}$ for all $y,y' \in [r]$ such that $y \neq y'$. This can only hold if $g \geq r$. Furthermore, since $m^{(y)}$ is the unique minimally sufficient representative of its equivalence class $[[m^{(y)}]]$, to which $\Psi_{\cdot | x_y}$ also belongs, we must have that $m^{(y)}$ can be obtained from $\Psi_{\cdot | x_y}$ by joining all of its pairwise linearly dependent effects as was explained in \Cref{prop:sim-irr-extr}. This means that for the outcomes we must have that $k \geq l$. 
\end{proof}

Thus, the only time an identity channel on a state space of column stochastic matrices can be factorized through a state space of column stochastic matrices of different size is exactly when column stochastic matrix is a submatrix of a larger column stochastic matrix . This means that no pathologies appear in the classical simulability of a multimeter $M$, since otherwise we could decrease the size of the column stochastic matrices via factorization of the identity channel.

\subsection{Compatibility}

Recall that a multimeter $M_{a|x}$ acting on a GPT state space $K$ can be viewed as a channel $M: K \to CS^1_{k,g}$, where $k$ denotes the number of outcomes for each of the $g$ measurement settings. The state space $CS^1_{k,g}$ is a polysimplex.

Let us prove the following well-known result using the factorization characterization of compatibility. 

\begin{prop}
    Let $M_{a|x}$ be an arbitrary multimeter, $P_{a|x} = p_{a|x} \one_K$ a trivial multimeter, and $q_x$ a probability measure on $g$. Then, the noisy multimeter
    \begin{equation}
        N_{a|x} = q_x M_{a|x} + (1-q_x)p_{a|x} \one_K
    \end{equation}
    is compatible. In particular,
    \begin{equation}
        \tilde M_{a|x} = \frac 1 g M_{a|x} + (1-\frac1 g) \frac{\one_K}{k}
    \end{equation}
    is compatible.
\end{prop}
\begin{proof}
    For any $z \in [g]$, consider the following multimeter:
    \begin{equation}
        M^{(z)}_{a|x} =
        \begin{cases}
        M_{a|x} &\qquad \text{ if }x=z\\
        p_{a|x}\one_K&\qquad \text{ if }x \neq z.
        \end{cases}
    \end{equation}
    The intuition from quantum mechanics is that this multimeter should be compatible, since the effects of different measurements commute. Indeed, we have 
    \begin{equation}
    	\left[K \xlongrightarrow{M^{(z)}} CS^1_{k,g}\right] = \left[ K \xlongrightarrow{M} CS^1_{k,g}  \xlongrightarrow{\operatorname{marg}_z} S_k \xlongrightarrow{\operatorname{triv}_z} CS^1_{k,g}\right].
    \end{equation}
    where $\operatorname{marg}_z$ is the channel that discards all the indices but $z$: 
    \begin{equation}
	   \operatorname{marg}_z[(q_{a|x})_{a \in [k], x \in [g]}] = q_{\cdot|z}
    \end{equation}
    and $\operatorname{triv}_z$ is the channel that embeds the $z$-th copy of the simplex $S_k$ into the polysimplex $CS^1_{k,g}$ in the following way:
    \begin{equation}
        \operatorname{triv}_z[(q_a)] = (r_{a|x}) \quad \text{ with } r_{a|x} = \begin{cases}
        q_{a} &\qquad \text{ if }x=z\\
        p_{a|x}\one_K &\qquad \text{ if }x \neq z.
        \end{cases}
    \end{equation}
    Note that the multimeter $M^{(z)}$ is compatible, since it factors through the $z$-th copy of $S_k$. Finally, we have
    \begin{equation}
        N = \sum_{z \in [g]} q_z M^{(z)}
    \end{equation}
    proving the claim via convexity.
\end{proof}

Let us consider now noisy multimeters as GPT channels between some GPT with state space $K$ and some noisy version of the polysimplex, seen as a GPT, that depends on the amount and the type of noise considered. 

We have 
\begin{equation}
	K \xlongrightarrow{M} CS^1_{k,g; \, \text{noisy}}.
\end{equation}
Let us consider two examples. Firstly, let us consider adding white noise to $k$-outcome measurements: The noisy effects are
\begin{equation}
	\tilde A_i := t A_i + (1-t)\frac{\one_K}{k},
\end{equation}
where $t \in [0,1]$ is the \emph{noise parameter}. The value $t=1$ corresponds to having no noise (the original POVM), while the value $t=0$ corresponds to the uniform trivial measurement having effects $\one_K / k$. Note that the range of the map 
\begin{equation}
	K \xlongrightarrow{A} S_k,
\end{equation}
when restricted to states in $K$, is necessarily \emph{smaller} than the full probability simplex on $k$ vertices. Indeed, if the outcome probabilities for the original measurement are $p$, then the probabilities corresponding to the noisy measurement are 
\begin{equation}
	\tilde p_i = t p_i + \frac{1-t}{k}.
\end{equation}
Hence, the noisy measurement has range 
\begin{equation}
	K \xlongrightarrow{\tilde A} S_{k; t},
\end{equation}
where $S_{k; t}$ is the state space of the GPT having $\Rnum^k$ as a vector space, an order unit identical to the one of $S_k$, i.e., for $v \in \mathds R^k$
\begin{equation}
	\one_{S_{k; t}}(v) = \sum_{i=1}^k v_i,
\end{equation}
but a \emph{thinner} cone given by
\begin{equation}
	(\Rnum_+^k)_t := \{ v\in \Rnum^k \, : \, \forall i \in [k], \, v_i \geq (1-t) \bar v, \text{ where } \bar v:=k^{-1}\sum_{j=1}^k v_j\}.
\end{equation}
Geometrically, the set $S_{k; t}$ is obtained by scaling the usual simplex $S_k$ by a factor of $t$, around its ``central'' point $(1/k, \ldots, 1/k)$. 

As a second example, let us consider a multimeter $M$ with two $2$-outcome measurements, to which we apply uniform noise, with parameters $t_1$ and $t_2$ respectively: 
\begin{equation}
	\tilde M_{a|x} = t_x M_{a|x} + (1-t_x) \frac{\one_K}{2}.
\end{equation}
The situation is similar to the one before: the noisy multimeter has range 
\begin{equation}
	K \xlongrightarrow{\tilde M} CS^1_{2,2; (t_1,t_2)},
\end{equation}
where the noisy polysimplex $CS^1_{2,2; (t_1,t_2)}$ has vector space 
\begin{equation}
	\Rnum^{2+2} \cap E,
\end{equation}
where 
\begin{equation}
	E = \{ (v_{a|x}) \in \Rnum^4 \, : \, \sum_a v_{a|1} = \sum_{a'} v_{a'|2} \},
\end{equation}
order unit
\begin{equation}
	\one_{CS^1_{2,2; (t_1,t_2)}}(v) = \sum_a v_{a|1} = \sum_{a'} v_{a'|2},
\end{equation}
and cone
\begin{equation}
	\left[(\Rnum_+^2)_{t_1} \times \Rnum_+^2)_{t_2} \right] \cap E,
\end{equation}
where the thin cones $(\Rnum_+^2)_{t_1}$ and $(\Rnum_+^2)_{t_2}$ have been defined above. We display the noisy (poly)simplices in \Cref{fig:noisy-polysimplices}.

\begin{figure}[htb]
    \centering
    \includegraphics[width=0.3\linewidth]{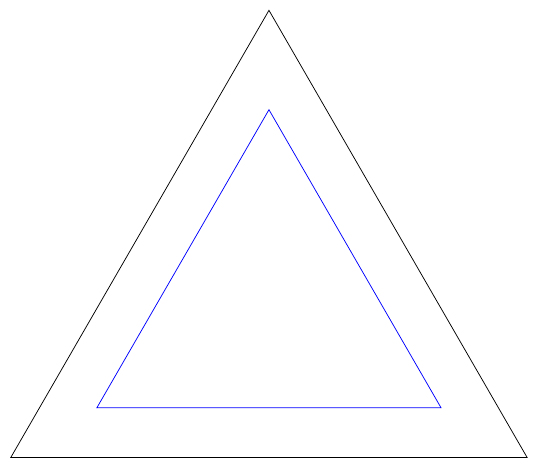}\qquad\qquad\qquad
    \includegraphics[width=0.3\linewidth]{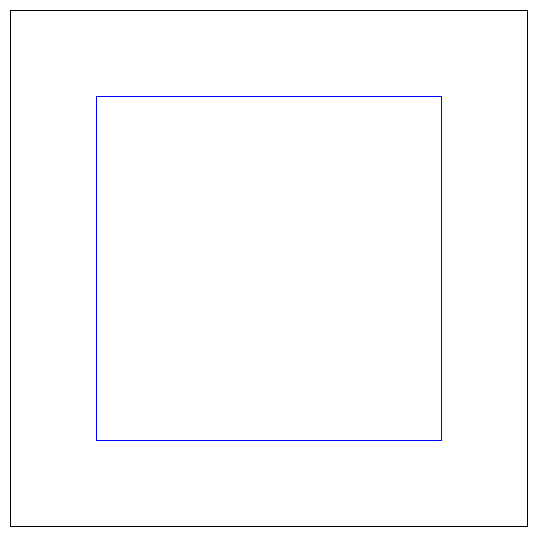}
    
    \caption{State spaces of the probability simplex GPT $S_3$ (left) and the polysimplex $CS^1_{2,2}$, along with their noisy versions in blue, for noise parameters $t=2/3$, respectively $(t_1,t_2) = (2/3,2/3)$.}
    \label{fig:noisy-polysimplices}
\end{figure}

We shall be interested in what follows in the following question: is it possible to factorize the identity map between a noisy version of some (poly)simplex and its noiseless version through a different (poly)simplex: 
\begin{equation}\label{eq:factorization-general}
    CS^1_{k,g; \, \text{noisy}} \xlongrightarrow{\operatorname{id}} CS^1_{k,g} = CS^1_{k,g; \, \text{noisy}} \xlongrightarrow{\varphi} CS^1_{l,g} \xlongrightarrow{\psi} CS^1_{k,g}?
\end{equation}
Above, the maps $\varphi, \psi$ are GPT channels, i.e.~positive and unit-preserving linear maps.

This type of factorization implies that a noisy multimeter with outcome counts $\underline k$ can be simulated by a multimeter with outcome counts $\underline l$, and that \emph{independently of the GPT} $G$ that the multimeter acts on. Importantly, the factorization above depends only on the respective polysimplices, no other physical GPTs being involved. 

In practice, we shall tackle the existence of such factorizations from a geometric perspective, as follows: First, choose coordinates for the vector space space of the GPT $CS^1_{k,g}$ such that
\begin{equation}
	V(CS^1_{k,g}) \ni v = (s, v_1, \ldots, v_n) \quad \text{ and } \quad \one_{CS^1_{k,g}}(v) = s,
\end{equation}
where $n:=g(k-1)$. 
Proceed similarly for the GPT $CS^1_{l,g}$ and its associated vector space $\Rnum^m$, where $m=g(l-1)$. Note that $m \geq n$ is a necessary condition for the existence of such a factorization, for obvious rank reasons. We have the following result. 

\begin{prop}\label{prop:robustness-inclusion}
     A factorization as in \cref{eq:factorization-general} exists if and only if there exists $T$, an affine image of $CS^1_{l,g}$ in $\Rnum^m$, such that:
\begin{itemize}
    \item $CS^1_{k,g; \, \text{noisy}} \oplus 0_{m-n} \subseteq T$;
    \item there exists a projection $\Pi : \Rnum^m \to \Rnum^n$, not necessarily orthogonal, such that $\Pi(T) \subseteq CS^1_{k,g}$.
\end{itemize}
\end{prop}
\begin{proof}
    Consider the linear extensions of the affine maps $\varphi$ and $\psi$ from \cref{eq:factorization-general}. Let $\tilde \varphi: \Rnum^n \to \Rnum^m$ and $\tilde \psi: \Rnum^m \to \Rnum^n$ denote their unique linear extensions to the corresponding vector spaces. Let $\iota: \Rnum^n \to \Rnum^m$ be the standard inclusion $x \mapsto (x, 0)$ and $\pi: \Rnum^m \to \Rnum^n$ be the standard projection $(x, y) \mapsto x$. Note $\pi \circ \iota = \operatorname{id}_{\Rnum^n}$. 

    Define the affine map $A: \Rnum^m \to \Rnum^m$ by $A = \iota \circ \tilde \psi $.
    Let $T := A(CS^1_{l, g}) = \iota (\tilde\psi(CS^1_{l, g}))$. 
    Since $\iota$ is injective and $\tilde \psi$ maps the $m$-dimensional affine space $\operatorname{aff}(CS^1_{l, g})$ to the $n$-dimensional affine space $\operatorname{aff}(CS^1_{k, g})$, $T$ is an affine image of $CS^1_{l, g}$. It is obvious that $\Pi:=\pi$ verifies the second point in the statement.
\end{proof}

\subsection{Two noisy binary measurements simulated by a three-outcome measurement}

It is a standard result that any pair of dichotomic measurements are compatible, provided that they have noise parameters $t_1=t_2=t \leq 1/2$. Moreover, the two measurements can be post-processed from a joint $4$-outcome measurement. We shall now investigate the same question, with the difference that we ask about pairs of dichotomic measurements that can be post-processed from a $3$-outcome measurement. 

\begin{prop}
    For all $t \leq 1/3$, the identity map 
    \begin{equation}
        \operatorname{id} : CS^1_{2,2; (t,t)} \to CS^1_{2,2}
    \end{equation}
    factorizes through $S_3$. In particular, any pair of noisy binary measurements with noise parameter $t \leq 1/3$ admit a joint measurement with $3$ outcomes.
\end{prop}
\begin{proof}
The construction is depicted in \Cref{fig:two-dichotomic-measurements}, left panel. More precisely, noisy measurements 
\begin{align}
    \tilde X_i &= \frac 1 3 X_i + \frac 1 3 \one_K \qquad i=1,2\\
    \tilde Y_j &= \frac 1 3 Y_j + \frac 1 3 \one_K \qquad j=1,2
\end{align}
can be post-processed from the $3$-outcome measurement $(A,B,C)$ with 
\begin{align}
    A &= \tilde Y_1 = \frac 1 3 Y_1 + \frac 1 3 \one_K\\
    B &= \tilde X_1 - \frac 1 2 \tilde Y_1 = \frac 1 3 X_1 + \frac 1 6 Y_2\\
    C &= \one_K - \tilde X_1 - \frac 1 2 \tilde Y_1 = \frac 1 3 X_2 + \frac 1 6 Y_2.  
\end{align}

\end{proof}

\begin{conjecture}
The value of $t$ in the result above is optimal. 
\end{conjecture}

Note that if we allow measurements with four outcomes as the middle polysimplex, then one can achieve $t=1/2$; this is standard result about compatibility of two noisy dichotomic measurements, see \Cref{fig:two-dichotomic-measurements}, right panel.

\begin{figure}
    \centering
    \includegraphics[width=0.3\linewidth]{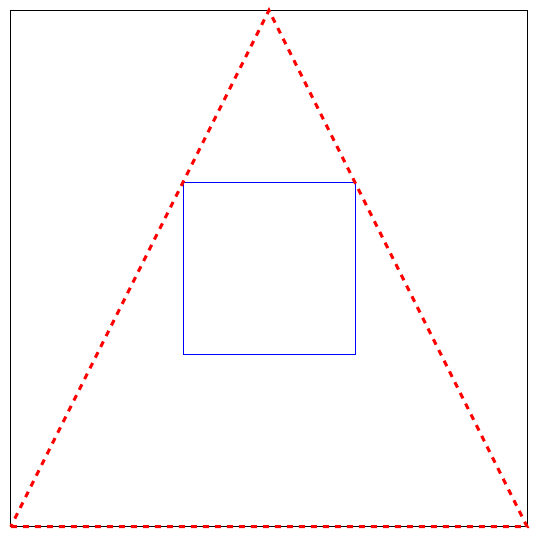}\qquad\qquad
    \includegraphics[width=0.3\linewidth]{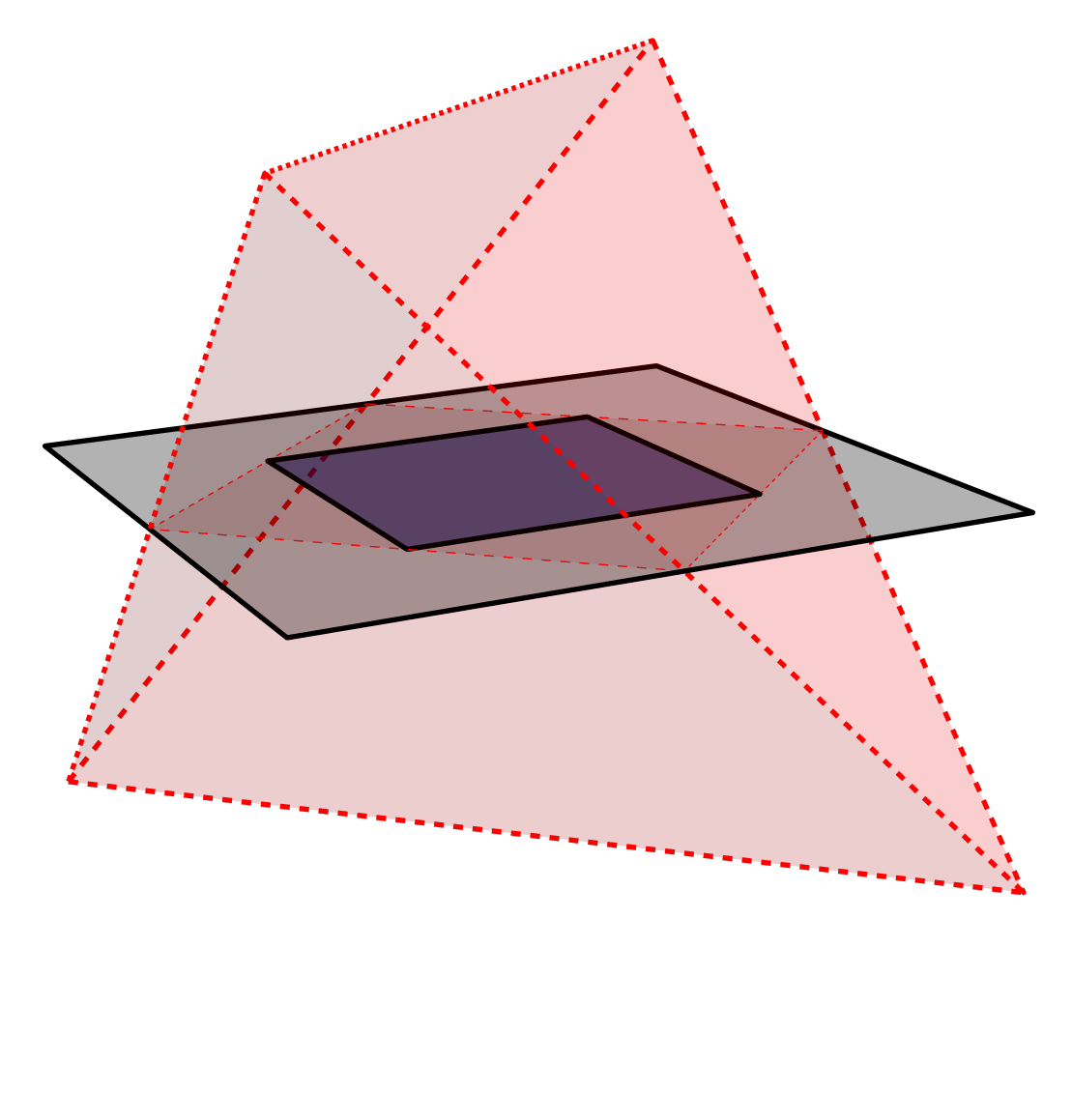}
    \caption{Any pair of two noisy dichotomic measurements with noise parameter $t=1/3$ can be post-processed from a $3$-outcome measurement (left). Any pair of two noisy dichotomic measurements with noise parameter $t=1/2$ can be post-processed from a $4$-outcome measurement (right). The state space of $CS^1_{2,2}$ is the outer square (black), its noisy version is the inner square (blue), and the affine image of the probability simplex $S_3$ (resp.~$S_4$) is the triangle (resp.~tetrahedron), in red. The inclusions of the state spaces correspond to those from \cref{prop:robustness-inclusion}.}
    \label{fig:two-dichotomic-measurements}
\end{figure}

\begin{figure}[!t]
    \centering
    \includegraphics[width=0.4\linewidth]{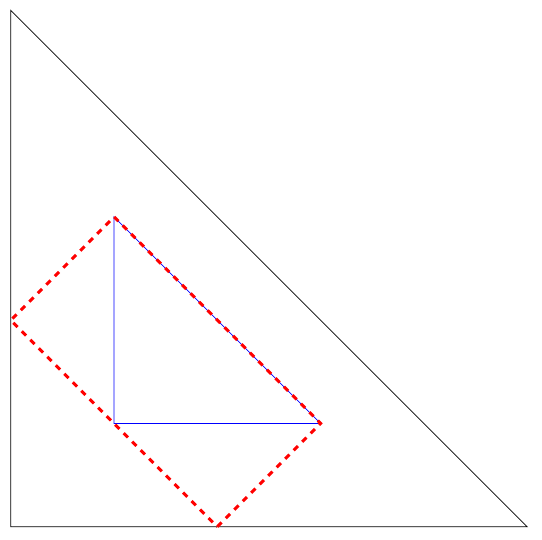}
    \caption{Any noisy measurement with noise parameter $t=2/5$ can be post-processed from a pair of dichotomic measurements. The state space of $S_3$ is the outer triangle (black), its noisy version is the inner triangle (blue), and the affine image of the polysimplex $CS^1_{2,2}$ is the red rectangle. The inclusions of the state spaces correspond to those from \cref{prop:robustness-inclusion}; note that here the state space dimensions are equal.}
    \label{fig:ternary-measurements}
\end{figure}

\subsection{A noisy three-outcome measurement simulated by two dichotomic measurements}

We ask now the reverse question from the previous section: when can one noisy $3$-outcome measurement can be simulated by two dichotomic measurements? 

One can think that this is always possible: given a three outcome measurement $(A,B,C)$, simply define 
\begin{align}
    X_1 = A \qquad &X_2 = B+C\\
    Y_1 = B \qquad &Y_2 = A+C.
\end{align}
However, the original three outcome measurement cannot be post-processed from $X,Y$ because of the ambiguity in the case where the two measurements yield the first result. 

\begin{prop}
    For all $t \leq 2/5$, the identity map 
    \begin{equation}
        \operatorname{id} : S_{3; t} \to S_3
    \end{equation}
    factorizes through $CS^1_{2,2}$. In particular, any noisy ternary measurement with noise parameter $t \leq 2/5$ can be post-processed from a pair of dichotomic measurements.
\end{prop}
\begin{proof}
The construction is depicted in \Cref{fig:ternary-measurements}. More precisely, the noisy measurement
\begin{align}
    \tilde A = \frac 2 5 A + \frac 1 5 \one_K\\
    \tilde B = \frac 2 5 B + \frac 1 5 \one_K\\
    \tilde C = \frac 2 5 C + \frac 1 5 \one_K
\end{align}
can be post-processed from the two dichotomic measurements $(X_1, \one_K - X_1)$, $(Y_1, \one_K - Y_1)$ with 
\begin{align}
    X_1 &= \frac 5 2 \tilde A + \frac 5 2 \tilde B - \one_K\\ 
    Y_1 &= -\frac 5 4 \tilde A + \frac 5 4 \tilde B + \frac 1 2 \one_K.
\end{align}
One can easily check that $X_1$ and $Y_1$ above define valid measurements.
\end{proof}

\begin{conjecture}
The value of $t$ in the result above is optimal. 
\end{conjecture}

\section*{Acknowledgments}
TA acknowledges support from the Deutsche Forschungsgemeinschaft (DFG, German Research Foundation, project numbers 447948357 and 440958198), the Sino-German Center for Research Promotion (Project M-0294), the German Ministry of Education and Research (Project QuKuK, BMBF Grant No. 16KIS1618K), the DAAD, and the Alexander von Humboldt Foundation. TA and LL are supported by the Business Finland project BEQAH (Between
Quantum Algorithms and Hardware). AB was supported by the ANR project PraQPV, grant number ANR-24-CE47-3023. IN was supported by the ANR project \href{https://esquisses.math.cnrs.fr/}{ESQuisses}, grant number ANR-20-CE47-0014-01. MP acknowledges support from the Niedersächsisches Ministerium für Wissenschaft und Kultur. 

\bibliographystyle{alpha}
\bibliography{main.bbl}

\newcommand{\etalchar}[1]{$^{#1}$}
\begin{thebibliography}{GBTCA17}

\bibitem[ALPP21]{aubrun2021entangleability}
Guillaume Aubrun, Ludovico Lami, Carlos Palazuelos, and Martin Pl{\'a}vala.
\newblock Entangleability of cones.
\newblock {\em Geometric and Functional Analysis}, 31(2):181--205, 2021.

\bibitem[AMHP24]{monogamy}
Guillaume Aubrun, Alexander Müller-Hermes, and Martin Plávala.
\newblock Monogamy of entanglement between cones.
\newblock {\em Mathematische Annalen}, 2024.

\bibitem[BBB{\etalchar{+}}10]{Barnum_2010}
Howard Barnum, Salman Beigi, Sergio Boixo, Matthew~B. Elliott, and Stephanie
  Wehner.
\newblock Local quantum measurement and no-signaling imply quantum
  correlations.
\newblock {\em Physical Review Letters}, 104(14), 2010.

\bibitem[BBLW07]{barnum2007generalized}
Howard Barnum, Jonathan Barrett, Matthew Leifer, and Alexander Wilce.
\newblock Generalized no-broadcasting theorem.
\newblock {\em Physical Review Letters}, 99(24):240501, 2007.

\bibitem[BCF{\etalchar{+}}96]{barnum1996noncommuting}
Howard Barnum, Carlton~M. Caves, Christopher~A. Fuchs, Richard Jozsa, and
  Benjamin Schumacher.
\newblock Noncommuting mixed states cannot be broadcast.
\newblock {\em Physical Review Letters}, 76(15):2818, 1996.

\bibitem[BCG{\etalchar{+}}22]{budroni2022kochen}
Costantino Budroni, Ad{\'a}n Cabello, Otfried G{\"u}hne, Matthias Kleinmann,
  and Jan-{\AA}ke Larsson.
\newblock {Kochen-Specker} contextuality.
\newblock {\em Reviews of Modern Physics}, 94(4):045007, 2022.

\bibitem[BCP{\etalchar{+}}14]{Brunner2014}
Nicolas {Brunner}, Daniel {Cavalcanti}, Stefano {Pironio}, Valerio {Scarani},
  and Stephanie {Wehner}.
\newblock {Bell nonlocality}.
\newblock {\em Reviews of Modern Physics}, 86:419--478, 2014.

\bibitem[BGW13]{barnum2013ensemble}
Howard Barnum, Carl~Philipp Gaebler, and Alexander Wilce.
\newblock Ensemble steering, weak self-duality, and the structure of
  probabilistic theories.
\newblock {\em Foundations of Physics}, 43:1411--1427, 2013.

\bibitem[BJN22]{bluhm2022incompatibility}
Andreas Bluhm, Anna Jen{\v{c}}ov{\'a}, and Ion Nechita.
\newblock Incompatibility in general probabilistic theories, generalized
  spectrahedra, and tensor norms.
\newblock {\em Communications in Mathematical Physics}, 393(3):1125--1198,
  2022.

\bibitem[BLN25]{bluhm2025simulation}
Andreas Bluhm, Leevi Lepp{\"a}j{\"a}rvi, and Ion Nechita.
\newblock On the simulation of quantum multimeters.
\newblock {\em Quantum}, 9:1608, 2025.

\bibitem[BV18]{Bene2018}
Erika Bene and Tam\'as V\'ertesi.
\newblock {Measurement incompatibility does not give rise to Bell violation in
  general}.
\newblock {\em New Journal of Physics}, 20(1):013021, 2018.

\bibitem[CS14]{chruscinski2014entanglement}
Dariusz Chru{\'s}ci{\'n}ski and Gniewomir Sarbicki.
\newblock Entanglement witnesses: construction, analysis and classification.
\newblock {\em Journal of Physics A: Mathematical and Theoretical},
  47(48):483001, 2014.

\bibitem[DPS04]{doherty_complete_2004}
Andrew~C. Doherty, Pablo~A. Parrilo, and Federico~M. Spedalieri.
\newblock Complete family of separability criteria.
\newblock {\em Physical Review A}, 69(2):022308, 2004.

\bibitem[FHL18]{filippov2018simulability}
Sergey~N. Filippov, Teiko Heinosaari, and Leevi Lepp{\"a}j{\"a}rvi.
\newblock Simulability of observables in general probabilistic theories.
\newblock {\em Physical Review A}, 97(6):062102, 2018.

\bibitem[Fin82]{Fine2}
Arthur Fine.
\newblock Hidden variables, joint probability, and the {B}ell inequalities.
\newblock {\em Physical Review Letters}, 48:291--295, 1982.

\bibitem[GBTCA17]{guerini2017operational}
Leonardo Guerini, Jessica Bavaresco, Marcelo Terra~Cunha, and Antonio
  Ac{\'\i}n.
\newblock Operational framework for quantum measurement simulability.
\newblock {\em Journal of Mathematical Physics}, 58:092102, 2017.

\bibitem[GHK{\etalchar{+}}23]{guhne2023colloquium}
Otfried G{\"u}hne, Erkka Haapasalo, Tristan Kraft, Juha-Pekka
  Pellonp{\"a}{\"a}, and Roope Uola.
\newblock Colloquium: Incompatible measurements in quantum information science.
\newblock {\em Reviews of Modern Physics}, 95(1):011003, 2023.

\bibitem[Gis89]{steering_g}
Nicolas Gisin.
\newblock Stochastic quantum dynamics and relativity.
\newblock {\em Helvetica Physica Acta}, 62:363--371, 1989.

\bibitem[GT09]{guhne2009entanglement}
Otfried G{\"u}hne and G{\'e}za T{\'o}th.
\newblock Entanglement detection.
\newblock {\em Physics Reports}, 474(1-6):1--75, 2009.

\bibitem[Hei27]{heisenberg1927anschaulichen}
Werner Heisenberg.
\newblock {\"U}ber den anschaulichen {Inhalt} der quantentheoretischen
  {Kinematik} und {Mechanik}.
\newblock {\em Zeitschrift f{\"u}r Physik}, 43(3):172--198, 1927.

\bibitem[HHH96]{horodecki_ppt}
Michał Horodecki, Paweł Horodecki, and Ryszard Horodecki.
\newblock Separability of mixed states: necessary and sufficient conditions.
\newblock {\em Physics Letters A}, 223(1):1--8, 1996.

\bibitem[HJW93]{steering_lrw}
Lane~P. Hughston, Richard Jozsa, and William~K. Wootters.
\newblock A complete classification of quantum ensembles having a given density
  matrix.
\newblock {\em Physics Letters A}, 183(1):14--18, 1993.

\bibitem[HLP19]{Heinosaari2019nofreeinformation}
Teiko Heinosaari, Leevi Lepp{\"{a}}j{\"{a}}rvi, and Martin Pl{\'{a}}vala.
\newblock No-free-information principle in general probabilistic theories.
\newblock {\em {Quantum}}, 3:157, 2019.

\bibitem[HMZ16]{Heinosaari2016}
Teiko Heinosaari, Takayuki Miyadera, and M\'ario Ziman.
\newblock An invitation to quantum incompatibility.
\newblock {\em Journal of Physics A: Mathematical and Theoretical},
  49(12):123001, 2016.

\bibitem[HQB18]{noJM-noB}
Flavien Hirsch, Marco~Túlio Quintino, and Nicolas Brunner.
\newblock Quantum measurement incompatibility does not imply {Bell}
  nonlocality.
\newblock {\em Physical Review A}, 97(1), 2018.

\bibitem[HWVE14]{howard2014contextuality}
Mark Howard, Joel Wallman, Victor Veitch, and Joseph Emerson.
\newblock Contextuality supplies the ‘magic’ for quantum computation.
\newblock {\em Nature}, 510(7505):351--355, 2014.

\bibitem[ISD{\etalchar{+}}22]{IoannouSimulability2022}
Marie Ioannou, Pavel Sekatski, S\'ebastien Designolle, Benjamin D.~M. Jones,
  Roope Uola, and Nicolas Brunner.
\newblock Simulability of high-dimensional quantum measurements.
\newblock {\em Physical Review Letters}, 129:190401, 2022.

\bibitem[Jen18]{jencova2018incompatible}
Anna Jen{\v{c}}ov{\'a}.
\newblock Incompatible measurements in a class of general probabilistic
  theories.
\newblock {\em Physical Review A}, 98(1):012133, 2018.

\bibitem[JEPU24]{Jokinen2024}
Pauli Jokinen, Sophie Egelhaaf, Juha-Pekka Pellonpää, and Roope Uola.
\newblock Compressing continuous variable quantum measurements.
\newblock {\em Journal of Physics A: Mathematical and Theoretical}, 57:325302,
  2024.

\bibitem[JUC{\etalchar{+}}23]{Jones2023}
Benjamin D.~M. Jones, Roope Uola, Thomas Cope, Marie Ioannou, S\'ebastien
  Designolle, Pavel Sekatski, and Nicolas Brunner.
\newblock Equivalence between simulability of high-dimensional measurements and
  high-dimensional steering.
\newblock {\em Physical Review A}, 107:052425, 2023.

\bibitem[JWP{\etalchar{+}}24]{jokinen2024no}
Pauli Jokinen, Mirjam Weilenmann, Martin Pl{\'a}vala, Juha-Pekka
  Pellonp{\"a}{\"a}, Jukka Kiukas, and Roope Uola.
\newblock No-broadcasting characterizes operational contextuality.
\newblock {\em Physical Review Letters}, 133(24):240201, 2024.

\bibitem[Ken27]{kennard1927quantenmechanik}
Earle~H Kennard.
\newblock Zur {Quantenmechanik} einfacher {Bewegungstypen}.
\newblock {\em Zeitschrift f{\"u}r Physik}, 44(4):326--352, 1927.

\bibitem[KNI10]{KIMURA2010175}
Gen Kimura, Koji Nuida, and Hideki Imai.
\newblock Distinguishability measures and entropies for general probabilistic
  theories.
\newblock {\em Reports on Mathematical Physics}, 66(2):175--206, 2010.

\bibitem[KS90]{kochen1990problem}
Simon Kochen and Ernst~P. Specker.
\newblock The problem of hidden variables in quantum mechanics.
\newblock {\em Ernst Specker Selecta}, pages 235--263, 1990.

\bibitem[Kur15]{Kuramochi2015}
Yui Kuramochi.
\newblock Minimal sufficient positive-operator valued measure on a separable
  {Hilbert} space.
\newblock {\em Journal of Mathematical Physics}, 56:102205, 2015.

\bibitem[Lam18]{lami2018non}
Ludovico Lami.
\newblock Non-classical correlations in quantum mechanics and beyond.
\newblock {\em PhD thesis. arXiv preprint arXiv:1803.02902}, 2018.

\bibitem[Lep21]{leppajarvi2021measurement}
Leevi Lepp{\"a}j{\"a}rvi.
\newblock Measurement simulability and incompatibility in quantum theory and
  other operational theories.
\newblock {\em PhD thesis. arXiv preprint arXiv:2106.03588}, 2021.

\bibitem[MG23]{muller2023testing}
Markus~P. M{\"u}ller and Andrew J.~P. Garner.
\newblock Testing quantum theory by generalizing noncontextuality.
\newblock {\em Physical Review X}, 13(4):041001, 2023.

\bibitem[M{\"u}l21]{muller2021probabilistic}
Markus M{\"u}ller.
\newblock Probabilistic theories and reconstructions of quantum theory.
\newblock {\em SciPost Physics Lecture Notes}, page 028, 2021.

\bibitem[NB10]{book_narici}
L.~Narici and E.~Beckenstein.
\newblock {\em Topological {V}ector {S}paces}.
\newblock CRC Press, 2010.

\bibitem[NDN{\etalchar{+}}22]{nadlinger2022experimental}
David~P. Nadlinger, Peter Drmota, Bethan~C. Nichol, Gabriel Araneda, Dougal
  Main, Raghavendra Srinivas, David~M. Lucas, Christopher~J. Ballance, Kirill
  Ivanov, Ernest Y.-Z. Tan, Pavel Sekatski, Rüdiger~L. Urbanke, Renato Renner,
  Nicolas Sangouard, and Jean-Daniel Bancal.
\newblock Experimental quantum key distribution certified by {Bell's} theorem.
\newblock {\em Nature}, 607(7920):682--686, 2022.

\bibitem[OGWA17]{oszmaniec2017simulating}
Micha{\l} Oszmaniec, Leonardo Guerini, Peter Wittek, and Antonio Ac{\'\i}n.
\newblock Simulating positive-operator-valued measures with projective
  measurements.
\newblock {\em Physical Review Letters}, 119:190501, 2017.

\bibitem[OMP19]{Oszmaniec2019}
Micha\l{} Oszmaniec, Filip~B. Maciejewski, and Zbigniew Pucha\l{}a.
\newblock Simulating all quantum measurements using only projective
  measurements and postselection.
\newblock {\em Physical Review A}, 100:012351, 2019.

\bibitem[PG24]{plavala2024contextuality}
Martin Pl{\'a}vala and Otfried G{\"u}hne.
\newblock Contextuality as a precondition for quantum entanglement.
\newblock {\em Physical Review Letters}, 132(10):100201, 2024.

\bibitem[PGQ24]{plavala2024all}
Martin Pl{\'a}vala, Otfried G{\"u}hne, and Marco~T{\'u}lio Quintino.
\newblock All incompatible measurements on qubits lead to multiparticle {Bell}
  nonlocality.
\newblock {\em arXiv preprint arXiv:2403.10564}, 2024.

\bibitem[Pir14]{Pironio_2014}
Stefano Pironio.
\newblock All {Clauser–Horne–Shimony–Holt} polytopes.
\newblock {\em Journal of Physics A: Mathematical and Theoretical},
  47(42):424020, 2014.

\bibitem[Pis20]{book_pisier}
Gilles Pisier.
\newblock {\em Tensor {P}roducts of {C}*-{A}lgebras and {O}perator {S}paces:
  {T}he {C}onnes–{K}irchberg {P}roblem}.
\newblock Cambridge University Press, 2020.

\bibitem[Pl{\'a}17]{plavala2017conditions}
Martin Pl{\'a}vala.
\newblock Conditions for the compatibility of channels in general probabilistic
  theory and their connection to steering and {Bell} nonlocality.
\newblock {\em Physical Review A}, 96(5):052127, 2017.

\bibitem[Pl{\'a}22]{plavala2022incompatibility}
Martin Pl{\'a}vala.
\newblock Incompatibility in restricted operational theories: connecting
  contextuality and steering.
\newblock {\em Journal of Physics A: Mathematical and Theoretical},
  55(17):174001, 2022.

\bibitem[Pl{\'a}23]{plavala2023general}
Martin Pl{\'a}vala.
\newblock General probabilistic theories: {A}n introduction.
\newblock {\em Physics Reports}, 1033:1--64, 2023.

\bibitem[PLG25]{plavala2025polarization}
Martin Pl{\'a}vala, Laurens~T.\ Ligthart, and David Gross.
\newblock The polarization hierarchy for polynomial optimization over convex
  bodies, with applications to nonnegative matrix rank.
\newblock {\em Linear Algebra and its Applications}, 732:15--32, 2025.

\bibitem[PR94]{Popescu_Rohrlich_94}
Sandu Popescu and Daniel Rohrlich.
\newblock Quantum nonlocality as an axiom.
\newblock {\em Foundations of Physics}, 24(3):379--385, 1994.

\bibitem[Rob29]{robertson1929uncertainty}
Howard~Percy Robertson.
\newblock The uncertainty principle.
\newblock {\em Physical Review}, 34(1):163, 1929.

\bibitem[Roc70]{Rockafellar1970}
R.~Tyrrell Rockafellar.
\newblock {\em Convex Analysis}.
\newblock Number~28 in Princeton Mathematical Series. Princeton University
  Press, 1970.

\bibitem[Rya02]{book_raymond}
Raymond~A Ryan.
\newblock {\em Introduction to tensor products of {B}anach spaces}.
\newblock Springer, 2002.

\bibitem[SB14]{stevens2014steering}
Neil Stevens and Paul Busch.
\newblock Steering, incompatibility, and {Bell}-inequality violations in a
  class of probabilistic theories.
\newblock {\em Physical Review A}, 89(2):022123, 2014.

\bibitem[SBC{\etalchar{+}}15]{postQ_steering}
Ana~Bel\'en Sainz, Nicolas Brunner, Daniel Cavalcanti, Paul Skrzypczyk, and
  Tam\'as V\'ertesi.
\newblock Postquantum steering.
\newblock {\em Physical Review Letters}, 115:190403, 2015.

\bibitem[Sch35]{schrodinger1935gegenwartige}
Erwin Schr{\"o}dinger.
\newblock Die gegenw{\"a}rtige {S}ituation in der {Q}uantenmechanik.
\newblock {\em Die Naturwissenschaften}, 48:23, 1935.

\bibitem[Spe05]{spekkens2005contextuality}
Robert~W. Spekkens.
\newblock Contextuality for preparations, transformations, and unsharp
  measurements.
\newblock {\em Physical Review A—Atomic, Molecular, and Optical Physics},
  71(5):052108, 2005.

\bibitem[SSW{\etalchar{+}}21]{schmid2021characterization}
David Schmid, John~H. Selby, Elie Wolfe, Ravi Kunjwal, and Robert~W. Spekkens.
\newblock Characterization of noncontextuality in the framework of generalized
  probabilistic theories.
\newblock {\em PRX Quantum}, 2(1):010331, 2021.

\bibitem[TDS03]{Terhal2003symmetric}
Barbara~M. Terhal, Andrew~C. Doherty, and David Schwab.
\newblock Symmetric extensions of quantum states and local hidden variable
  theories.
\newblock {\em Phys. Rev. Lett.}, 90:157903, Apr 2003.

\bibitem[TPKLR22]{tavakoli2022bell}
Armin Tavakoli, Alejandro Pozas-Kerstjens, Ming-Xing Luo, and Marc-Olivier
  Renou.
\newblock Bell nonlocality in networks.
\newblock {\em Reports on Progress in Physics}, 85(5):056001, 2022.

\bibitem[TU20]{tavakoli2020measurement}
Armin Tavakoli and Roope Uola.
\newblock Measurement incompatibility and steering are necessary and sufficient
  for operational contextuality.
\newblock {\em Physical Review Research}, 2(1):013011, 2020.

\bibitem[UBGP15]{uola2015one}
Roope Uola, Costantino Budroni, Otfried G{\"u}hne, and Juha-Pekka
  Pellonp{\"a}{\"a}.
\newblock One-to-one mapping between steering and joint measurability problems.
\newblock {\em Physical Review Letters}, 115(23):230402, 2015.

\bibitem[UCNG20]{Uola2020}
Roope Uola, Ana C.~S. Costa, H.~Chau Nguyen, and Otfried G\"uhne.
\newblock Quantum steering.
\newblock {\em Reviews of Modern Physics}, 92:015001, 2020.

\bibitem[UMG14]{Uola2014}
Roope Uola, Tobias Moroder, and Otfried G\"uhne.
\newblock Joint measurability of generalized measurements implies classicality.
\newblock {\em Physical Review Letters}, 113:160403, 2014.

\bibitem[vDdB20]{Bruyn2020TensorPO}
Josse van Dobben~de Bruyn.
\newblock Tensor products of convex cones, part {II}: Closed cones in
  finite-dimensional spaces.
\newblock {\em arXiv preprint arXiv:2009.11843}, 2020.

\bibitem[VPLU24]{veeren2024semi}
Isadora Veeren, Martin Pl{\'a}vala, Leevi Lepp{\"a}j{\"a}rvi, and Roope Uola.
\newblock Semi-device-independent certification of the number of measurements.
\newblock {\em Physical Review A}, 109(6):062203, 2024.

\bibitem[WF23]{wright2023invertible}
Victoria~J. Wright and M{\'a}t{\'e} Farkas.
\newblock Invertible map between {Bell} nonlocal and contextuality scenarios.
\newblock {\em Physical Review Letters}, 131(22):220202, 2023.

\bibitem[Wol21]{wolf2021quantum}
Ramona Wolf.
\newblock {\em Quantum key distribution}, volume 988 of {\em Lecture notes in
  physics}.
\newblock Springer, 2021.

\bibitem[WPGF09]{Wolf_2009}
Michael~M. Wolf, David P\'erez-Garc\'ia, and Carlos Fern\'andez.
\newblock Measurements incompatible in quantum theory cannot be measured
  jointly in any other no-signaling theory.
\newblock {\em Physical Review Letters}, 103:230402, 2009.

\bibitem[WZ82]{wootters1982single}
William~K. Wootters and Wojciech~H. Zurek.
\newblock A single quantum cannot be cloned.
\newblock {\em Nature}, 299(5886):802--803, 1982.

\bibitem[ZvLR{\etalchar{+}}22]{zhang2022device}
Wei Zhang, Tim van Leent, Kai Redeker, Robert Garthoff, Ren{\'e} Schwonnek,
  Florian Fertig, Sebastian Eppelt, Wenjamin Rosenfeld, Valerio Scarani,
  Charles C.-W. Lim, and Harald Weinfurter.
\newblock A device-independent quantum key distribution system for distant
  users.
\newblock {\em Nature}, 607(7920):687--691, 2022.

\end{thebibliography}

\appendix

\addtocontents{toc}{\protect\setcounter{tocdepth}{1}}

\section{Cones, tensor products, and positive maps} 
\subsection{Cones and ordered vector spaces} \label{appendix:cones}

In \cite{plavala2023general} the notion and several basic results of cones were reviewed.
More information on this topic can be found in \cite{Bruyn2020TensorPO, book_narici} and more about tensor products is available in \cite{book_raymond, book_pisier}.
A central part of the theory of cones, included in the above references, is the equivalence of the existence of a cone in a vector space and the vector space having an order relation.
Therefore, it is reasonable to associate positivity with cones. However, this is what makes the definition of tensor products of cones non trivial and in fact ambiguous as we will discuss shortly.

For the upcoming definitions let $V, \, V_A$ and $V_B$ be finite dimensional real vector spaces. We start with the definition of a cone.

\begin{defi}[Cones of vector spaces]
    A subset $\cone \subset V$ of a vector space is called a \emph{cone} if for all $v \in \cone$ it holds that $\lambda v \in \cone$ for any $\lambda \in \Rnum_+$.
    
    A cone $\cone$ is: 
    \begin{itemize}
        \item convex if $\cone$ is a convex set
        \item closed if $\cone$ is closed in the Euclidean topology on $V$
        \item pointed if $\cone \cap -\cone = \{0\}$
        \item generating if $\operatorname{span}(\cone) = V$, i.e., if $\cone-\cone=\{w-v \,|\, w \in \cone, \, v \in \cone\} =V$
    \end{itemize}
    If a cone $\cone$ is convex, closed, pointed, and generating, it is called a proper cone.
\end{defi}

Throughout this whole work, to avoid pathologies all cones will be proper.
Furthermore $\cone$, $\acone$ and $\bcone$ will denote the corresponding cones of $V$, $V_A$ and $V_B$. An \emph{ordered vector space} is a tuple $(V, V^+)$ of a vector space and a proper cone.

The dual cone is a subset of the dual space $V^*$ of $V$ and defined by collecting all linear functionals that map elements of the cone to a positive number. In this way, the notion of positivity is preserved.

\begin{defi}[Dual cone]
\label{dual_cone}
The dual cone is a subset of all positive functionals on a vector space, i.e., $(V^*)^+ \coloneqq (V^+)^\ast = \{\epsilon \in V^* \, | \, \epsilon(v) \geq 0 \mspace{5mu} \forall v \in \cone\}$.
\end{defi}

An example of a cone is the set of $d \times d$ positive semi-definite complex matrices $\Pos$ (for $d<\infty$). That is, for $A \in \Pos$ it holds that $\langle x, \, A x \rangle \geq 0$ for all $x \in   \Cnum^d$. The underlying vector space for $\Pos$ is the real vector space of $d \times d$ self-adjoint (or Hermitian) complex matrices $\mathcal M(\Cnum)_d^{\mathrm{sa}}$. 

\begin{prop}
\label{pos_self_dual}
The set of positive semi-definite matrices is self dual up to isomorphism, i.e.,\ $\Pos \cong \Pos^*$.
\end{prop}

Up to some normalization, the set $\Pos$ can be thought of as the set of quantum states, that is, for every $A \in \Pos$ we have $A = \lambda \varrho$ for some density matrix $\varrho \in D(\Cnum^d)$. Density matrices are consequently the base of the cone $\Pos$ since any element of the cone is a nonnegative multiple of a density matrix.

\subsection{Tensor products of cones} \label{appendix:tensor-cones}

When considering the tensor product of ordered vector spaces, we need to preserve the notion of positivity.
Put differently, the result of a tensor product between cones should again give a cone. In general, there are infinitely many ways of constructing the tensor product of two cones, but there is a minimal and a maximal way to do this:

\begin{defi}[Minimal tensor product]
    Let $\acone \text{ and } \bcone$ be two cones.
    Their \emph{minimal tensor product} is then defined as $\acone \tmin \bcone = \operatorname{conv } \{x_A \otimes x_B \, | \, x_A \in \acone, \, x_B \in \bcone \}$.
\end{defi}
It can be verified that the minimal tensor product indeed gives a cone. Inspired by quantum mechanics, we call elements of the minimal tensor product of cones \emph{separable}, since the normalized elements of $\Posa \tmin \Posb$ are the separable states.

\begin{defi}[Maximal tensor product]
\label{tmax}
    The \emph{maximal tensor product} of two cones is the set
    $\acone \tmax \bcone = (\aconed \tmin \bconed)^*$.
\end{defi}
It can again be verified that the above definition gives a cone. Normalized elements of the cone $\Posa \tmax \Posb$ are entanglement witnesses. These definitions can immediately be generalized to any finite number of cones. We will recall next some basic but useful results about cones and their tensor products:
\begin{lem}
\label{cone_facts}
Let $(V,V^+)$, $(V_A,V_A^+)$, $(V_B,V_B^+)$ be ordered vector spaces. Then,
\begin{enumerate}
    \item $\cone \cong (V^{**})^+ \label{cone_isom_double}$
    \item Let $v \in V$ and $\epsilon(v) \geq 0 \mspace{5mu} \forall \epsilon \in (V^*)^+$. Then $v \in \cone \label{cone_isom_double_cor}$.
    \item The sets $\acone \tmin \bcone$ and $\acone \tmax \bcone$ are proper cones.
    \item $\acone \tmin \bcone \subseteq \acone \tmax \bcone \label{tmin_subset_tmax}$. Equality holds if and only if at least one of the cones is simplicial, i.e., isomorphic to $\mathds R_+^n$ for some $n \in \mathds N$ \cite{aubrun2021entangleability}.
    \item $(\acone \tmax \bcone)^* = \aconed \tmin \bconed$
\end{enumerate}
\end{lem}

 We will omit the isomorphism between $V$ and $V^{**}$ when we make use of the fact $V \cong V^{**}$ for real, finite-dimensional vector spaces. A tensor product of cones $V_A$ and $V_B$ is any cone $V_{AB}^+$ such that $\acone \tmin \bcone \subseteq V_{AB}^+\subseteq \acone \tmax \bcone$. Note that $\Posa \tmin \Posb \subsetneq \mathrm{PSD}_{d_A \cdot d_B} \subsetneq \Posa \tmax \Posb$.

\subsection{Positive maps} \label{appendix:positive-maps}
So far, we have focused on the notion of positivity and emphasized the connection to cones of vector spaces.
We can now define what it means for a linear map to be  positive.

\begin{defi}[Positive maps]
\label{positive_maps}
    Let $L: V_A \rightarrow V_B$ be a linear map. Then, $L$ is positive if $L(\acone) \subset \bcone$.
\end{defi}

We can in fact connect tensor products between cones and positive maps. For this, we will use that one can identify a linear map $L: V_A \rightarrow V_B$ with a tensor $\xi_L \in V_A^* \otimes V_B$. As pointed out, e.g., in \cite{book_raymond}, the one-to-one correspondence is given via
\begin{equation}
   \psi_B(L(v_A)) = (v_A \otimes \psi_B)(\xi_L) \text{, where } v_A \in V_A, \, \psi_B \in V_B^*.
\end{equation}

\begin{lem}
\label{tmax_map}
    Let $\acone$ and $\bcone$ be two cones. Then $\xi_L \in \aconed \tmax \bcone$ if and only if $L: V_A \rightarrow V_B$ is positive in the sense of \Cref{positive_maps}.
\end{lem}

The next proposition shows again that the maximal tensor product is in general strictly bigger, which means that the conditions on its elements are less strong.
Here in the case of positive maps this is demonstrated by not having an \textit{if and only if} statement in contrast to the above lemma.

\begin{prop}
    $L \in \aconed \tmin \bcone \overset{\cancel \Leftarrow}{\Rightarrow} L(\acone) \subset \bcone$
\end{prop}

Quantum theory already provides us with a counterexample for the direction from right to left, that is, set $\aconed = \Posa^*$ and $\bcone = \Posb$.
Separable maps can be thought of as being entanglement breaking and not every positive map is entanglement breaking, but every entanglement breaking map is positive.

It is crucial for this work to understand that inseparable operators can be distinguished from separable operators by functionals (``witnesses"), which can map some of the former to negative numbers and always map the latter to positive numbers.
This is formalized by the next lemma.
\begin{lem}
\label{witness_exists}
For $v \in \acone \tmax \bcone \text{, such that } v \notin \acone \tmin \bcone $, there is a witness type functional $ W \in \aconed \tmax \bconed$ with $W(v) < 0$.
\end{lem}

The proof can be done with the help of the strict hyperplane separation theorem, using that we only consider proper cones which are in particular closed.

Finally, we extend a positive, linear map to a map between tensor products while not breaking the positivity.

\begin{lem}
\label{extend_map}
    Let $\acone, \, \bcone$ as well as $V_C^+$ be cones and $\Phi: \acone \rightarrow \bcone$ be a linear map. Then $\Phi$ induces the maps:
    \begin{itemize}
        \item[(i)] $\Phi \otimes \id: \acone \tmin V_C^+ \rightarrow \bcone \tmin V_C^+$
        \item[(ii)] $\Phi \otimes \id: \acone \tmax V_C^+ \rightarrow \bcone \tmax V_C^+$
    \end{itemize}
\end{lem}

Note that the induced maps are still positive because a cone gets mapped to a cone.
Creating a tensor product between cones and hence adding an identity to the map $\Phi$ preserves the notion of positivity.

\begin{remark}
At this point one needs to be careful when talking about positive maps.
A map from \Cref{extend_map} might be positive in the sense that it maps a cone to a cone, though this is not the same as complete positivity, which is a term also found in quantum theory.
If we let $\bcone$ and $V_C^+$ be sets of positive operators on a Hilbert space, then the operators in the maximal tensor product are not always positive in the sense that they map all bipartite states to positive numbers.
There might be entangled states that get mapped to something negative. An example is the partial transposition which is a positive, however not completely positive map. This fact is the quintessence of the PPT-criterion for entanglement detection \cite{horodecki_ppt}.
\end{remark}

\end{document}